\documentclass[preprint,12pt]{elsarticle}
\usepackage{amssymb}
\usepackage{physics}
\usepackage{float}
\usepackage{amsmath}
\usepackage{amsthm}
\numberwithin{equation}{section}
\usepackage{caption}
\usepackage{subcaption}
\usepackage[colorlinks,linkcolor=red,anchorcolor=blue,citecolor=green]{hyperref}
\usepackage{geometry}
\usepackage{mathrsfs}
\newcommand*\Laplace{\mathop{}\!\mathbin\bigtriangleup}

\newtheorem{lemma}{Lemma}

\journal{Elsevier}

\begin{document}

\begin{frontmatter}

\title{UGKS and UGKWP Methods for Multiscale Simulation of Electrostatic Plasma in Quasineutral and Hydrodynamic Limits}

\author[a]{Zhigang PU}
\ead{zpuac@connect.ust.hk}
\author[a,b]{Kun XU\corref{cor1}}
\ead{makxu@ust.hk}
\cortext[cor1]{Cooresponding author}

\affiliation[a]{organization={Department of Mathematics, Hong Kong University of Science and Technology},
            addressline={Clear Water Bay, Kowloon},
            city={Hong Kong},
            country={China}}
\affiliation[b]{organization={Shenzhen Research Institute, Hong Kong University of Science and Technology},
            city={Shenzhen},
            country={China}}

\begin{abstract}

This study extends the Unified Gas-Kinetic Scheme (UGKS) and the Unified Gas-Kinetic Wave-Particle (UGKWP) method for electrostatic plasma modeling, ensuring the correct asymptotic limits with respect to both the Debye length and the mean free path. By coupling collision and transport processes within the numerical flux, the proposed approach effectively removes the hydrodynamic-limit constraint associated with the mean free path. In addition, a reformulated Poisson equation, coupled with the macroscopic moment equations, is introduced to overcome the inefficiency of the standard Poisson formulation in the quasineutral regime.
The accuracy and asymptotic consistency of the proposed schemes are verified through several benchmark tests, including linear and nonlinear Landau damping and the bump-on-tail instability. The results demonstrate that the methods robustly capture plasma dynamics across hydrodynamic and quasineutral regimes, without resolution constraints imposed by either the Debye length or the mean free path.

\end{abstract}


\end{frontmatter}


\section{Introduction}
\label{introduction}

Plasma physics is fundamental to the development of advanced technologies such as hypersonic flight, nuclear fusion, and semiconductor processing~\cite{chen1984introduction}. A major challenge in its computational modeling lies in the inherently multiscale nature of plasma behavior, which is governed by several key dimensionless parameters. The Knudsen number, defined as $Kn = l_{mfp} / L$, where $l_{mfp}$ denotes the mean free path and $L$ is the characteristic flow length, quantifies the degree of deviation of a flow from local thermodynamic equilibrium. A typical example arises in hypersonic flows, where a dense, continuum shock region coexists with a highly rarefied wake. Another important scale is the Debye length, $\lambda = \sqrt{\epsilon_0 k_B T_e / n_e e^2}$, where $\epsilon_0$ is the permittivity of free space, $k_B$ is the Boltzmann constant, $T_e$ is the electron temperature, $n_e$ is the electron number density, and $e$ is the elementary charge. In sheath problems—common in low-temperature plasmas and hypersonic applications—$\lambda$ is comparable to the sheath thickness inside the sheath region but becomes negligibly small relative to the system scale outside it. Furthermore, the Larmor radius introduces an additional magnetohydrodynamic scale. For instance, in magnetic reconnection, the Larmor radius is much smaller than the system scale outside the current sheet but comparable to the sheet’s width inside it.
Therefore, the development of efficient computational frameworks capable of seamlessly bridging these disparate physical scales—without being restricted by parameters such as the Knudsen number, Debye length, or Larmor radius—is essential for accurate and robust plasma modeling.

This study aims to develop numerical schemes for electrostatic plasma that are asymptotic-preserving (AP) with respect to both the mean free path and the Debye length.
We begin with a concise review of the relevant literature. Research on collisionless electrostatic plasma has established the theoretical foundation 
for AP methods concerning the Debye length. Early work by Mason \cite{mason1981implicit} and later by Degond \cite{degond2010asymptotic} 
introduced the key concept of analytically coupling the Poisson equation with the fluid moment equations, 
under the assumption that the macroscopic numerical flux derived from the moment equations remains consistent 
with the kinetic description within a single time step. 
In this formulation, the Poisson equation is reconstructed by integrating macroscopic moment quantities, yielding what is commonly known 
as the Reformulated Poisson Equation (RPE). 
The RPE approach was first implemented within Particle-in-Cell (PIC) methods and subsequently extended to deterministic solvers 
by Crouseilles et al. \cite{belaouar2009asymptotically,crouseilles2016multiscale}. More recently, Liu et al.\cite{liu2020discrete} 
incorporated this methodology into the Discrete Unified Gas-Kinetic Scheme (DUGKS), thereby simultaneously addressing the numerical stiffness arising in the small mean-free-path regime.

Significant progress has been achieved in developing asymptotic-preserving (AP) schemes for handling the mean free path in rarefied gas dynamics. A major milestone in this area is the Unified Gas-Kinetic Scheme (UGKS) proposed by Xu~\cite{xu2010unified}, which enables seamless simulations across the rarefied–continuum flow regimes. Its core innovation lies in the coupled collision–transport numerical flux, derived from the integral solution of the Bhatnagar–Gross–Krook (BGK) equation at cell interfaces, ensuring consistent macroscopic flux evaluation.
The versatility of UGKS has been widely demonstrated through its successful extensions to various physical domains, including radiation transfer\cite{sun2020multiscale}, granular flows\cite{liu2019unified}, and plasma dynamics\cite{liu_unified_2016}. Recently, Guo et al.\cite{guo2023unified} further established that UGKS possesses a broader property—termed Unified Preserving (UP)—which guarantees accurate recovery not only of the Euler equations but also of the Navier–Stokes limit. This UP property provides a substantial advantage over conventional AP schemes.
Building on the UGKS framework, the Unified Gas-Kinetic Wave–Particle (UGKWP) method~\cite{liu2021unified} was developed to improve computational efficiency in problems requiring large velocity-space resolution, such as hypersonic flows\cite{long2024nonequilibrium, CAO2026106896}. The UGKWP method has also demonstrated remarkable flexibility across a wide range of transport phenomena, including radiative transfer, phonon transport, plasma modeling, photon transport, and turbulent flows\cite{pu2025unified, hz9s-5qbm, yang2025wave}.
In addition, other effective multiscale methods—such as the General Synthetic Iterative Scheme (GSIS)\cite{su2020fast} and the Unified Stochastic Particle (USP) method\cite{fei2020unified}—have also exhibited strong potential for rarefied gas simulation and related multiscale transport problems.

In this work, we integrate the UGKWP method with the RPE approach to model electrostatic plasmas, and validate the proposed scheme against standard benchmark problems. This represents the first stochastic particle method capable of accurately simulating plasma dynamics without numerical constraints imposed by either the Debye length or the mean free path.
To establish a robust foundation, a UGKS–RPE scheme is first developed, upon which the UGKWP–RPE method is subsequently constructed for enhanced efficiency and multiscale capability. The integration of the RPE formulation endows the UGKWP method with the ability to effectively handle vanishingly small Debye lengths, a particularly valuable feature given UGKWP’s proven efficiency in hypersonic flow simulations.
Overall, the extended UGKWP–RPE framework offers a powerful and efficient tool for modeling multiscale electrostatic plasma flows, with great potential for application to hypersonic plasma environments.

The paper is organized as follows. Section \ref{sec:model} introduces the Vlasov-Poisson-BGK model and its reformulated version. Section \ref{sec:numerical} details the proposed numerical schemes. Section \ref{sec:numericalstudies} presents numerical validations. Finally, Section \ref{sec:conclusion} provides concluding remarks.

\section{Kinetic Model Equations}
\label{sec:model}

\subsection{Vlasov-Poisson-BGK model}
\label{sec:vlasovmodel}

For an electrostatic collisional plasma with a fixed ion background, the kinetic equation describing electron dynamics can be expressed as:

$$
\frac{\partial f}{\partial t} + \boldsymbol{u} \cdot \nabla_{\boldsymbol{x}} f + \frac{q\boldsymbol{E}}{m} \cdot \nabla_{\boldsymbol{u}} f = \frac{g - f}{\tau},
$$
where $f(t, \boldsymbol{x}, \boldsymbol{u})$ represents the electron distribution function at time $t$, spatial position $\boldsymbol{x}$, and microscopic velocity $\boldsymbol{u}$. The electron mass is denoted by $m$, and $g$ is a Maxwellian distribution.
$$
g = \frac{\rho}{(2\pi RT)^{D/2}} e^{-\frac{(\boldsymbol{u}-\boldsymbol{U})^2}{2RT}},
$$
where $\rho$ is the mass density, $R=k_B/m$ is the gas constant, $k_B$ is the Boltzmann constant, $T$ is the temperature, $D$ is the dimension, and $\boldsymbol{U}$ is the macroscopic velocity.
The electric field $\boldsymbol{E}$ is derived from the electric potential $\phi$ via $\boldsymbol{E} = -\nabla_{\boldsymbol{x}} \phi$, where $\phi$ is determined by the Poisson equation:

$$
\nabla_{\boldsymbol{x}}^2 \phi = -\frac{\rho_c}{\epsilon_0},
$$
where $\epsilon_0$ is the vacuum permittivity. The charge density $\rho_c$ is given by $\rho_c = q(n_0 - n_e)$, where $q = -e$ is an electron's charge, $n_0$ is number density of the background ion and $n_e = \frac{1}{m}\int f(\boldsymbol{x}, \boldsymbol{u}, t) d^3u$ is the electron number density.

Next, we derive the dimensionless system. First, define reference quantities as follows:
\begin{align*}
u_0=\sqrt{\frac{k_BT_0}{m_0}}, f_0=\frac{m_0n_0}{u_0^3}, \phi_0=\frac{k_B T_0}{e}, q_0=e,
\end{align*}
where $u_0$ is reference velocity, $T_0$ is reference temperature, $m_0$ is the reference mass, $f_0$ is reference distribution function, $n_0$ is reference number density, $\phi_0$ is reference electrical potential, $q_0$ is reference charge. Then the dimensionless quantities are given as
\begin{align*}
\hat{\lambda}=\lambda/l_0,\quad \lambda=\sqrt{\frac{\epsilon_0 k_B T_0}{e^2 n_0}}=\sqrt{\frac{\epsilon_0 m_0}{n_0 e^2}} \cdot \sqrt{\frac{k_B T_0}{m_e}}=\frac{u_{th}}{\omega_{p e}},\quad \omega_{pe}=\sqrt{\frac{n_{e}e^2}{m_e \epsilon_0}},
\end{align*}
where $l_0$ is the reference length, $\hat{\lambda}$ is the dimensionless Debye length, $\omega_{pe}$ is plasma frequency.
Then we get the dimensionless Vlasov-Poisson-BGK model, denoted as $\mathbf{P}^{\lambda,\tau}$:
\begin{align}
&\frac{\partial \hat{f}}{\partial \hat{t}} + \boldsymbol{\hat{u}}\cdot \nabla_{\hat{x}} \hat{f}+{\nabla_{\hat{x}}} \hat{\phi} \nabla_{\hat{u}} \hat{f}=\frac{\hat{g}^M - \hat{f}}{\hat{\tau}},\label{eq:electron_kinetic}\\
&\hat{\lambda}^2 \nabla_{\hat{x}}^2 \hat{\phi}=\hat{n}_e - 1\label{eq:poisson},
\end{align}
For simplicity, the hat symbol is omitted from the subsequent discussion.

This system exhibits two distinct multiscale behaviors. The first concerns the mean collision time, $\tau$. When $\tau \gg \Delta t$, the collision source term is non-stiff, placing the system in the rarefied regime. Conversely, when $\tau \ll \Delta t$, the collision source term becomes very stiff, and the system resides in the near-continuum regime. To model these two regimes continuously without requiring a timestep strictly smaller than the mean collision time, the UGKS and UGKWP are adopted in this work \cite{xu2010unified,guo2023unified}.

The second multiscale behavior arises from the Debye length, $\lambda$. In the quasineutral regime, where the mesh size $\Delta x$ is significantly larger than the Debye length ($\Delta x \gg \lambda$), charge separation effects decay at a rate substantially faster than the characteristic length scale of the mesh. Directly resolving the Debye length in this regime is computationally infeasible. However, a direct substitution of $\lambda = 0$ into the Poisson equation \eqref{eq:poisson} results in a singular form, reducing it to the quasineutrality constraint ${n}_i = {n}_e$ and rendering the electrostatic potential $\phi$ indeterminate. To circumvent this limitation, the Poisson equation is reformulated to capture large-scale electrostatic effects. This work employs a reformulated Poisson equation to determine $\phi$ within the quasineutral regime \cite{degond2010asymptotic}.

\subsection{Reformulated Vlasov-Poisson-BGK model}
\label{sec:rpemodel}

By taking the zeroth and first moments of electron kinetic equations \eqref{eq:electron_kinetic}, we obtain the density and momentum equations:
\begin{align}
& \frac{\partial n}{\partial t} + \nabla \cdot (n \boldsymbol{U}) = 0, \label{eq:electron_continuity} \\
& \frac{\partial }{\partial t}(n \boldsymbol{U}) + \nabla \cdot \mathbb{S} = n \nabla \phi, \label{eq:electron_momentum}
\end{align}
where the momentum densities $(n \boldsymbol{U})$ and momentum flux tensors $\mathbb{S}$ are defined as
\begin{equation*}
n \boldsymbol{U} = \frac{1}{m}\int \boldsymbol{u} f \, d\boldsymbol{u}, \quad \mathbb{S} = \frac{1}{m} \int \boldsymbol{u} \otimes \boldsymbol{u} f \, d\boldsymbol{u},
\end{equation*}
with $\otimes$ denoting the tensor product.
Differentiating \eqref{eq:electron_continuity} with respect to time and taking the divergence of \eqref{eq:electron_momentum}, we combine the results to eliminate $\nabla\cdot(n \boldsymbol{U})$:
\begin{equation}
-\partial_t^2 n + \nabla^2 : \mathbb{S} = \nabla \cdot \left(n\nabla \phi\right), \label{eq:rpe_pre}
\end{equation}
where $\nabla^2$ denotes the Hessian tensor and $:$ represents the tensor contracted product. Substituting the original Poisson equation \eqref{eq:poisson} into \eqref{eq:rpe_pre} produces
\begin{equation}
-\lambda^2 \partial_t^2 \Delta \phi + \nabla^2 : \mathbb{S} = \nabla \cdot \left( n\nabla \phi\right).
\label{eq:rpe0}
\end{equation}
Moving $\phi$ to the left-hand side yields the reformulated Poisson equation:
\begin{equation}
\nabla \cdot \left[\left(n + \lambda^2 \partial_t^2\right) \nabla \phi\right] = \nabla^2 : \mathbb{S}. \label{eq:rpe}
\end{equation}
Combined with \eqref{eq:electron_kinetic}, we get the reformulated Vlasov-Poisson-BGK model, denoted as $\mathbf{RP}^{\lambda,\tau}$,
\begin{align}
&\frac{\partial f}{\partial t} + \boldsymbol{u} \cdot \nabla_{\boldsymbol{x}} f - \nabla_{\boldsymbol{x}} \phi \cdot \nabla_{\boldsymbol{u}} f =  \frac{g-f}{\tau}, \label{eq:reform_vlasov} \\
&\nabla \cdot \left[\left(n + \lambda^2 \partial_t^2\right) \nabla \phi\right] = \nabla^2 : \mathbb{S}.\label{eq:rpe_poisson_final}
\end{align}

\begin{lemma}
The Vlasov-Poisson-BGK system $\mathbf{P}^{\lambda,\tau}$ is equivalent to the reformulated Vlasov-Poisson-BGK system $\mathbf{RP}^{\lambda,\tau}$ if and only if the initial potential $\phi_0$ and its time derivative $\phi^{'}_0$ satisfy the following initial conditions:
\begin{align*}
& \lambda^2 \Delta \phi_0 = (n - 1)_0  \\
& \lambda^2 \Delta \phi_0' = \left. \nabla \cdot  (n \boldsymbol{U}) \right|_{t=0}. \label{eq:rpe_ic2}
\end{align*}
\end{lemma}

\begin{proof}
On the one hand, following the derivation from Eq.\eqref{eq:electron_continuity} to \eqref{eq:rpe}, solutions $(f, \phi)$ to the $\mathbf{P}^{\lambda,\tau}$ system satisfy the reformulated system $\mathbf{RP}^{\lambda,\tau}$.
On the other hand,  solutions $(f, \phi)$ to the reformulated system $\mathbf{RP}^{\lambda,\tau}$ must also satisfy the intermediate relations \eqref{eq:rpe_pre} and \eqref{eq:rpe0}. Subtracting these two relations yields:
\begin{equation*}
\partial_t^2 \left(\lambda^2 \Delta \phi - n\right) = 0.
\end{equation*}
This implies that if the initial potential $\phi_0$ and its time derivative $\phi_0'$ satisfy the following initial conditions:
\begin{align*}
& \lambda^2 \Delta \phi_0 = (n - n_0)_0, \\
& \lambda^2 \Delta \phi_0' = \left. \nabla \cdot (n \boldsymbol{U}) \right|_{t=0},
\end{align*}
then the Poisson equation \eqref{eq:poisson} is satisfied for all $t \geq 0$. Therefore, the system $\mathbf{P}^{\lambda,\tau}$ and $\mathbf{RP}^{\lambda,\tau}$ is equivalent.
\end{proof}

\subsection{Asymptotic limit}
\label{sec:asymlimit}

\subsubsection{Quasineutral limit}
\label{sec:qnlimit}

Taking the limit of Debye length $\lambda \rightarrow 0$ of system $\mathbf{RP}^{\lambda,\tau}$ produces the well-defined quasineutral system $\mathbf{RP}^{0,\tau}$:
\begin{align}
&\frac{\partial f}{\partial t} + \boldsymbol{u} \cdot \nabla_{\boldsymbol{x}} f + \nabla_{\boldsymbol{x}} \phi \cdot \nabla_{\boldsymbol{u}} f =  \frac{g-f}{\tau},  \\
&\nabla \cdot \left( n \nabla \phi\right) = \nabla^2 : \mathbb{S}\label{eq:rpe_qn},
\end{align}
with initial constraints
\begin{align*}
&(n_i - n_e)|_{t=0} = 0, \label{eq:quasi_ic1} \\
&\left. \nabla \cdot \left[(n \boldsymbol{U})_i - (n \boldsymbol{U})_e\right] \right|_{t=0} = 0.
\end{align*}

\subsubsection{Hydrodynamical and quasineutral limit}
\label{sec:hqlimit}

This first order Chapman-Enskog expansion of the distribution function $f$ gives the two fluid Navier-Stokes equation \cite{xu2001},
\begin{equation*}
    \begin{aligned}
        & \frac{\partial \rho}{\partial t} + \nabla \cdot (\rho \boldsymbol{U}) = 0,  \\
        & \frac{\partial }{\partial t}(\rho \boldsymbol{U}) + \nabla \cdot (\rho\boldsymbol{U}\boldsymbol{U} + p\mathbb{I}-\mu\mathbb{\sigma}) = \nabla \phi, \\
        & \frac{\partial }{\partial t}(\rho \mathscr{E}) + \nabla \cdot (((\rho\mathscr{E}) + p)\boldsymbol{U} - \mu \mathbb{\sigma} + \kappa\nabla T) =   \nabla \phi\cdot\boldsymbol{U},  \\
        &\nabla \cdot \left[\left(n + \lambda^2 \partial_t^2\right) \nabla \phi\right] = \nabla^2 : \mathbb{S},
    \end{aligned}
\end{equation*}
where $\mu, \kappa$ are the viscosity and thermal coefficient. $\mathbb{\sigma}$ is the strain tensor.

\section{UGKS and UGKWP method}
\label{sec:numerical}

\subsection{General Framework}
\label{sec:framework}

In this subsection, we introduce the numerical scheme $\mathbf{RP}^{\lambda,\tau}_h$ to solve the system $\mathbf{RP}^{\lambda,\tau}$, denoted as UGKS-RPE and UGKWP-RPE, depending on the fluid solver UGKS or UGKWP. $h$ here represents the discretized mesh size. In comparison, numerical scheme $\mathbf{P}^{\lambda,\tau}_h$ to solve the system $\mathbf{P}^{\lambda,\tau}$ is denoted as UGKS-PE and UGKWP-PE. The scheme is based on the finite volume method (FVM), where the cell averaged macroscopic conservative variables $\boldsymbol{W}_i = (\rho_i, (\rho\boldsymbol{U})_i, (\rho\mathscr{E})_i)$ on a physical cell $\Omega_i$, microscopic distribution function $f_i$ and electrical potential $\phi_i$ are defined as
$$
\boldsymbol{W}_i = \frac{1}{|\Omega_i|}\int_{\Omega_i}\boldsymbol{W}(\boldsymbol{x})\mathrm{d}\boldsymbol{x},\quad f_i = \frac{1}{|\Omega_i|}\int_{\Omega_i}f(\boldsymbol{x})\mathrm{d}\boldsymbol{x},\quad \phi_i = \frac{1}{|\Omega_i|}\int_{\Omega_i}\phi(\boldsymbol{x})\mathrm{d}\boldsymbol{x},
$$
where $|\Omega_i|$ is the volume of cell $\Omega_i$. A discretized time step is defined as $\Delta t=t^{n+1}-t^n$.  To evolve $(\boldsymbol{W}_i^n, f_i^n, \phi_i^n)$  to $(\boldsymbol{W}_i^{n+1}, f_i^{n+1}, \phi_i^{n+1})$, the operator splitting method is used,
\begin{enumerate}[\textbf{Step}1:]
\item Electric field evolution: $$\phi^{n+1}_i = \mathcal{L}_{es}(\phi^n_i, \boldsymbol{W}^n_i)$$
\item Transport and collision process: $$ \boldsymbol{W}^{*}_i, f^{*}_i = \mathcal{L}_{tr}(\boldsymbol{W}^{n}_i, f^{n}_i, \Delta t)$$
\item Acceleration: $$ \boldsymbol{W}^{n+1}_i, f^{n+1}_i = \mathcal{L}_{ac}(\boldsymbol{W}^{*}_i, f^{*}_i, \phi^{n+1}, \Delta t)$$
\end{enumerate}
In this formulation, $\mathcal{L}_{es}$ represents the schemes for the evolution of the electrostatic field as introduced in section \ref{sec:les}, $\mathcal{L}_{tr}$ represents the numerical schemes for the coupled transport and collision processes, implemented using the UGKS and UGKWP method, as introduced in section \ref{sec:ugks} and \ref{sec:ugkwp}. $\mathcal{L}_{ac}$ denotes the schemes governing particle acceleration due to the electric field as shown in section \ref{sec:ugks} and \ref{sec:ugkwp}.  Specifically, the above schemes separately solve the following sub-systems of $\mathbf{RP}^{\lambda, \tau}$,
\begin{align}
    &\mathcal{L}_{es}:\nabla \cdot \left[\left(n + \lambda^2 \partial_t^2\right) \nabla \phi\right] = \nabla^2 : \mathbb{S},\label{eq:les}\\
    &\mathcal{L}_{tr}: \frac{\partial \boldsymbol{W}}{\partial t} + \nabla\cdot \mathbb{F}_{\boldsymbol{W}}=0 ,\quad \frac{\partial f}{\partial t}+\boldsymbol{u} \cdot \nabla_{\boldsymbol{x}} f = \frac{g-f}{\tau},\label{eq:ltr}\\
    &\mathcal{L}_{ac}: \frac{\partial \boldsymbol{W}}{\partial t} = \boldsymbol{S},\quad \frac{\partial f}{\partial t}- \frac{q\nabla{\phi}}{m} \cdot \nabla_{\boldsymbol{u}} f = 0, \label{eq:lac}
\end{align}
where $\boldsymbol{S}$ is the source term due to electrostatic forces. Detailed descriptions of all these schemes are provided in the subsequent section.

\subsection{Electrical field evolution scheme}
\label{sec:les}

First, we focus on solving $\mathcal{L}_{es}$. Eq.\eqref{eq:les} can be discretized as
\begin{equation}
\lambda^2 \frac{\Laplace{\phi}^{n+1}-2\Laplace{\phi}^n+\Laplace{\phi}^{n-1}}{\Delta t^2} + \nabla\cdot(n\nabla\phi^{n+1}) = \nabla^2:\mathbb{S}^{n}.
\label{eq:rpe_dis}
\end{equation}
The discretization of the moment equations can be written as
\begin{align}
   &n^{n+1} = n^k - \Delta t \nabla \cdot (n\boldsymbol{u})^n ,\label{eq:mass_discretized}\\
   &(n\boldsymbol{u})^{n+1} = (n\boldsymbol{u})^{n} - \Delta t \nabla \cdot \mathbb{S}^{n} + \Delta t n^n \nabla \phi^{n+1}.
   \label{eq:moment_discretized}
\end{align}
Take divergence of Equation \ref{eq:moment_discretized},
\begin{align}
   &\nabla\cdot(n\boldsymbol{u})^{n+1} = \nabla\cdot(n\boldsymbol{u})^{n} - \Delta t \nabla^2 : \mathbf{S}^{n} + \Delta t\nabla( n^n \nabla \phi^{n+1}),
   \label{eq:moment_discretized_div}
\end{align}
and the combine Equation \ref{eq:moment_discretized_div} and \ref{eq:rpe_dis}, we have
\begin{equation}
    \lambda^2 (\Laplace{\phi}^{n+1}-2\Laplace{\phi}^n+\Laplace{\phi}^{n-1}) = \Delta t(\nabla\cdot(n\boldsymbol{u})^{n} - \nabla\cdot(n\boldsymbol{u})^{n+1}).
    \label{eq:rpe_dis2}
\end{equation}
Combine Equation \ref{eq:mass_discretized} and \ref{eq:rpe_dis2}, we have
\begin{equation}
    \lambda^2 (\Laplace{\phi}^{n+1}-2\Laplace{\phi}^n+\Laplace{\phi}^{n-1}) = n^{n+2} - 2 n^{n+1} + n^{n}.
\end{equation}
Note $\lambda^2\Laplace{\phi}^n = n^n-1$ and $\lambda^2\Laplace{\phi}^{n-1} = n^{n-1}-1$, the above system is simplified as
\begin{equation}
\lambda^2 \Laplace{\phi}^{n+1} = n^{n+2} - 1.
\label{eq:rpe3}
\end{equation}
Similarly, we have $\lambda^2 \Laplace{\phi}^{n} = n^{n+1} - 1$ and $\lambda^2 \Laplace{\phi}^{n-1} = n^{n} - 1$, substitute these two equations into Equation \ref{eq:rpe_dis} to eliminate $\phi^n$ and $\phi^{n+1}$, we have
 \begin{equation}
\lambda^2 \Laplace{\phi}^{n+1} + \Delta t^2\nabla\cdot(n\nabla\phi^{n+1}) = \Delta t^2\nabla^2:\mathbb{S}^{n} + n^{n+1} -\Delta t \nabla\cdot(n\boldsymbol{U})^n .
\label{eq:rpe4}
\end{equation}
Combine Equation \ref{eq:mass_discretized} and \ref{eq:rpe4}, we finally have
 \begin{equation}
\lambda^2 \Laplace{\phi}^{n+1} + \Delta t^2\nabla\cdot(n\nabla\phi^{n+1}) = \Delta t^2\nabla^2:\mathbb{S}^{n} + n^{n} -2\Delta t \nabla\cdot(n\boldsymbol{U})^n .
\label{eq:rpe4}
\end{equation}
The mass flux vector and stress tensor can be calculated by definition or approximated by the numerical flux provided by the flux solver.
$$
\nabla^2 : \mathbb{S}^n = \nabla\cdot\left[\frac{1}{|\Omega_i|}\sum_{s\in\partial\Omega_i}|l_s|(\mathscr{F}_{(n\boldsymbol{U})})_s^n\right],\quad \nabla\cdot (n\boldsymbol{U})^n = \left[\frac{1}{|\Omega_i|}\sum_{s\in\partial\Omega_i}|l_s|(\mathscr{F}_{n})_s^n\right].
$$
Now we can get $\phi^{n+1}$ from quantities at time step $t_n$.

\subsection{UGKS method}
\label{sec:ugks}

We next introduce the UGKS framework for $\mathcal{L}_{tr}$. The UGKS method has been extensively studied; refer to \cite{xu2010unified, liu_unified_2016} for more details. UGKS concurrently evolves the macroscopic conservation equations and the microscopic kinetic equation:
$$
\frac{\partial \boldsymbol{W}}{\partial t} + \nabla\cdot \mathbb{F}_{\boldsymbol{W}}=0 ,\qquad
\frac{\partial f}{\partial t}+\boldsymbol{u} \cdot \nabla_{\boldsymbol{x}} f = \frac{g-f}{\tau}.
$$

Macro- and micro-scale evolution are coupled by computing the macroscopic numerical flux at cell interfaces from the interface distribution function, while updating the distribution function itself using macroscopic quantities at the interface.

The finite-volume discretization of the macroscopic conservation law is
\begin{equation}
\boldsymbol{W}_i^{n+1} = \boldsymbol{W}_i^n - \frac{\Delta t}{|\Omega_i|}\sum_{s\in\partial\Omega_i}|l_s|\mathscr{F}_{\boldsymbol{W}_s},
\label{eq:FVM discretization}
\end{equation}
where $l_s$ denotes a cell interface belonging to $\partial\Omega_i$, with center $\boldsymbol{x}_s$, outward unit normal $\boldsymbol{n}_{s}$, and area $|l_s|$. The numerical flux $\mathscr{F}_{\boldsymbol{W}_s}$ is obtained from the distribution function at the interface:
\begin{equation}
\mathscr{F}_{\boldsymbol{W}_s} = \frac{1}{\Delta t}\int_{t^n}^{t^{n+1}} \int \boldsymbol{u}\cdot \boldsymbol{n}_s \, f_s \, \boldsymbol{\Psi} \, \mathrm{d}\boldsymbol{u} \, \mathrm{d}t,
\label{eq:FVM macroscopic flux}
\end{equation}
with $f_s \equiv f(\boldsymbol{x}_{s},\boldsymbol{u},t)$ and $\boldsymbol{\Psi}=(1,\boldsymbol{u},\frac{1}{2}|\boldsymbol{u}|^2)^\mathsf{T}$ the vector of collision invariants.
The microscopic evolution follows the discretized BGK equation:
\begin{equation}
f_{i}^{n+1} = f_{i}^{n} + \frac{1}{|\Omega_i|} \sum_{s\in\partial\Omega_i} |l_s| \mathscr{F}_{f_{i,s}}
+ \int_{t^n}^{t^{n+1}} \frac{g_{i} - f_{i}}{\tau} \, \mathrm{d}t,
\label{eq:discretized-BGK-UGKS}
\end{equation}
where $f_{i}^n$ is the cell-averaged distribution function. The distribution function flux across the interface $l_s$ is
\begin{equation}
\mathscr{F}_{f_{i,s}} = \frac{1}{\Delta t} \int_{t^n}^{t^{n+1}} \boldsymbol{u}\cdot \boldsymbol{n}_s \, f(\boldsymbol{x}_{s},\boldsymbol{u},t) \, \mathrm{d}t.
\label{eq:interface-flux-ugks}
\end{equation}

To evaluate $\mathscr{F}_{f_{i,s}}$, we employ the integral solution of the BGK equation \cite{xu2001}:

\begin{equation}
f(\boldsymbol{x}_0,\boldsymbol{u}, t) = \frac{1}{\tau} \int_{0}^{t} g(\boldsymbol{x}^{\prime}, \boldsymbol{u}, t^{\prime}) e^{-(t-t^{\prime})/\tau} \, \mathrm{d} t^{\prime}
+ e^{-t/\tau} f_{0}\!\left(\boldsymbol{x}_0 - \boldsymbol{u} t \right),
\label{eq:BGKsoln}
\end{equation}which expresses the distribution function at a point $\boldsymbol{x}_0$ and time $t$. Here $f_0$ is the initial distribution at $t=0$, and $g$ denotes the equilibrium state along the characteristic $\boldsymbol{x}^{\prime} = \boldsymbol{x}_0 - \boldsymbol{u} t^{\prime}$. For second-order accuracy, we approximate
$$
g(\boldsymbol{x}^{\prime},\boldsymbol{u},t^{\prime}) = g_0 + \boldsymbol{g}_{\boldsymbol{x}} \cdot (\boldsymbol{x}^{\prime} - \boldsymbol{x}_0) + g_t \, t^{\prime},
$$
with $g_0 \equiv g(\boldsymbol{x}_0,\boldsymbol{u},t_0)$, and
$$
f_{0}(\boldsymbol{x}_0 - \boldsymbol{u} t) = f_{0}(\boldsymbol{x}_0) - \boldsymbol{f}_{\boldsymbol{x}} \cdot (\boldsymbol{u}t).
$$
The time-dependent distribution function at the cell interface then takes the analytic form
\begin{equation}
f(\boldsymbol{x}_0,\boldsymbol{u}, t) = c_1 g_0 + c_2 \, \boldsymbol{g}_{\boldsymbol{x}} \cdot \boldsymbol{u}
+ c_3 g_t + c_4 f_0 + c_5 \, \boldsymbol{f}_{\boldsymbol{x}} \cdot \boldsymbol{u},
\label{eq:tddf-solution-ugks}
\end{equation}
with coefficients
$$
\begin{aligned}
c_1 &= 1 - e^{-t/\tau}, \\
c_2 &= t e^{-t/\tau} - \tau\left(1 - e^{-t/\tau}\right), \\
c_3 &= t - \tau\left(1 - e^{-t/\tau}\right), \\
c_4 &= e^{-t/\tau}, \\
c_5 &= -t e^{-t/\tau}.
\end{aligned}
$$
The equilibrium state $g_0$ follows from the compatibility condition.
$$
\boldsymbol{W}_0 = \int g_0 \, \boldsymbol{\Psi} \, \mathrm{d}\boldsymbol{u}
= \int f_0 \, \boldsymbol{\Psi} \, \mathrm{d}\boldsymbol{u},
$$
where its spatial and temporal derivatives satisfy
$$
\int \boldsymbol{g}_{\boldsymbol{x}} \, \boldsymbol{\Psi} \, \mathrm{d}\boldsymbol{\Xi} = \boldsymbol{W}_{\boldsymbol{x}},
\qquad
\int g_t \, \boldsymbol{\Psi} \, \mathrm{d}\boldsymbol{\Xi}
= - \int \boldsymbol{u} \cdot \boldsymbol{g}_{\boldsymbol{x}} \, \boldsymbol{\Psi} \, \mathrm{d}\boldsymbol{\Xi}.
$$
Substituting Eq.~\eqref{eq:tddf-solution-ugks} into Eq.~\eqref{eq:interface-flux-ugks} yields
$$
\mathscr{F}_{f_{i,s}} = \frac{1}{\Delta t} \, \boldsymbol{u} \cdot \boldsymbol{n}_s \,
\Bigl( q_1 g_0 + q_2 \, \boldsymbol{g}_{\boldsymbol{x}} \cdot \boldsymbol{u}
+ q_3 g_t + q_4 f_0 + q_5 \, \boldsymbol{f}_{\boldsymbol{x}} \cdot \boldsymbol{u} \Bigr),
$$
where
$$
\begin{aligned}
q_1 &= \Delta t - \tau\left(1 - e^{-\Delta t/\tau}\right), \\
q_2 &= 2\tau^2\left(1 - e^{-\Delta t/\tau}\right) - \tau\Delta t - \tau\Delta t e^{-\Delta t/\tau}, \\
q_3 &= \frac{\Delta t^2}{2} - \tau\Delta t + \tau^2\left(1 - e^{-\Delta t/\tau}\right), \\
q_4 &= \tau\left(1 - e^{-\Delta t/\tau}\right), \\
q_5 &= \tau\Delta t e^{-\Delta t/\tau} - \tau^2\left(1 - e^{-\Delta t/\tau}\right).
\end{aligned}
$$
Finally, the collision source term in Eq.~\eqref{eq:discretized-BGK-UGKS} is discretized using the Crank–Nicolson method \cite{xu2010unified}.

Then we introduce the acceleration scheme $\mathcal{L}_{acc}$ for UGKS. For the acceleration equation
$$
\frac{\partial f}{\partial t}- \frac{q\nabla{\phi}}{m} \cdot \nabla_{\boldsymbol{u}} f = 0,
$$
we use the second-order semi-Lagrangian method to discretize this equation. The algorithm employs a backward characteristic tracing approach with linear interpolation. Specifically, the precise solution along the characteristics is
$$
f^*_i(\boldsymbol{u}_j) = f^n_i(\boldsymbol{u}_j-\boldsymbol{a}\Delta t),
$$
where $\boldsymbol{u}_j$ represents velocity points in velocity space $\Omega_j$, $\boldsymbol{a}=-\frac{q\nabla\phi}{m}$ is the accleration term. The backward characteristic tracing is performed for each velocity grid point $\boldsymbol{u}_j$ as
$$
\boldsymbol{u}_{source} = \boldsymbol{u}_j-\boldsymbol{a}\Delta t,
$$
and the linear interpolation scheme is used
$$
f^n_i(\boldsymbol{u}_{source}) = (1-\alpha) f^n_i(\boldsymbol{u}_k) + \alpha f^n_i(\boldsymbol{u}_{k+1}),
$$
where $f_k$ and $f_{k+1}$ are distribution function values at adjacent grid points, $k$ is the index such that $v_k \leq v_{source} \leq v_{k+1}$. The interpolation weight $\alpha$ is calculated as
$$
\alpha = \frac{v_{source}-v_k}{v_{k+1}-v_k}.
$$
Higher-order interpolation can be used to improve accuracy. After updating the distribution functions, the macroscopic variables are recomputed.

\subsection{UGKWP method}
\label{sec:ugkwp}

In contrast to the UGKS, the UGKWP method employs a hybrid representation of the velocity distribution function. Specifically, the equilibrium flux is computed analytically, whereas stochastic particles represent the non-equilibrium (free-transport) flux. This hybrid approach significantly enhances computational efficiency, particularly in three-dimensional simulations.

According to Eq.\eqref{eq:BGKsoln}, the cumulative distribution function of the particle's free streaming time $t_{f}$ before the collision is given as
$$
F\left( t_{f} < t \right) = \exp\left( - t\text{/}\tau \right),
$$
from which the free stream time $t_{f}$ can be sampled as $t_{f} = - \tau\ln(\eta)$ with $\eta$ a random varible subject to the uniform distribution $\eta \sim U(0,1)$ . For a time step $\Delta t$, the particles with $t_{f} \geq \Delta t$ will undergo collisionless free streaming, and the particles with $t_{f} < \Delta t$ will experience collisional interactions. The procedure of updating particles in the UGKWP method is
\begin{enumerate}[\textbf{Step} 1:]

\item During the time step, stream each particle $P_k$ for a time period of $\min(\Delta t, t_{f,k})$. Then identify and retain the collisionless particles, while removing the collisional particles. Calculate the free-transport flux across cell interfaces contributed by the particles and accumulate the total conservative quantities of the particles $\boldsymbol{W}_i^p$;

\item Updating the macroscopic conservative variables based on the equilibrium and free transport numerical flux across the cell interfaces, then calculate the total conservative quantities of the collisional particles $\boldsymbol{W}_i^h$ from the updated conservative quantities $\boldsymbol{W}_i$ as $\boldsymbol{W}_i^h = \boldsymbol{W}_i - \boldsymbol{W}_i^p$;

\item At the end of the time step, resample the collisionless particles from the distribution $\boldsymbol{W}_i^h$ according to the updated conservative quantities, sample the free-streaming time $t_{f,k}$ for each particle $P_k$ from the cumulative distribution function $F(t_f < t) = \exp(-t/\tau)$.
\end{enumerate}

Refer to \cite{pu2025unified} for a more detailed description of the procedure. In the procedure outlined above, the algorithm for updating the macroscopic conservative variables is introduced in the following paragraph.

The discretized evolution equation of macroscopic variables is
\begin{equation}
    \boldsymbol{W}_{\boldsymbol{i}}^{*} = \boldsymbol{W}_{\boldsymbol{i}}^{\boldsymbol{n}} - \sum_{s}^{}{\frac{\Delta t}{\left| \Omega_{i} \right|}\left| l_{s} \right|(\mathscr{F}_{\boldsymbol{W}})_s^{g}} - \sum_{s}^{}{\frac{\Delta t}{\left| \Omega_{i} \right|}\left| l_{s} \right|(\mathscr{F}_{\boldsymbol{W}})_s^{f,w}} + \frac{1}{\left| \Omega_{i} \right|}(\mathscr{F}_{\boldsymbol{W}})_s^{f,p}  ,
\end{equation}
where $(\mathscr{F}_{\boldsymbol{W}})_s^g$ is the equilibrium flux, $(\mathscr{F}_{\boldsymbol{W}})_s^{f,w}$ and $(\mathscr{F}_{\boldsymbol{W}})_s^{f,p}$ are the free transport flux contributed by wave and particles.

The numerical flux of the macroscopic conservative variable can be decomposed into the equilibrium and free-streaming fluxes according to Eq.\eqref{eq:BGKsoln}. The equilibrium flux is
\begin{equation}
(\mathscr{F}_{\boldsymbol{W}})_s^g = \frac{1}{\Delta t}\int_{t^n}^{t^{n+1}}\boldsymbol{u}\cdot \boldsymbol{n}_s
\left[\frac{1}{\tau} \int_{0}^{t} g(\boldsymbol{x}^{\prime}, \boldsymbol{u}, \boldsymbol{\xi},t^{'}) e^{-\left(t-t^{\prime}\right) / \tau} \mathrm{d} t^{\prime}\right]
\boldsymbol{\Psi} \mathrm{d}\boldsymbol{\Xi}\mathrm{d}t,
\label{eq:macroscopic eq flux}
\end{equation}
and the free streaming flux is
\begin{equation}
(\mathscr{F}_{\boldsymbol{W}})_s^f = \frac{1}{\Delta t}\int_{t^n}^{t^{n+1}}\boldsymbol{u}\cdot \boldsymbol{n}_s
\left[e^{-t / \tau} f_{0}\left(\boldsymbol{x}-\boldsymbol{u} t \right)\right]
\boldsymbol{\Psi} \mathrm{d}\boldsymbol{\Xi}\mathrm{d}t.
\label{eq:macroscopic fr flux}
\end{equation}

The equilibrium flux can be calculated as
\begin{equation}
(\mathscr{F}_{\boldsymbol{W}})_s^g = \int_{}^{}\boldsymbol{u} \cdot \boldsymbol{n}_{s}\left( C_{1}g_{0} + C_{2}\boldsymbol{u} \cdot {\boldsymbol{g}}_{0x} + C_{3}g_{0t} \right)\Psi d\boldsymbol{\Xi},
\label{eq:eqflux-numerical}
\end{equation}
where the time integration coefficients are
\begin{align*}
&C_{1} = \Delta t - \tau\left( 1 - e^{- \Delta t\text{/}\tau} \right)
,\\
&C_{2} = 2\tau^2(1 - e^{- \Delta t\text{/}\tau}) - \tau\Delta t -\tau \Delta t  e^{- \Delta t\text{/}\tau},\\
&C_{3} = \frac{\Delta t^2}{2} - \tau\Delta t + \tau^2(1-e^{\Delta t/\tau}).
\end{align*}

The free stream flux $(\mathscr{F}_{\boldsymbol{W}})_s^f$ is given as
$$
(\mathscr{F}_{\boldsymbol{W}})_s^{f,w} = \int_{}^{}\boldsymbol{u} \cdot \boldsymbol{n}_s\left( C_{4}g_{0}^{+,c} + C_{5}\boldsymbol{u} \cdot {g}_{0x}^{+,c} \right)\boldsymbol{\Psi}d\boldsymbol{\Xi},
$$
where the time integration coefficients are
\begin{align*}
&C_{4} = \tau\left( 1 - e^{- \Delta t\text{/}\tau} \right) - \Delta te^{- \Delta t\text{/}\tau},
,\\
&C_{5} = \tau\Delta t e^{- \Delta t\text{/}\tau} - \tau^{2}\left( 1 - e^{- \Delta t\text{/}\tau} \right) + \frac{\Delta t^{2}}{2}e^{- \Delta t\text{/}\tau}.
\end{align*}

And the net particle flux $(\mathscr{F}_{\boldsymbol{W}})_s^{f,p}$ is calculated as
$$
(\mathscr{F}_{\boldsymbol{W}})_s^{f,p}=  \sum_{k \in P_{\partial\Omega_{i}^{+}}}^{}\boldsymbol{W}_{P_{k}} - \sum_{k \in P_{\partial\Omega_{i}^{-}}}^{}\boldsymbol{W}_{P_{k}},
$$
where $\boldsymbol{W}_{P_{k}} = \left( m_{k},m_{k}\boldsymbol{v}_{k},\frac{1}{2}m_{k}\boldsymbol{v}_{\boldsymbol{k,}}^{2} \right),P_{\partial\Omega_{i}^{-}}$
is the index set of the particles streaming out of cell $\Omega_{i}$
during a time step, and $P_{\partial\Omega_{i}^{+}}$ is the
index set of the particles streaming into cell $\Omega_{i}$.

Next, we introduce the acceleration operator $\mathcal{L}_{\text{acc}}$ within the UGKWP framework. Acceleration is applied distinctly to both components of the hybrid representation. For each discrete particle $P_k$, the velocity update follows Newton's second law:
$$
\boldsymbol{v}^{n+1}_k = \boldsymbol{v}^{n}_k + \boldsymbol{a} \Delta t,
$$
where $\boldsymbol{a}$ is the acceleration computed from the cell-averaged electric field. Concurrently, the macroscopic wave (or continuum) component $\boldsymbol{U}^h$ undergoes an analogous update:
$$
\boldsymbol{U}^{h,n+1} = \boldsymbol{U}^{h,n} + \boldsymbol{a} \Delta t.
$$
Following the synchronous updates of the particle and wave components, the macroscopic conserved quantities (density, momentum, energy) are recalculated to maintain a consistent hydrodynamic state.

\subsection{Analysis}
\label{sec:analysis}

\begin{lemma}
The scheme $\mathbf{RP}_h^{\lambda,\tau}$ to the system $\mathbf{RP}^{\lambda,\tau}$ with spatial and temporal discretizations $\Delta x$ and $\Delta t$, exhibits unified preserving property in the hydrodynamical limit and asymptotic preserving property in the quasineutral limit. Specifically, in the hydrodynamical limit ($\tau \rightarrow 0$), the scheme is unified-preserving, as the time step $\Delta t$ is independent of $\tau$, allowing for second-order accuracy in recovering the Navier-Stokes and Euler equations. In the quasineutral limit ($\lambda \rightarrow 0$), the scheme is asymptotic-preserving, with the spatial step $\Delta x$ and temporal step $\Delta t$ independent of $\lambda$ and $\omega_{pe}$.
\end{lemma}

\begin{proof}
In the hydrodynamic limit, when $\tau = 0$, the exponential term $e^{-\Delta t/\tau}$ vanishes, leading to the disappearance of the free transport flux as described by Eq. \eqref{eq:BGKsoln}. Consequently, only the equilibrium flux contributes to the numerical flux. As given by Eq. \eqref{eq:eqflux-numerical}, when $\tau = 0$, the numerical flux simplifies to
$$
(\mathscr{F}_{\boldsymbol{W}})_s^g = \Delta t\int \boldsymbol{u} \cdot \boldsymbol{n}_{s} \left( g_{0} + \frac{\Delta t}{2} g_{0t} \right)\Psi d\boldsymbol{\Xi},
$$
which is the numerical flux in the Euler limit. For small but finite relaxation times $\tau$, the free transport flux in Eq. \eqref{eq:BGKsoln} decays exponentially, and the numerical flux becomes dominated by the equilibrium flux. Within the equilibrium flux expression (Eq. \eqref{eq:eqflux-numerical}), the terms $C_1$, $C_2$, and $C_3$ are dominated by $\Delta t$, $-\tau \Delta t$, and $\frac{\Delta t^2}{2}-\tau\Delta t$, respectively. This results in the equilibrium numerical flux
$$
(\mathscr{F}_{\boldsymbol{W}})_s^g = \Delta t\int \boldsymbol{u} \cdot \boldsymbol{n}_{s} \left( g_{0} -\tau (\boldsymbol{u}\cdot\boldsymbol{g}_{0x}+g_{0t}) \frac{\Delta t}{2} g_{0t} \right)\Psi d\boldsymbol{\Xi},
$$
which is precisely the numerical flux in the Navier-Stokes limit \cite{xu2001, liu_unified_2021}. Due to the inclusion of coupled collision and free transport fluxes in the numerical flux, the scheme is not restricted by the collision time $\tau$.

In the quasi-neutral regime, as $\lambda \rightarrow 0$, Eq. \eqref{eq:rpe4} simplifies to:
\begin{equation*}
\begin{aligned}
\Delta t^2\nabla\cdot(n\nabla\phi^{n+1}) = \Delta t^2\nabla^2:\mathbb{S}^{n} + n^{n} -2\Delta t \nabla\cdot(n\boldsymbol{U})^n.
\end{aligned}
\end{equation*}
This equation is precisely the discretization of the reformulated Poisson equation in the quasineutral limit (Eq. \eqref{eq:rpe_qn}). Since the Debye length $\lambda$ disappears, the time step is not restricted by the Debye length $\lambda$.
\end{proof}

\section{Numerical studies}
\label{sec:numericalstudies}

\subsection{Nonlinear Landau damping}
\label{sec:nld}

This section examines nonlinear Landau damping. Here, the Debye length is fully resolved, whereas the mean free path is not. We investigate this regime across a range of Knudsen numbers, from $Kn=\infty$ to $Kn=1$ and $Kn=0.0001$, to assess the scheme's capacity to accurately capture physics spanning collisionless to collisional regimes without limitations imposed by the mean free path or mean collision time.

The initial electron distribution function is given by:
$$
f_{0}(x, u)=\frac{1}{\sqrt{2 \pi}}(1+\alpha \cos (k x)) \mathrm{e}^{-u^2/2}.
$$
The total energy of the system, $E_t$, is calculated as:
$
\boldsymbol{E}_t=\frac{\lambda^2}{2} \int\left|\nabla\phi\right|^2 d x+\frac{1}{2} \int f|v|^2 d x d v.
$
The electrical energy, $E_{p}$, is defined as:
$
\boldsymbol{E}_{p}=\frac{\lambda^2}{2} \int\left|\nabla\phi\right|^2 d x.
$
Simulations were performed using 128 spatial grid points over the domain $x\in [0, L]$, where $L = 2\pi/k$ and $k=0.5$. The perturbation amplitude, $\alpha$, was set to 0.5. For the UGKS-RPE, the velocity space was discretized into 256 points within the interval $[-5,5]$. The UGKWP-RPE method employed an initial condition of 1000 particles per cell. Note that the particle number in UGKWP-RPE in the runtime varies with the Knudsen number.

Figure \ref{fig:nonlinear-potential-compare} presents a comparison of the electric potential energy for different Knudsen numbers. The results demonstrate excellent agreement between the UGKWP-RPE and UGKS-RPE methods across all simulated Knudsen numbers.
The first Figure displays results for $Kn=\infty$, representing the collisionless regime. The temporal evolution of the electric energy shows an initial damping followed by a rise, characteristic of nonlinear Landau damping. The observed damping rates align well with theoretical predictions.
The second Figure corresponds to $Kn=1$. In this regime, the introduction of collisions suppresses nonlinear damping, leading to a monotonic decrease in electric energy. A deviation between UGKWP-RPE and UGKS-RPE is observed around $t=35$, attributed to particle noise inherent in the UGKWP method.
The third Figure presents results for $Kn=0.0001$, a near-continuum regime. Here, damping behavior is absent. This is because in the near-continuum limit, the plasma behaves more like a fluid, with macroscopic interactions dominating over kinetic wave-particle interactions, thereby eliminating Landau damping.
Collectively, these results demonstrate the capability of both schemes to accurately resolve the underlying physics across a wide range of Knudsen numbers, from the collisionless to the collisional regimes.

\begin{figure}
    \centering
    \begin{subfigure}[b]{0.48\textwidth}
    \centering
    \includegraphics[width=1.0\linewidth]{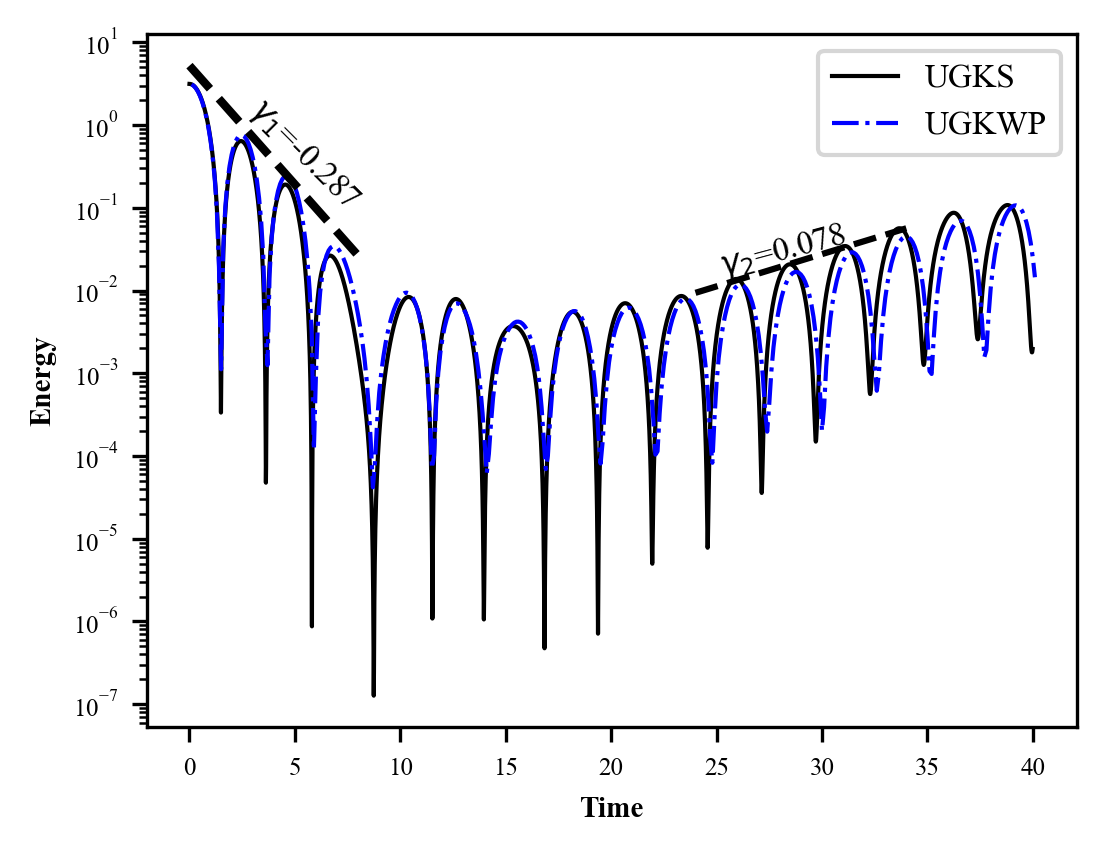}
    \end{subfigure}
    \hfill
    \begin{subfigure}[b]{0.48\textwidth}
    \centering
    \includegraphics[width=1.0\linewidth]{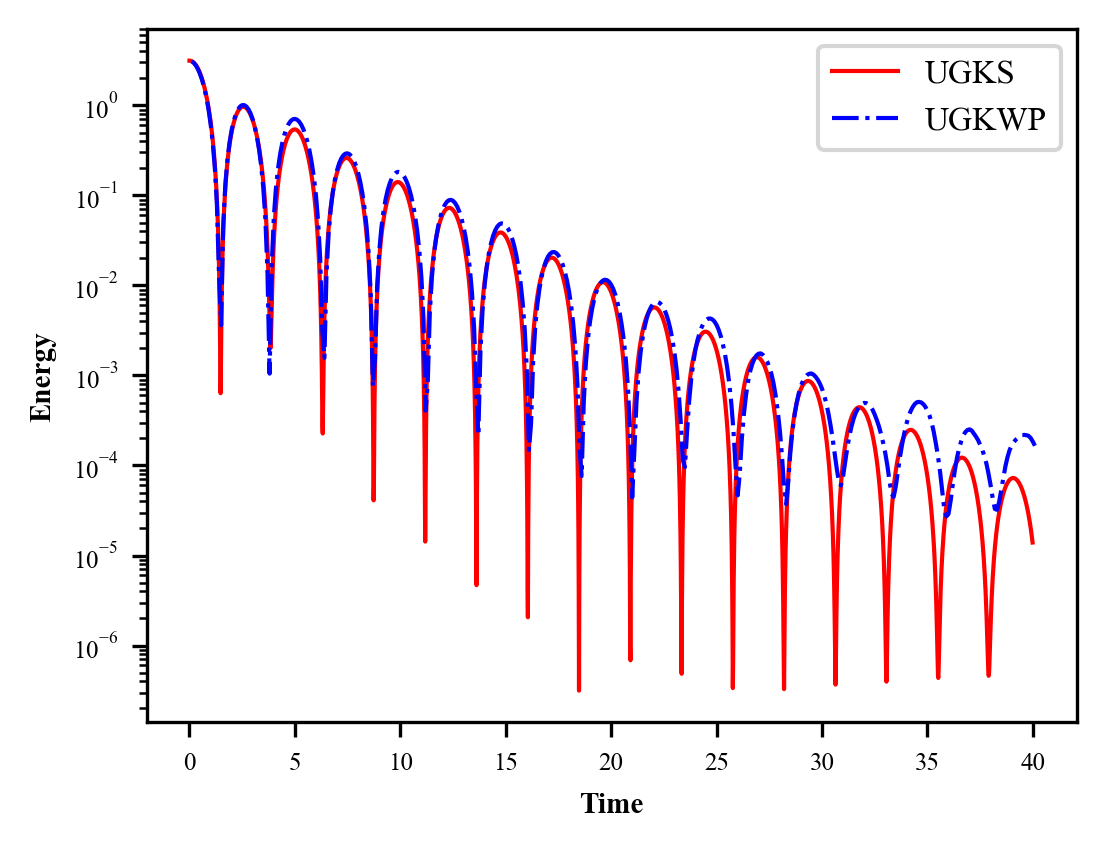}
    \end{subfigure}
    \hfill
    \begin{subfigure}[b]{0.48\textwidth}
    \centering
    \includegraphics[width=1.0\linewidth]{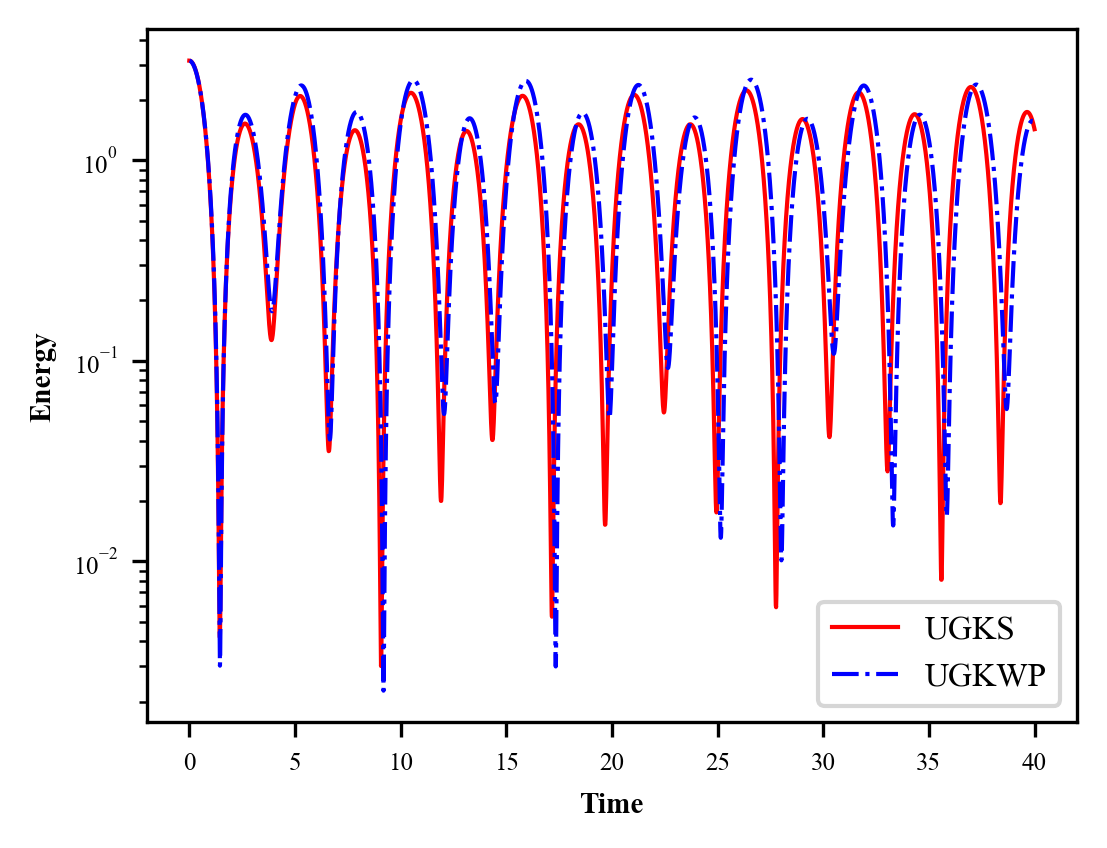}
    \end{subfigure}
    \caption{Nonlinear Landau damping with different Knudsen numbers by UGKS-RPE and UGKWP-RPE: $\lambda = 1, \Delta x = 0.1, \text{CFL}=0.9$. The Figure shows the time evolution of electrostatic energy.}
    \label{fig:nonlinear-potential-compare}
\end{figure}

Figure \ref{fig:nld_phase_diagram} presents the phase diagrams generated by the UGKS-RPE for various Knudsen numbers at simulation times $t=1, 4, 8, 15$. Each row in the Figure corresponds to a specific simulation time, while each column represents a different Knudsen number. In the collisionless regime ($Kn=\infty$), the filamentation process is clearly observed as time progresses, a characteristic signature of Landau damping. For $Kn=1$, the phase diagram at $t=1$ is similar to the collisionless case, but the phase space becomes noticeably smoother due to the presence of collisions. At $Kn=0.0001$, the phase space exhibits a significantly different structure compared to the other two regimes.
Figure \ref{fig:nld_phase_diagram-wp} illustrates the phase diagrams obtained using the UGKWP-RPE method. As UGKWP specifically tracks collisionless particles, only the non-equilibrium portion of the phase space is plotted. The collisionless regime ($Kn=\infty$) generated by UGKWP-RPE is identical to the UGKS-RPE results, displaying clear filaments. For $Kn=1$, the UGKWP-RPE phase diagram also shares the same structure with UGKS-RPE because a substantial number of particles remain in a non-equilibrium state at this Knudsen number, allowing for a discernible phase space representation. However, at $Kn=0.0001$, a much smaller proportion of particles are non-equilibrium. In this regime, the majority of particles are in equilibrium and are thus represented by the analytical wave formulation. The Figure reveals that the non-equilibrium particles are concentrated at the spatial extremes and the center of the domain, with sparse particle distribution elsewhere, indicating a local equilibrium state in those regions. This observation clearly demonstrates the adaptive nature of UGKWP in representing the phase space.

\begin{figure}
    \centering
    \begin{subfigure}[b]{0.325\textwidth}
    \centering
    \includegraphics[width=1.0\linewidth]{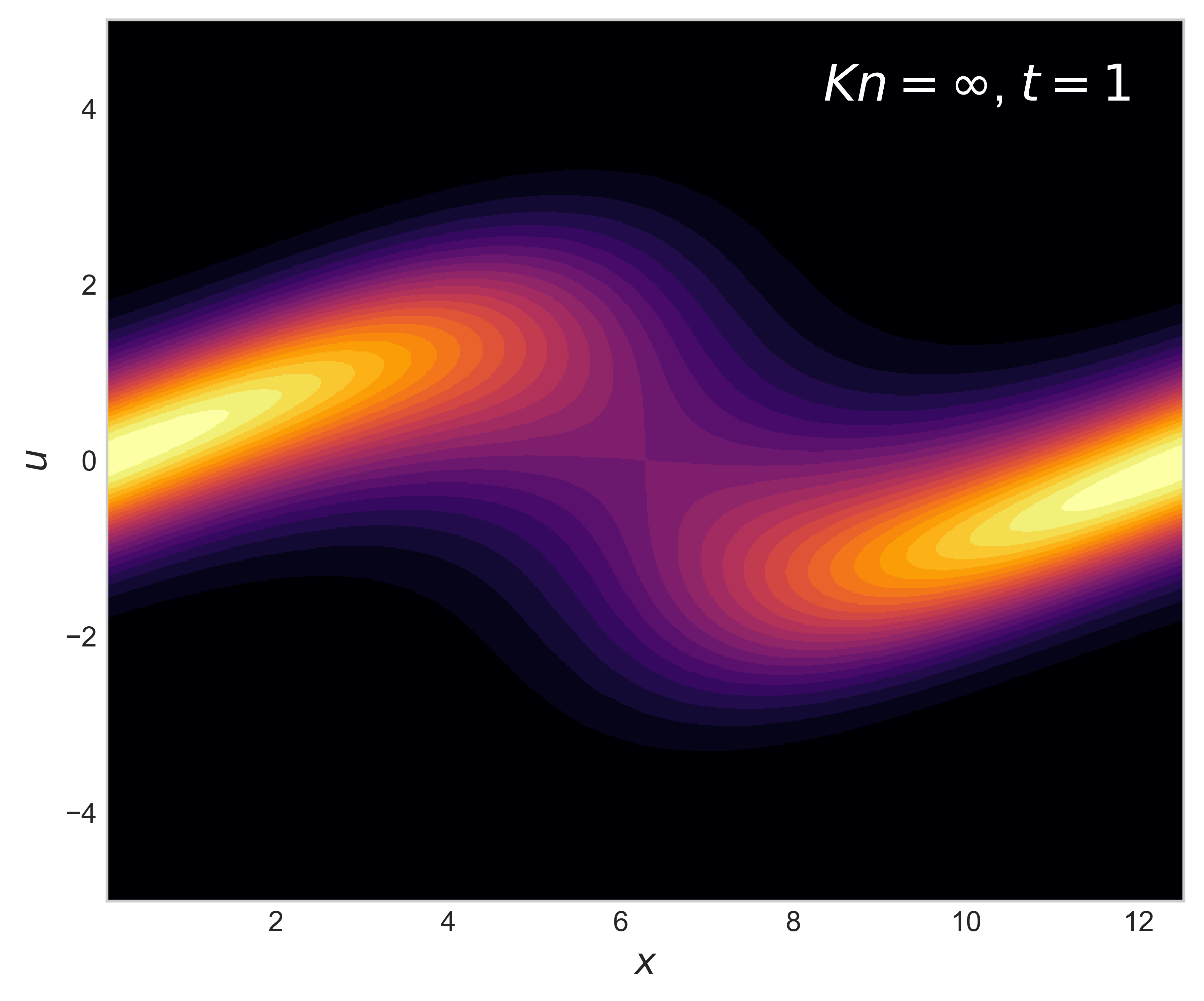}
    \end{subfigure}
    \hfill
    \begin{subfigure}[b]{0.325\textwidth}
    \centering
    \includegraphics[width=1.0\linewidth]{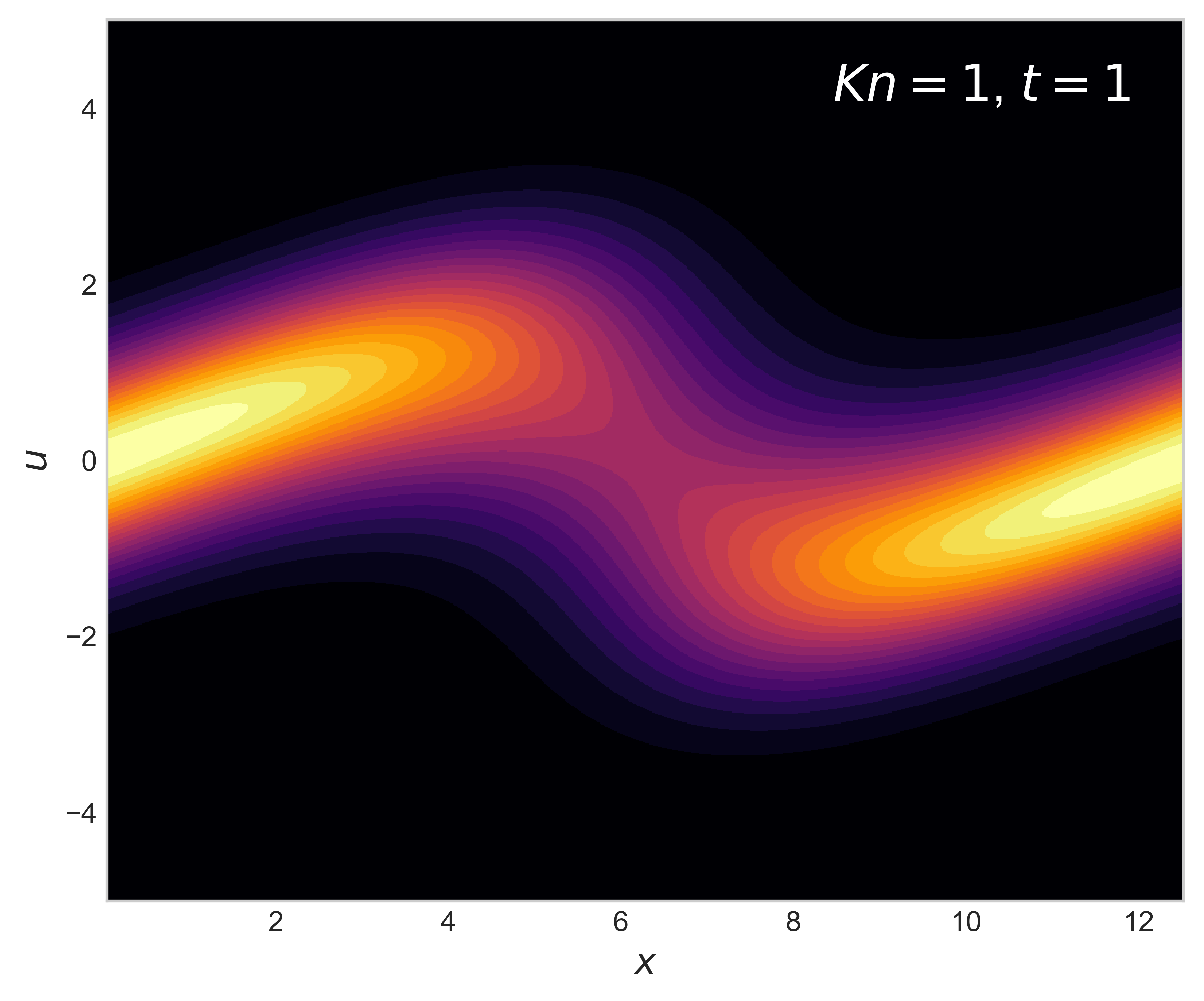}
    \end{subfigure}
    \hfill
    \begin{subfigure}[b]{0.325\textwidth}
    \centering
    \includegraphics[width=1.0\linewidth]{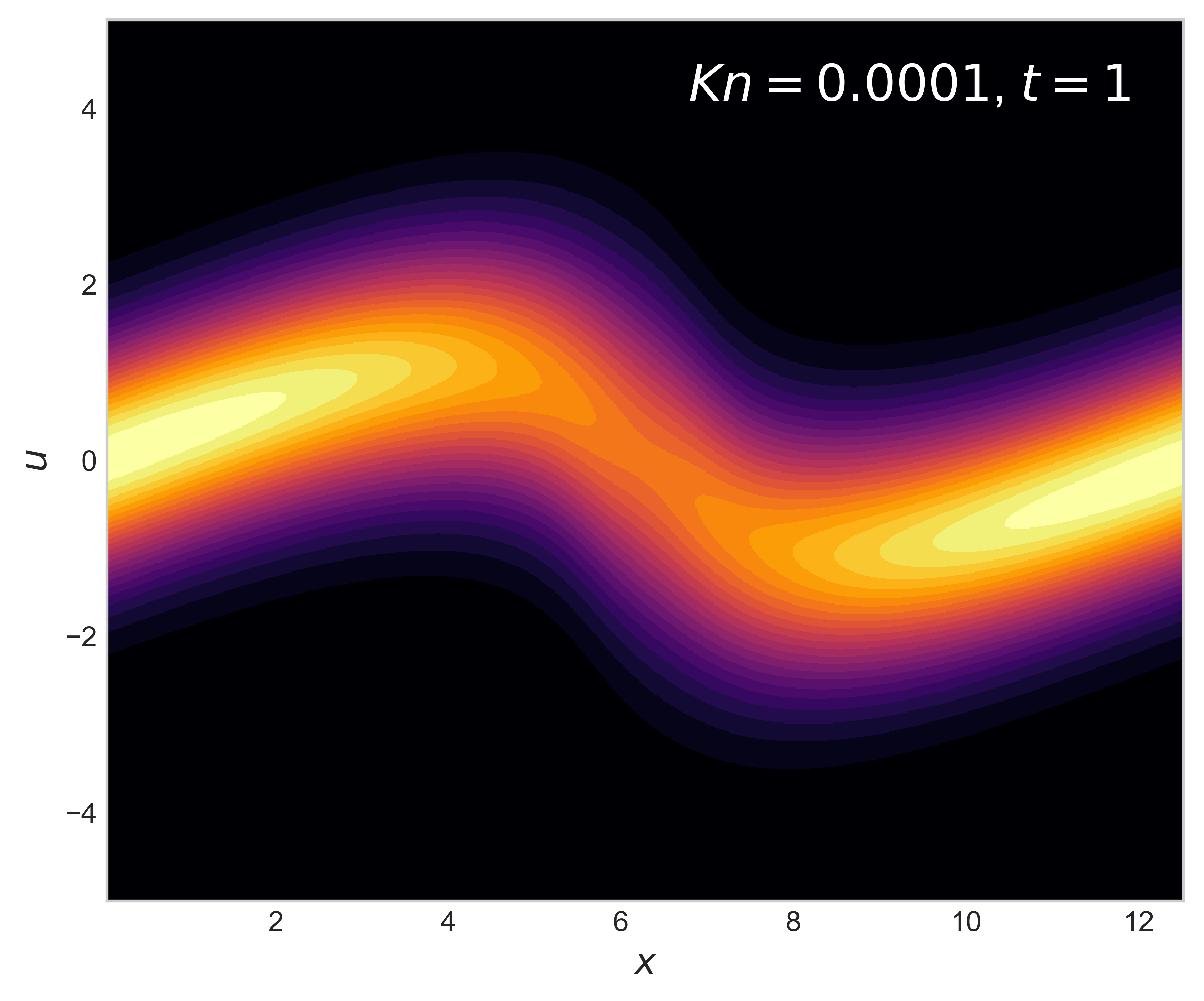}
    \end{subfigure}
    \vfill
    \centering
    \begin{subfigure}[b]{0.325\textwidth}
    \centering
    \includegraphics[width=1.0\linewidth]{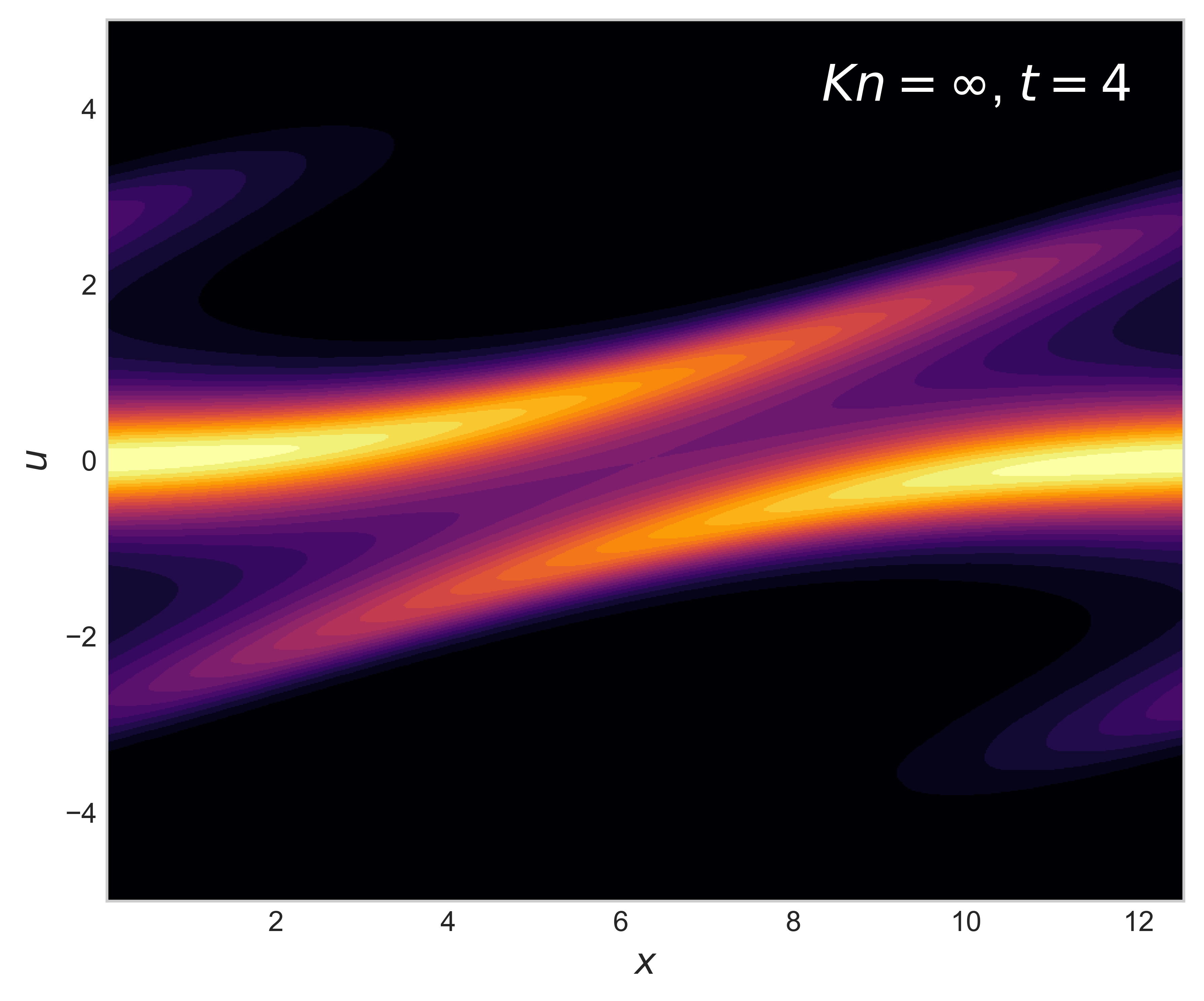}
    \end{subfigure}
    \hfill
    \begin{subfigure}[b]{0.325\textwidth}
    \centering
    \includegraphics[width=1.0\linewidth]{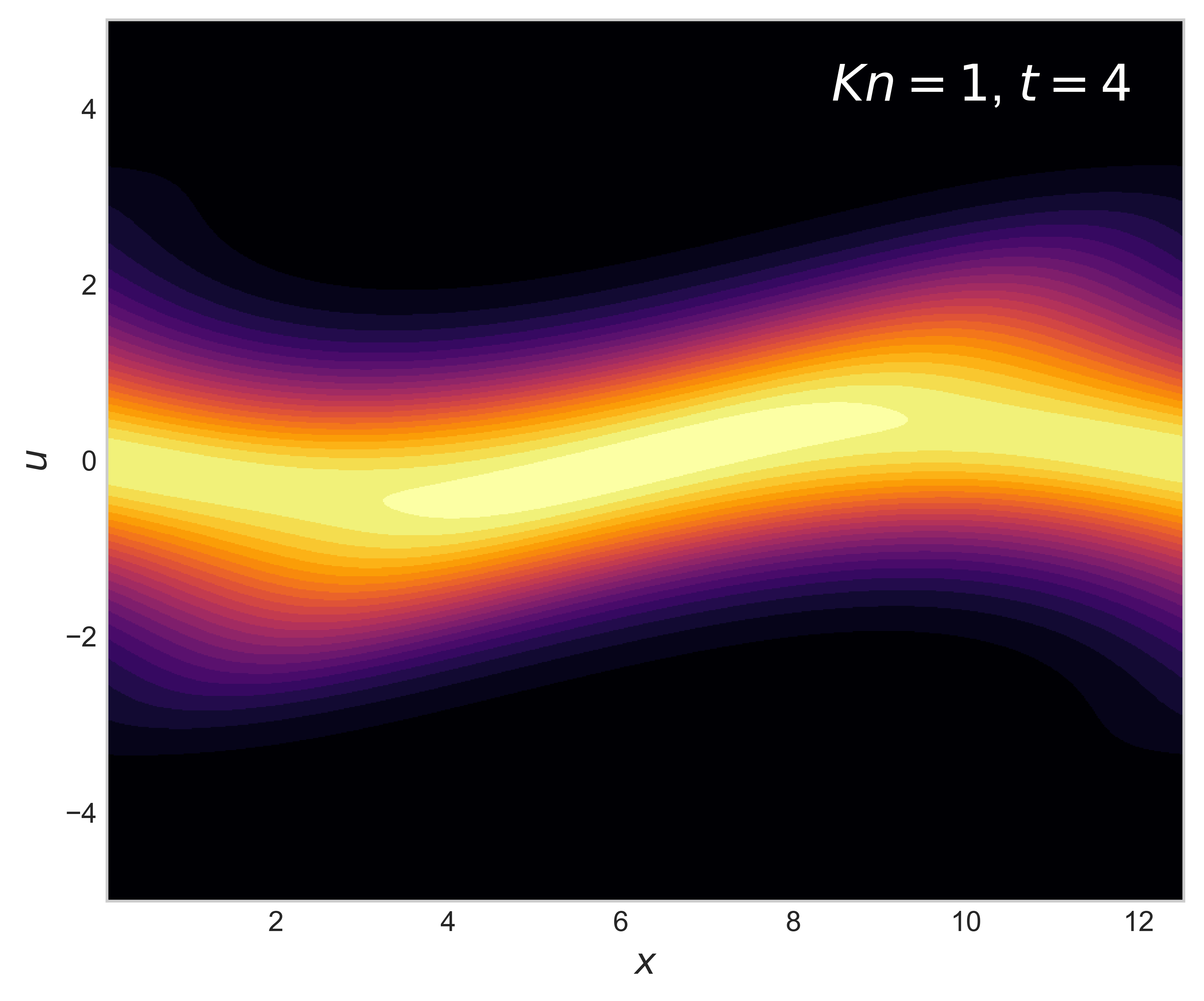}
    \end{subfigure}
    \hfill
    \begin{subfigure}[b]{0.325\textwidth}
    \centering
    \includegraphics[width=1.0\linewidth]{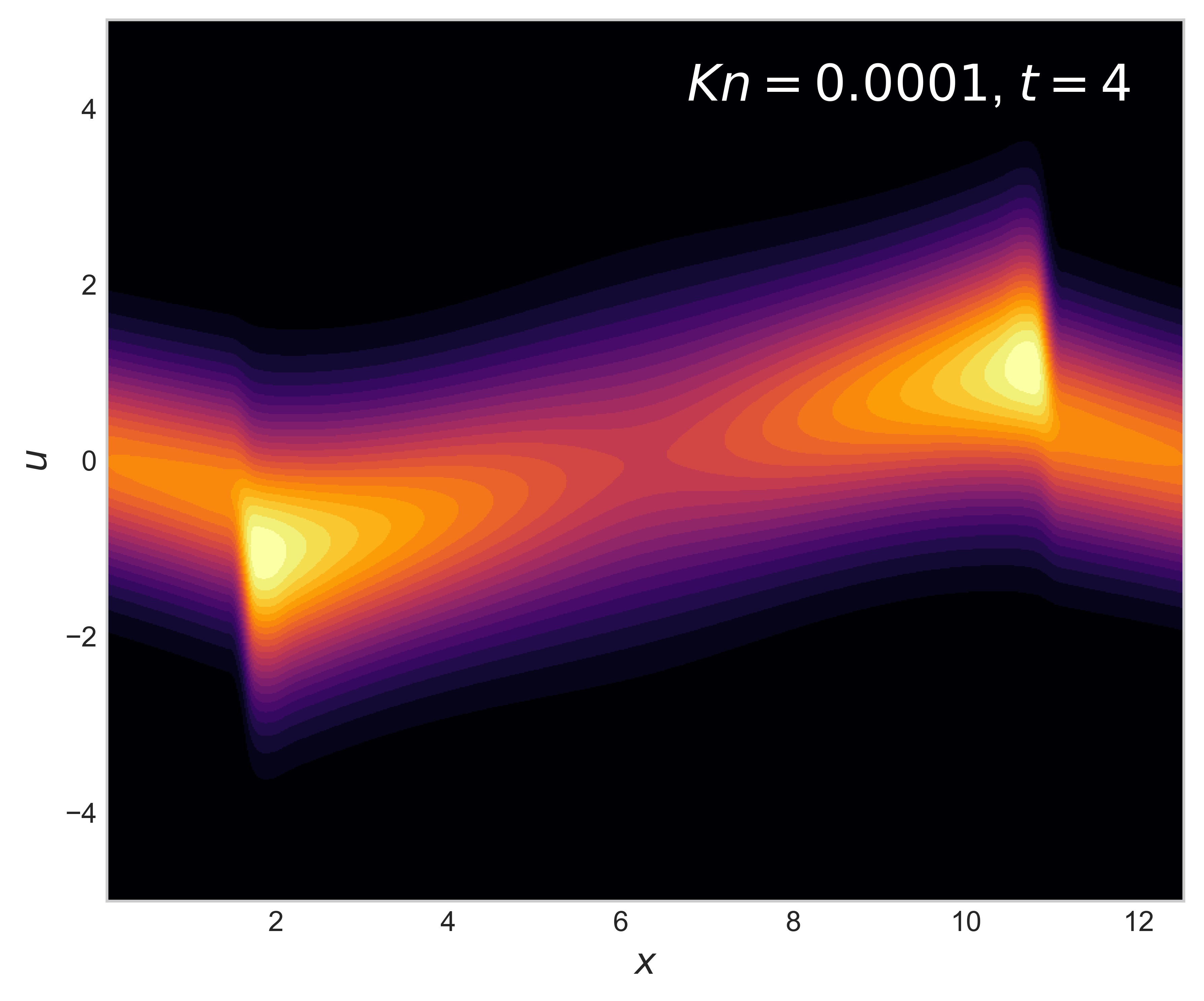}
    \end{subfigure}
    \vfill
    \centering
    \begin{subfigure}[b]{0.325\textwidth}
    \centering
    \includegraphics[width=1.0\linewidth]{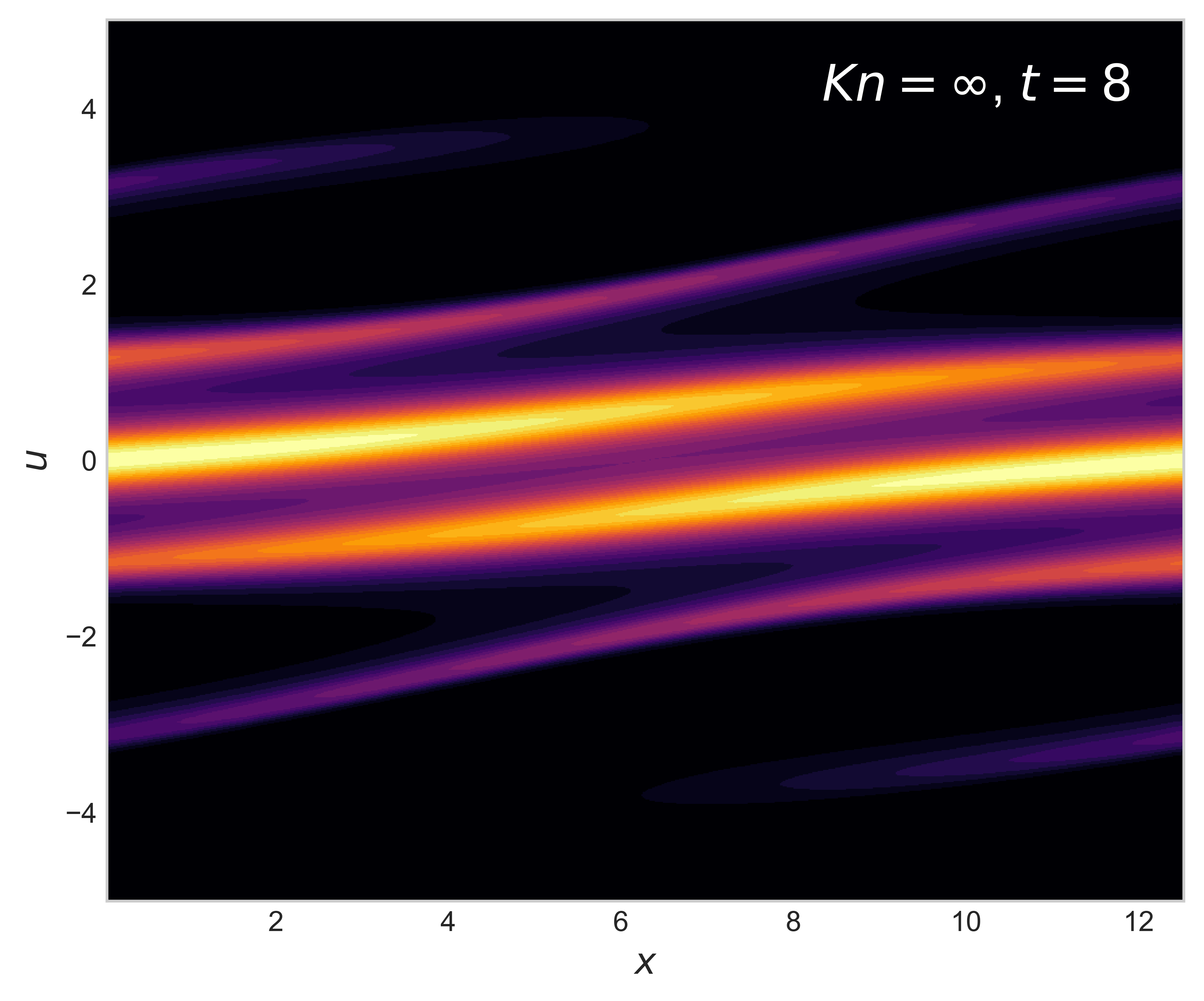}
    \end{subfigure}
    \hfill
    \begin{subfigure}[b]{0.325\textwidth}
    \centering
    \includegraphics[width=1.0\linewidth]{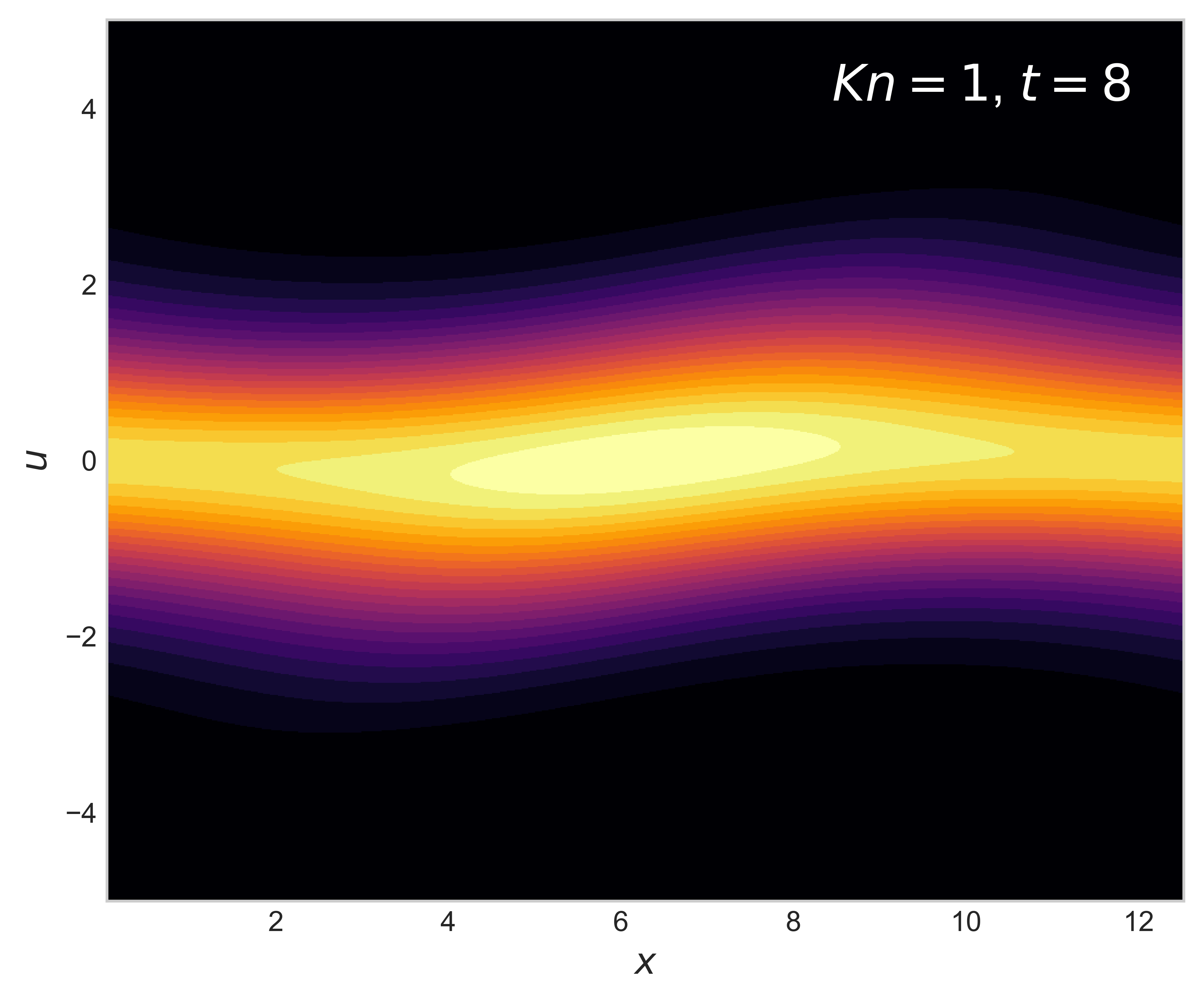}
    \end{subfigure}
    \hfill
    \begin{subfigure}[b]{0.325\textwidth}
    \centering
    \includegraphics[width=1.0\linewidth]{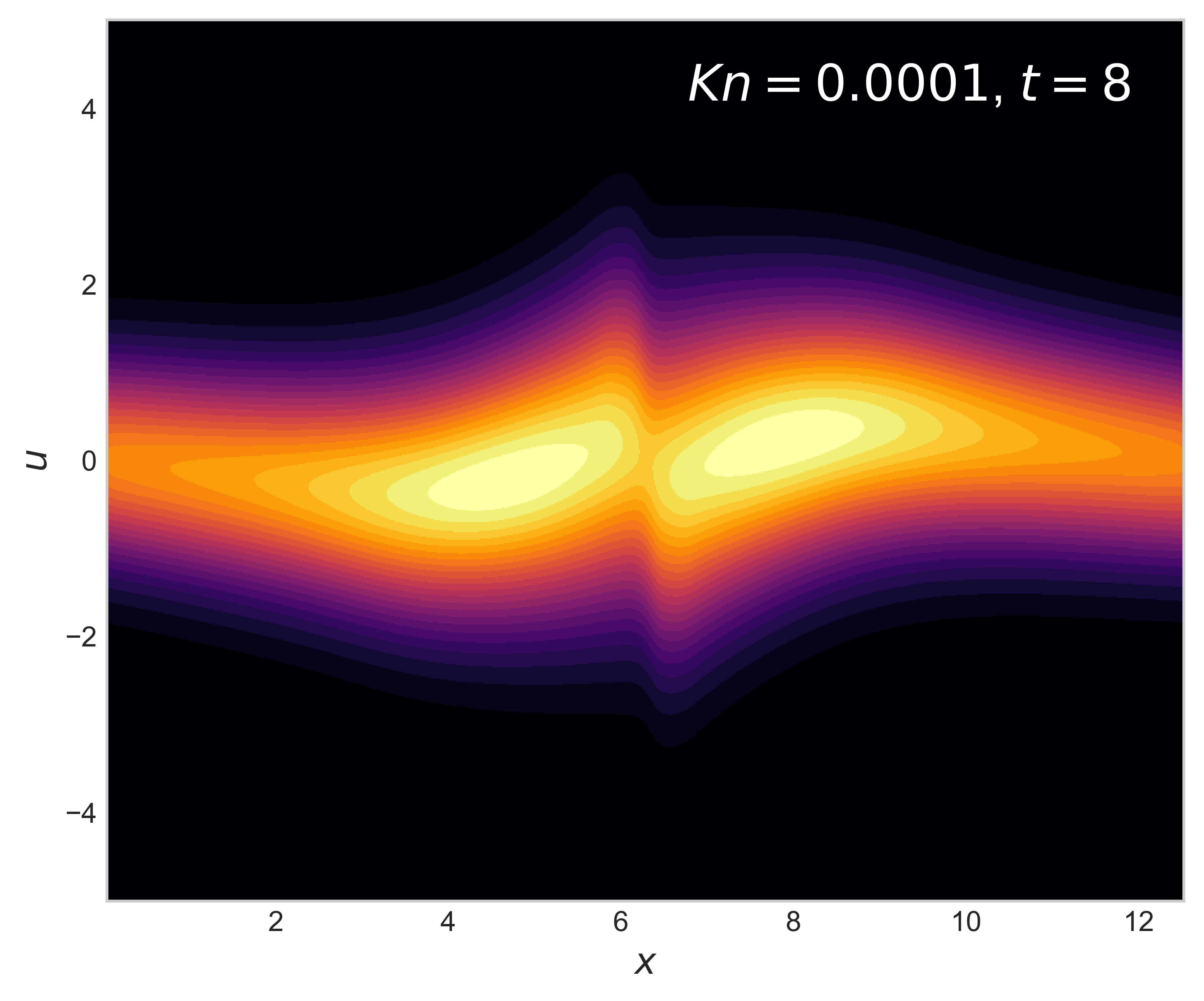}
    \end{subfigure}
    \vfill
    \centering
    \begin{subfigure}[b]{0.325\textwidth}
    \centering
    \includegraphics[width=1.0\linewidth]{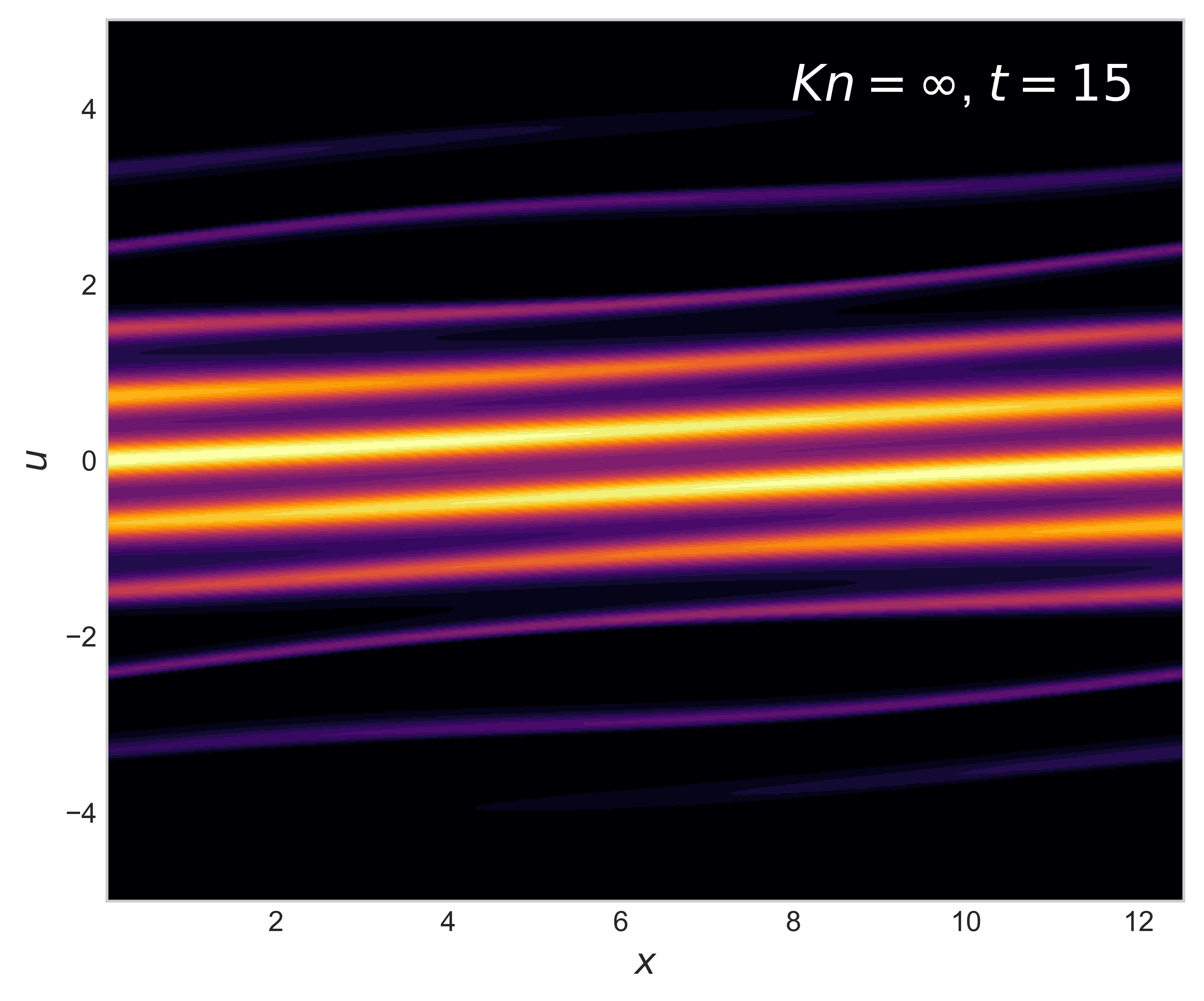}
    \end{subfigure}
    \hfill
    \begin{subfigure}[b]{0.325\textwidth}
    \centering
    \includegraphics[width=1.0\linewidth]{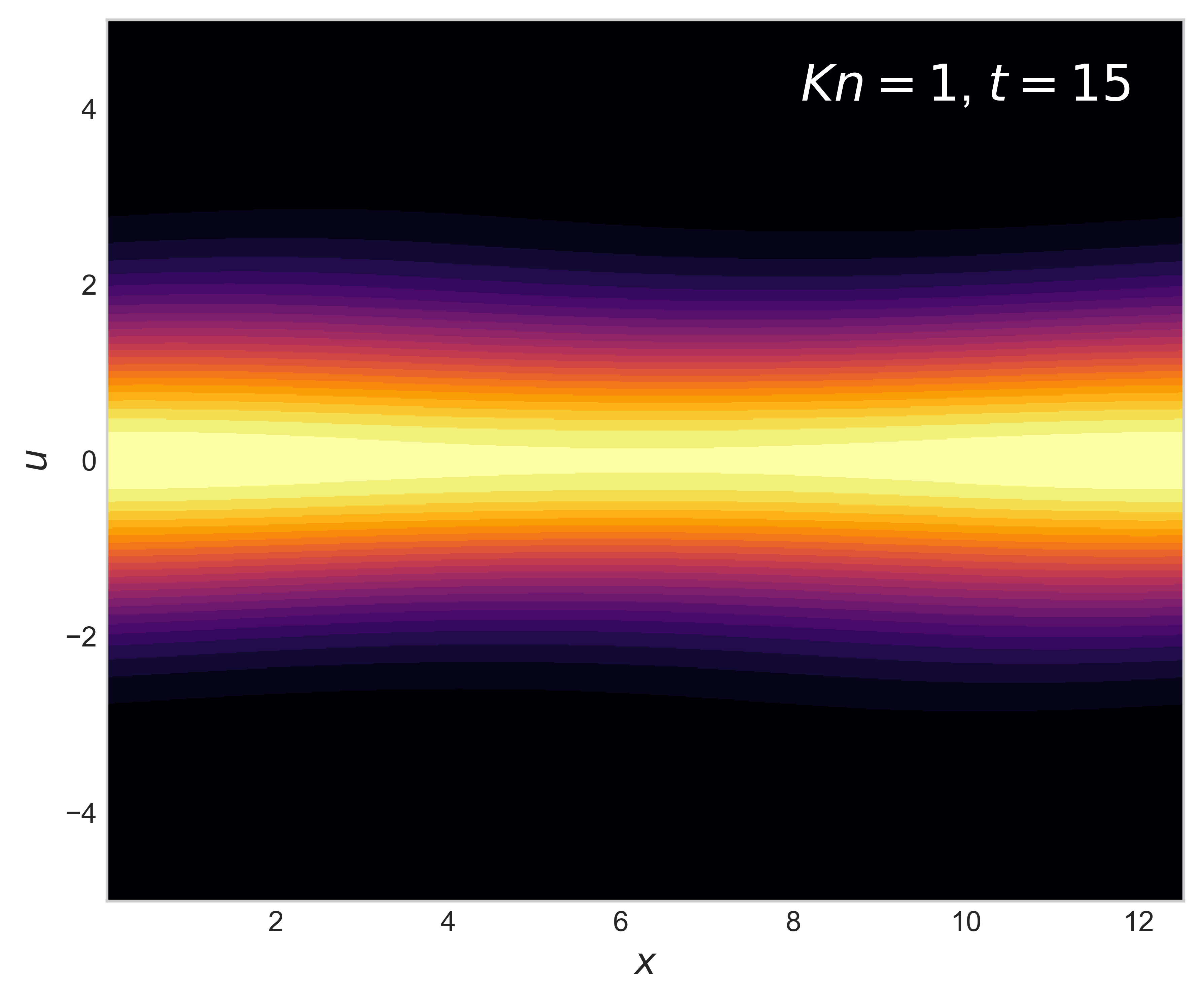}
    \end{subfigure}
    \hfill
    \begin{subfigure}[b]{0.325\textwidth}
    \centering
    \includegraphics[width=1.0\linewidth]{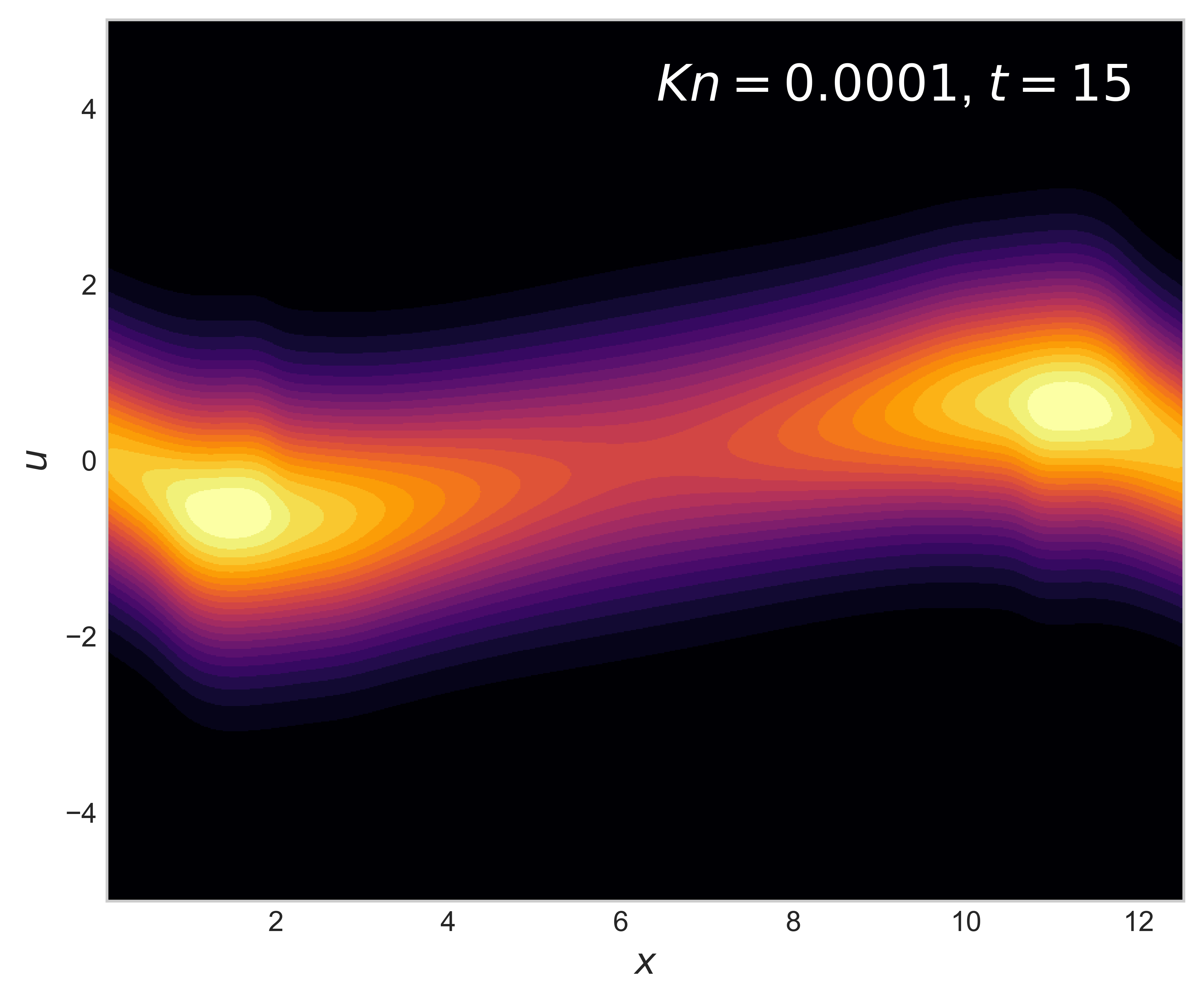}
    \end{subfigure}
    \caption{Nonlinear Landau damping with different Knudsen numbers by UGKS-RPE, from left to right are $Kn=\infty, 1, 10^{-4}$. $\lambda = 1, \Delta x = 0.1, \text{CFL}=0.9$. The Figure shows the phase diagram at different times.}
    \label{fig:nld_phase_diagram}
\end{figure}

\begin{figure}
    \centering
    \begin{subfigure}[b]{0.325\textwidth}
    \centering
    \includegraphics[width=1.0\linewidth]{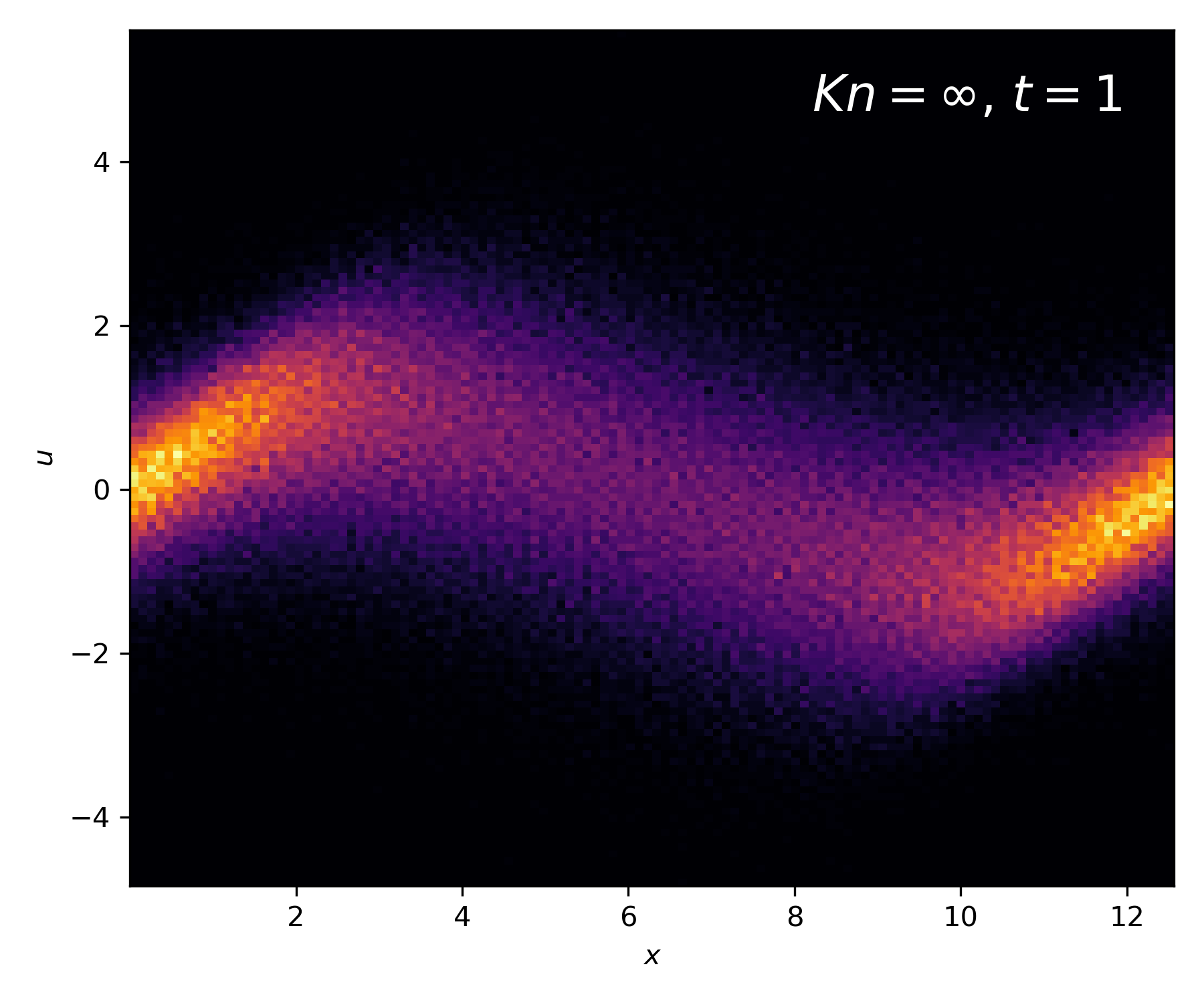}
    \end{subfigure}
    \hfill
    \begin{subfigure}[b]{0.325\textwidth}
    \centering
    \includegraphics[width=1.0\linewidth]{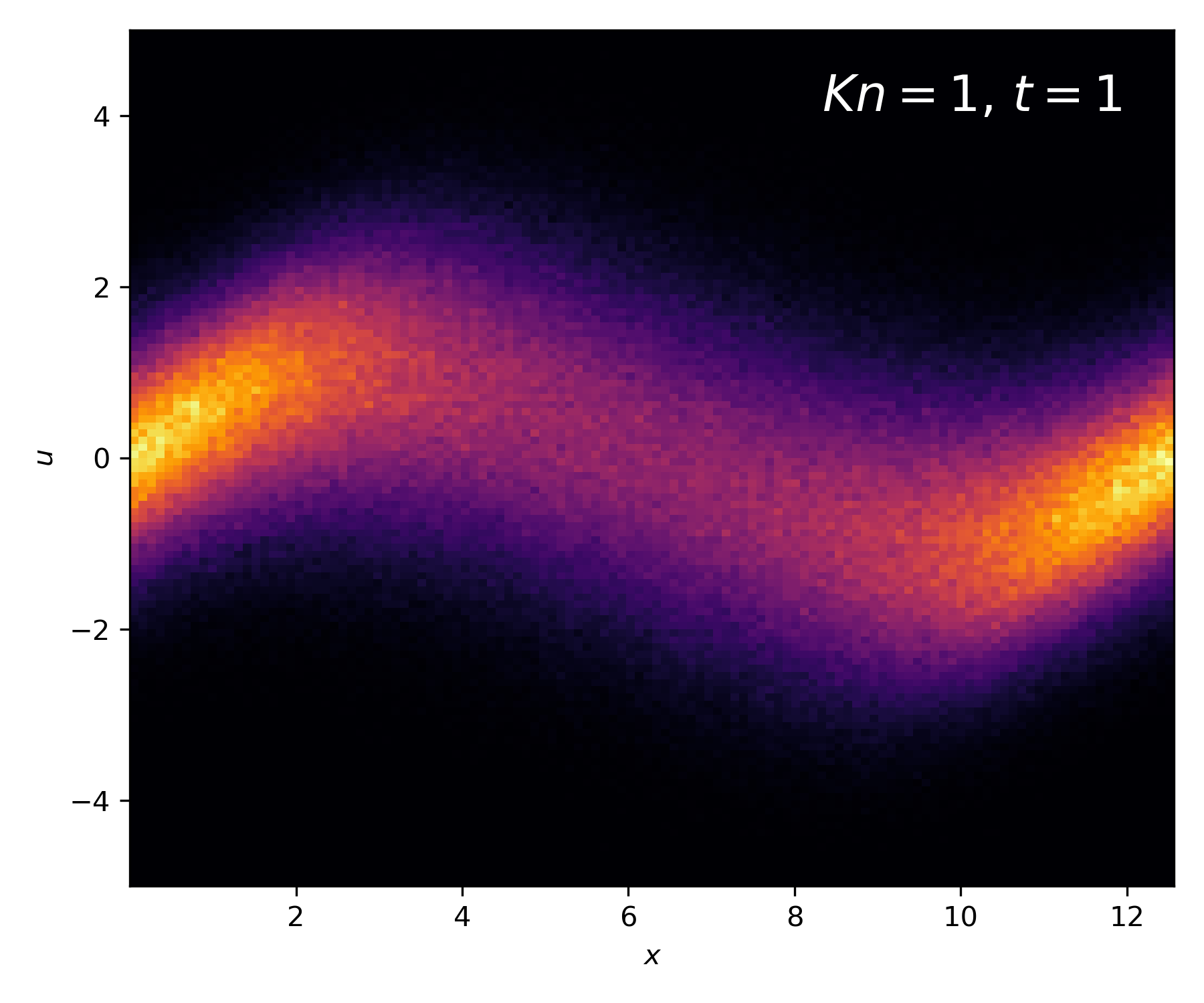}
    \end{subfigure}
    \hfill
    \begin{subfigure}[b]{0.325\textwidth}
    \centering
    \includegraphics[width=1.0\linewidth]{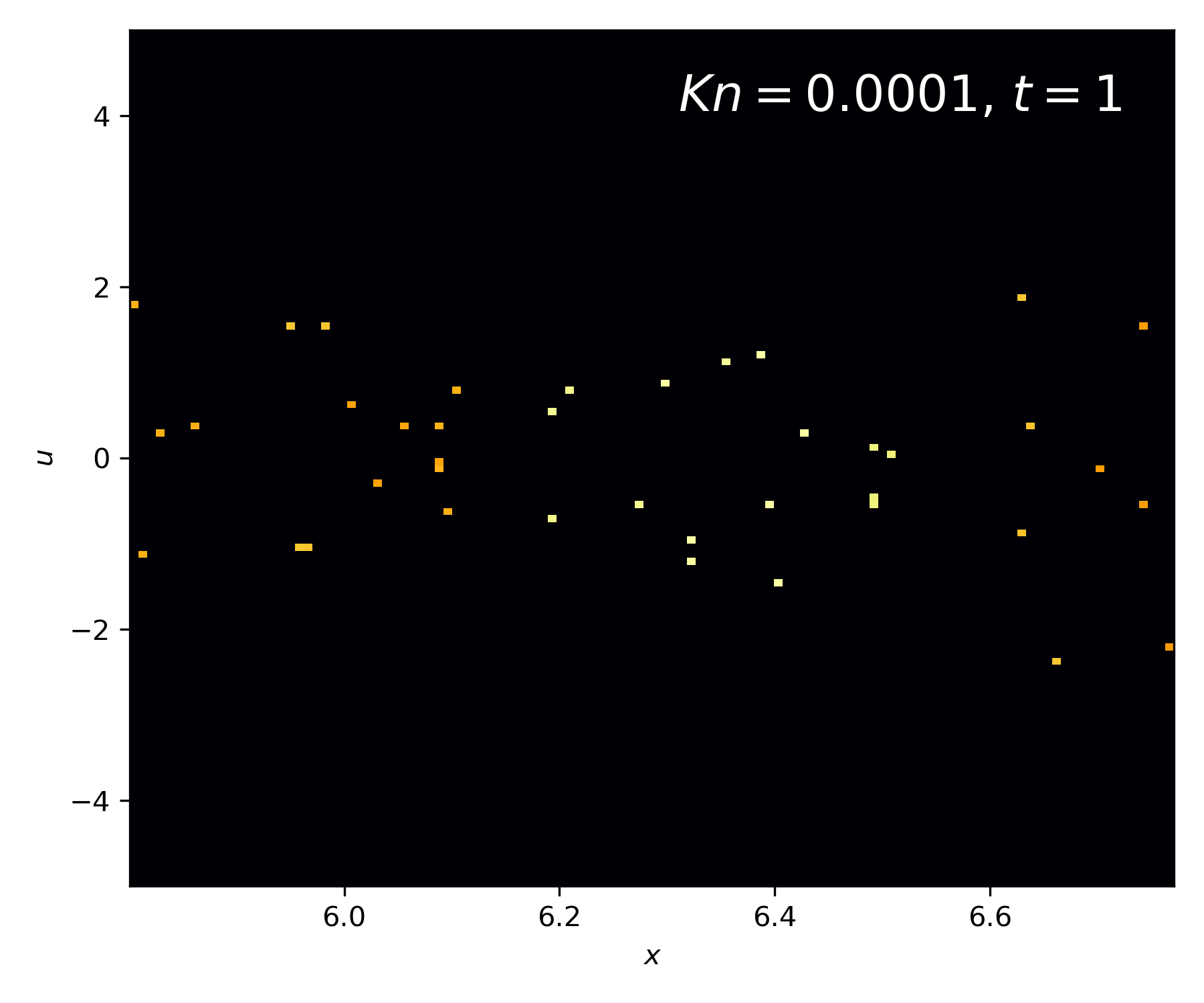}
    \end{subfigure}
    \vfill
    \centering
    \begin{subfigure}[b]{0.325\textwidth}
    \centering
    \includegraphics[width=1.0\linewidth]{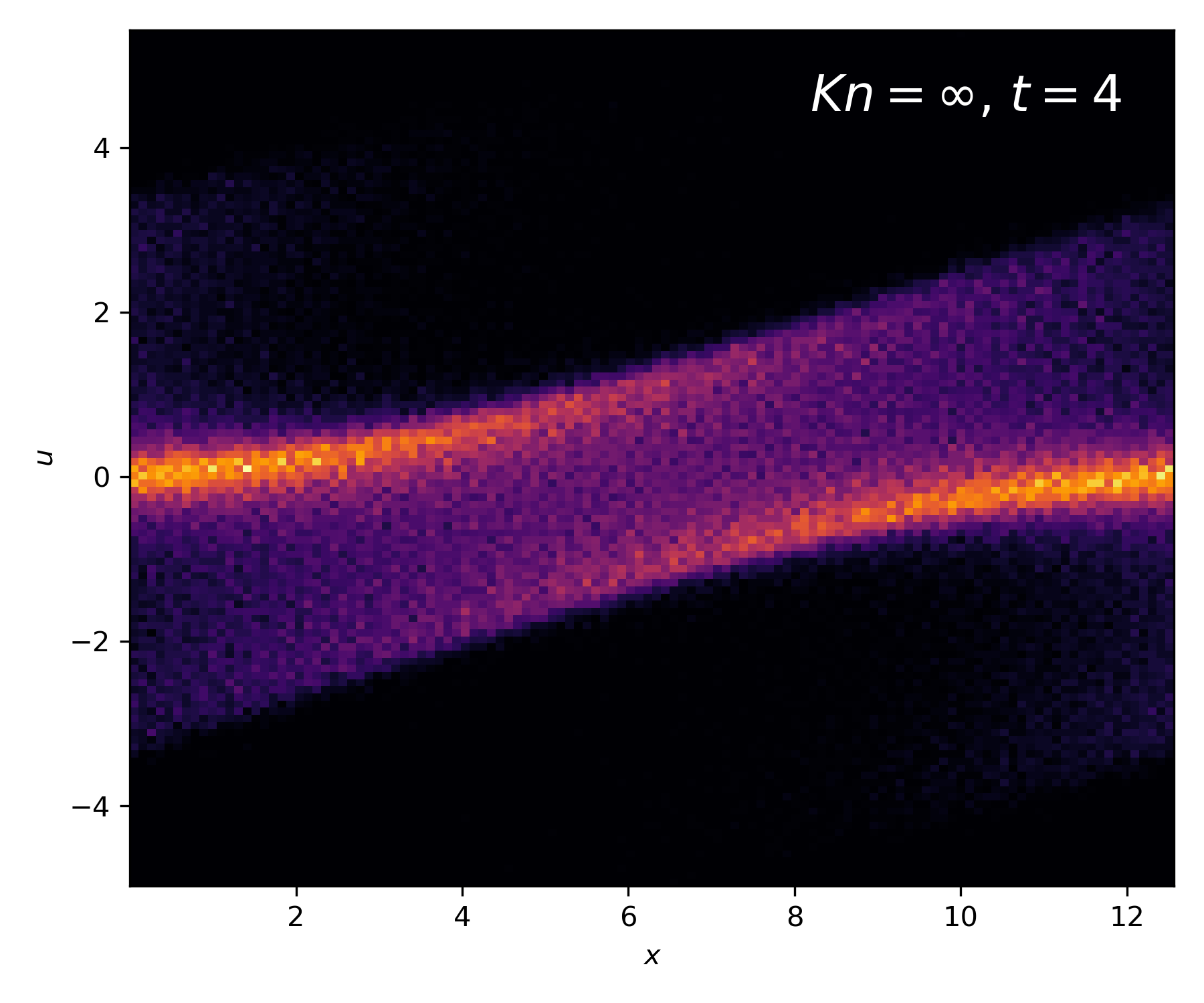}
    \end{subfigure}
    \hfill
    \begin{subfigure}[b]{0.325\textwidth}
    \centering
    \includegraphics[width=1.0\linewidth]{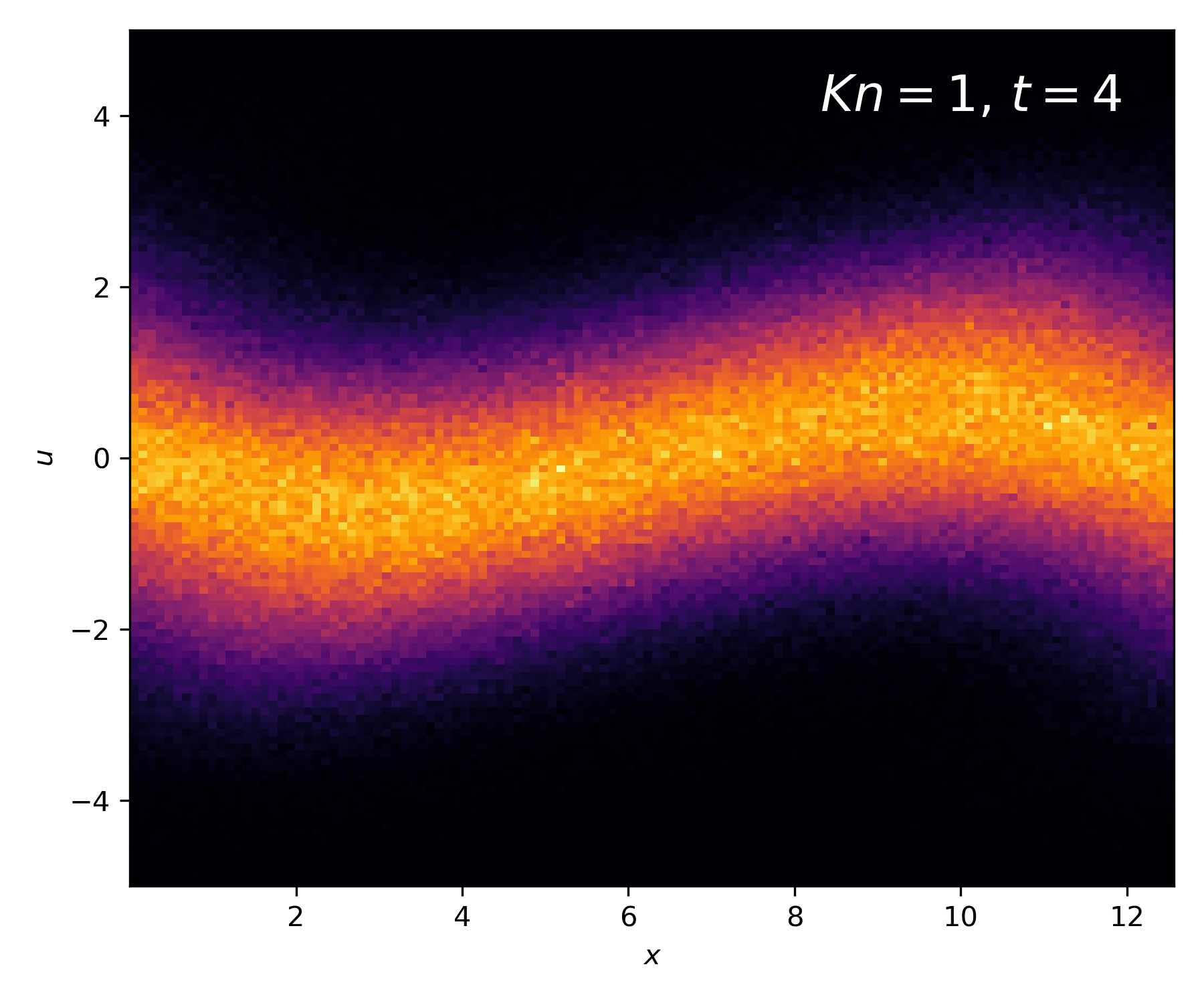}
    \end{subfigure}
    \hfill
    \begin{subfigure}[b]{0.325\textwidth}
    \centering
    \includegraphics[width=1.0\linewidth]{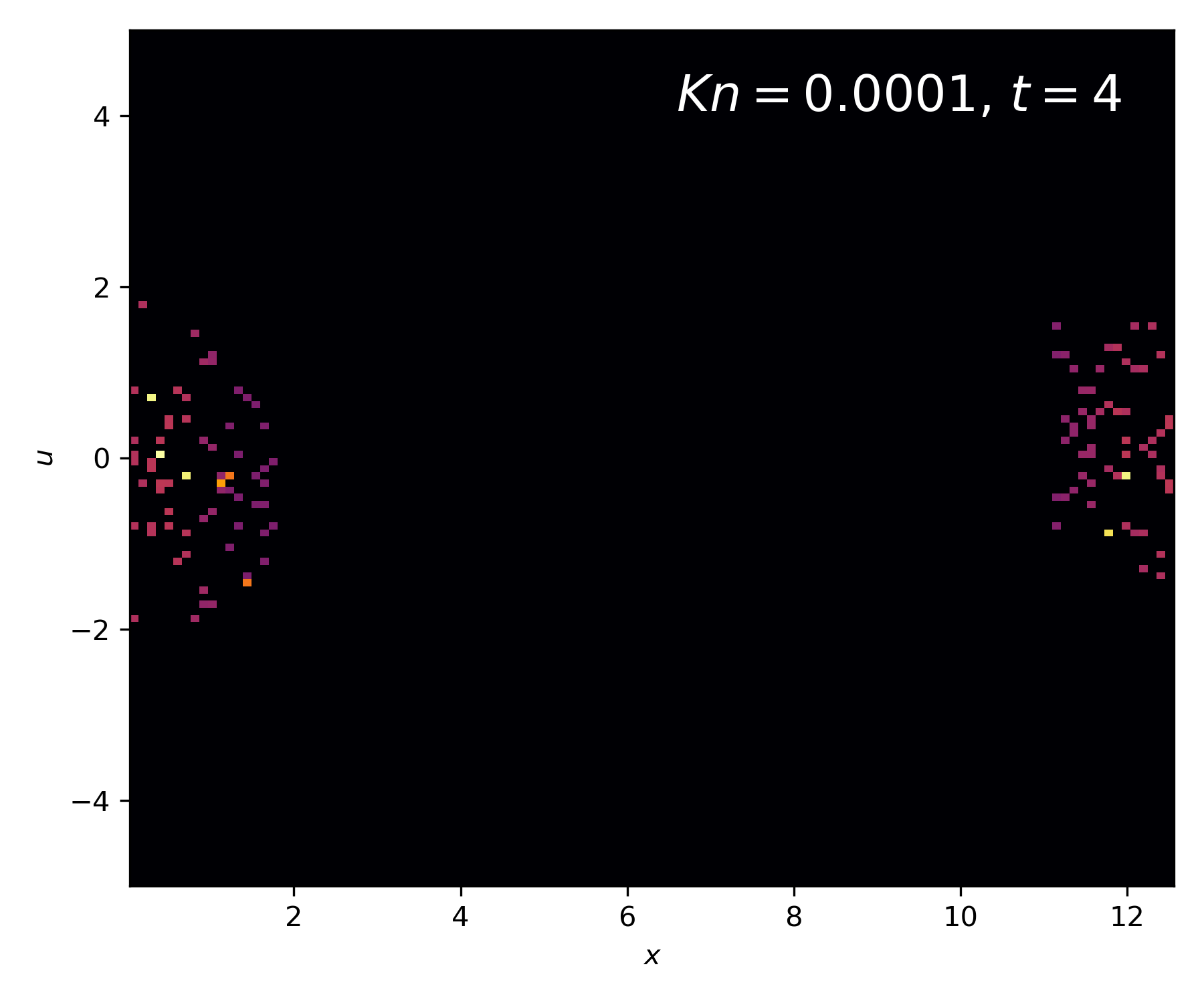}
    \end{subfigure}
    \vfill
    \centering
    \begin{subfigure}[b]{0.325\textwidth}
    \centering
    \includegraphics[width=1.0\linewidth]{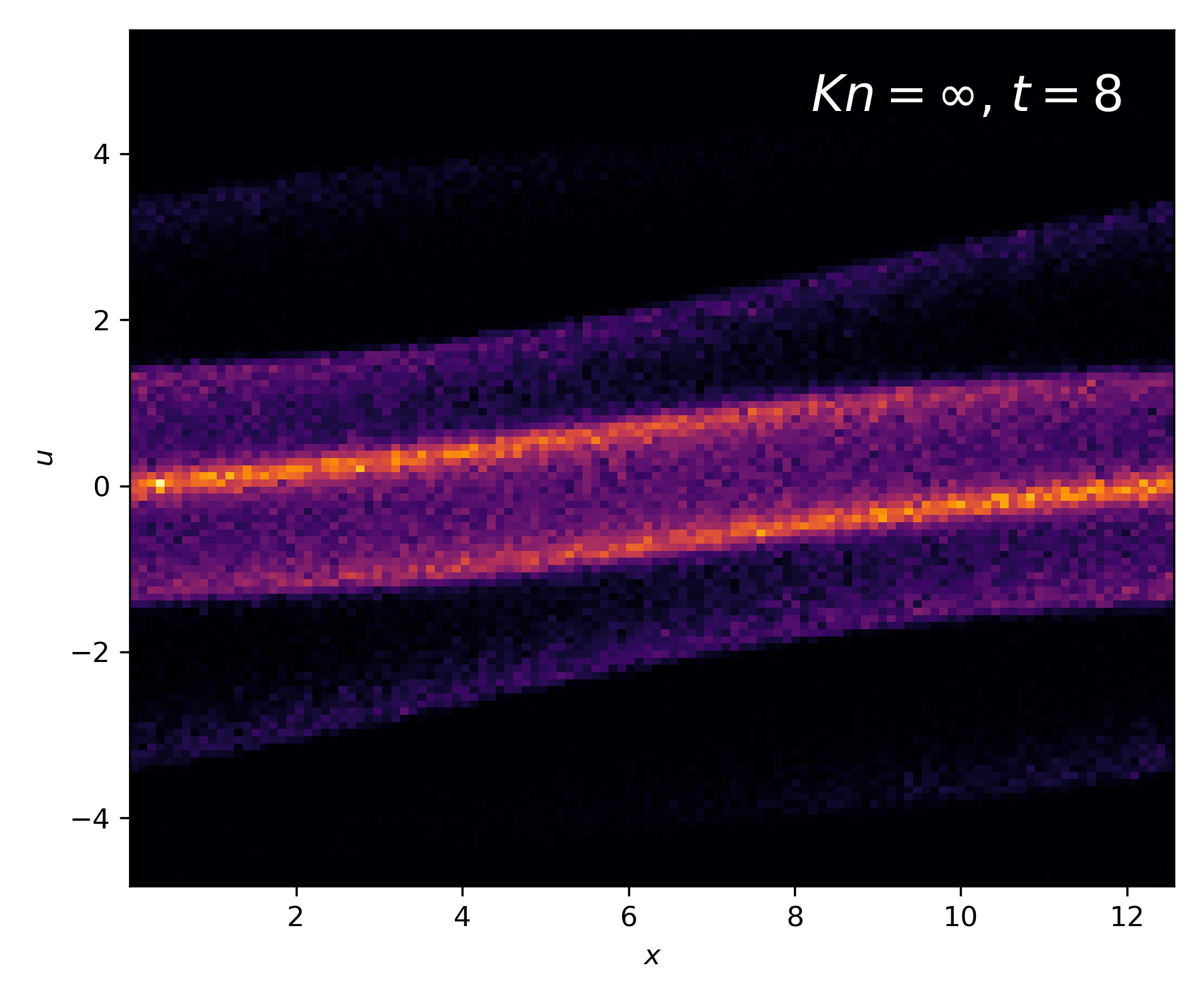}
    \end{subfigure}
    \hfill
    \begin{subfigure}[b]{0.325\textwidth}
    \centering
    \includegraphics[width=1.0\linewidth]{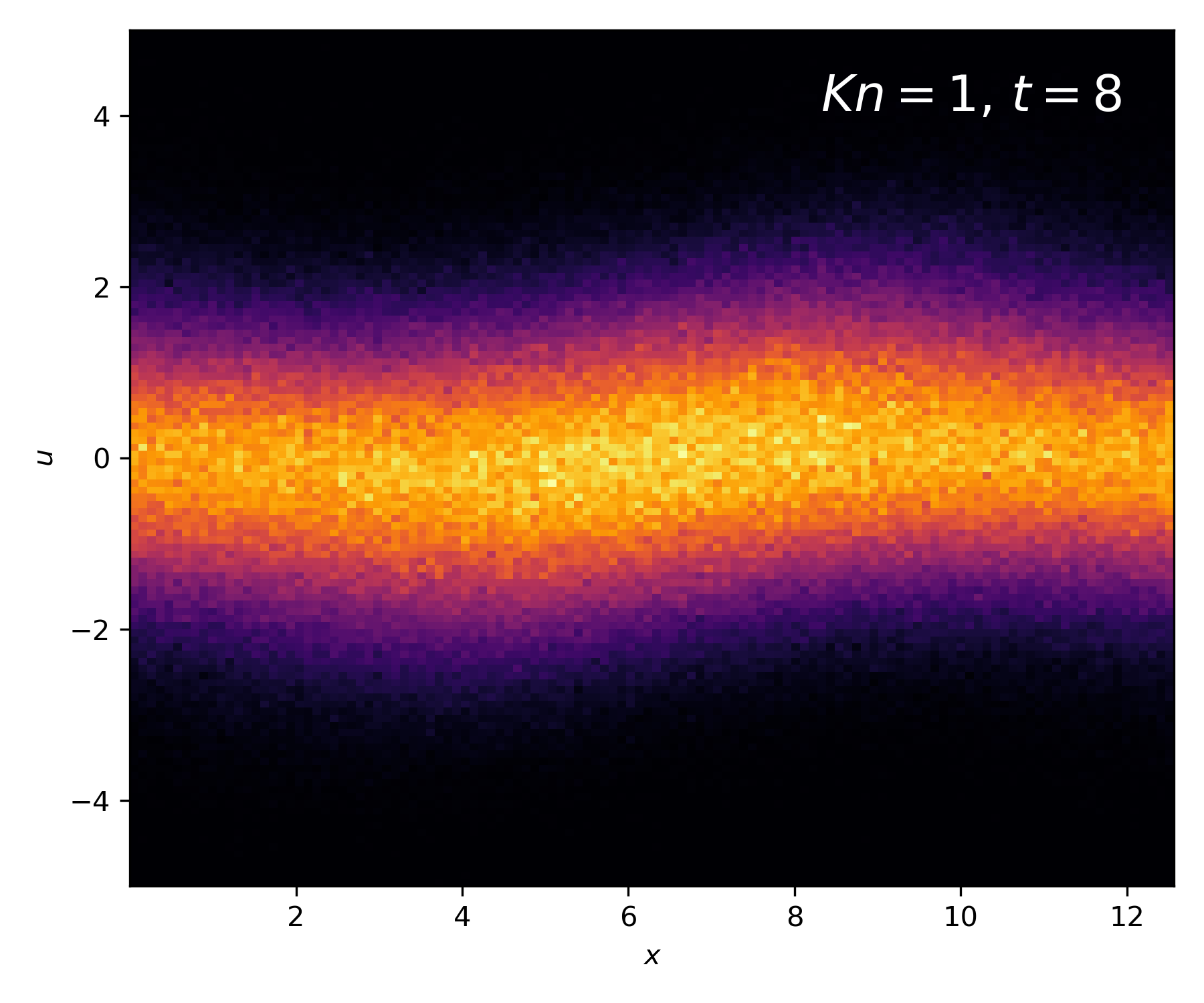}
    \end{subfigure}
    \hfill
    \begin{subfigure}[b]{0.325\textwidth}
    \centering
    \includegraphics[width=1.0\linewidth]{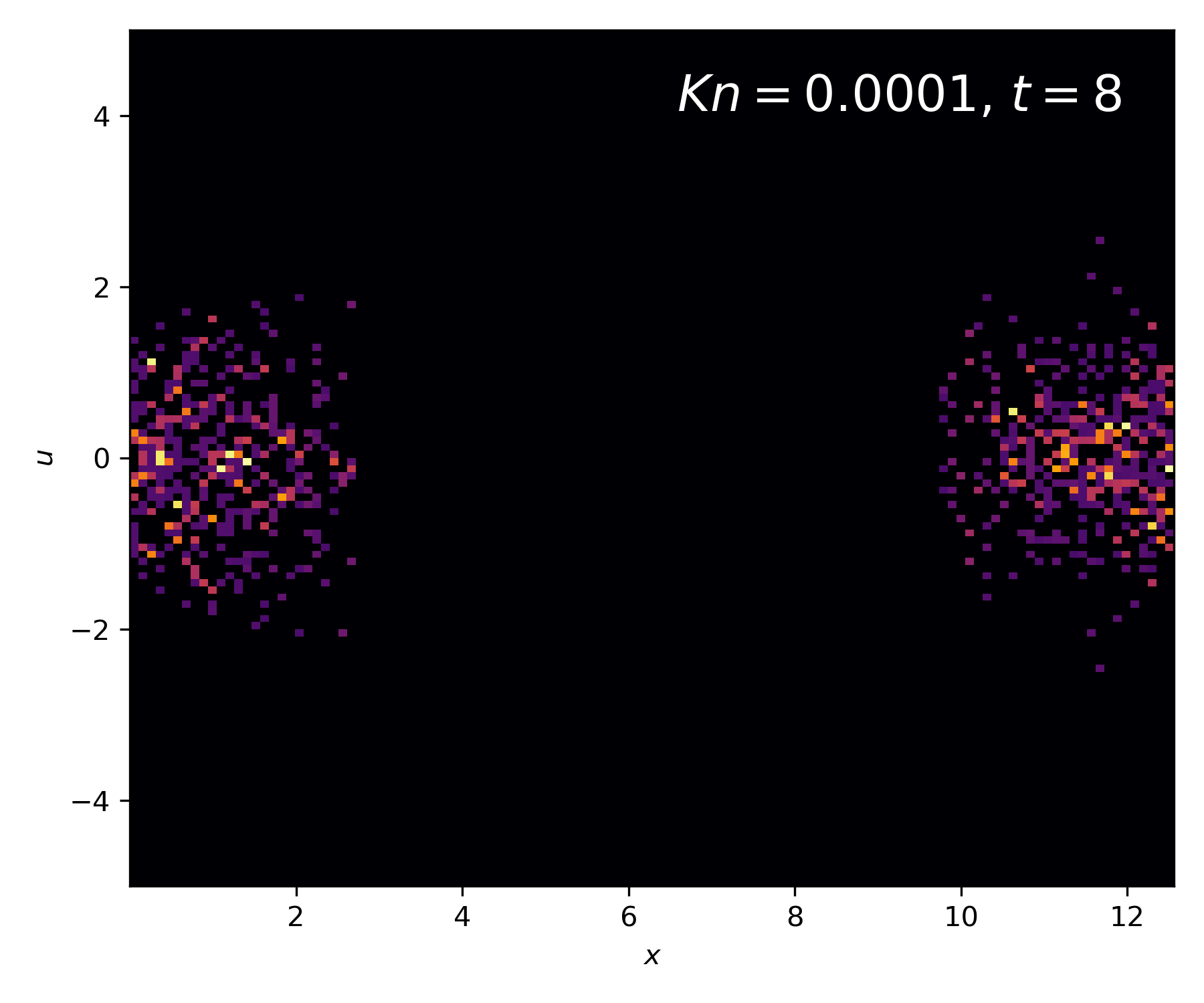}
    \end{subfigure}
    \vfill
    \centering
    \begin{subfigure}[b]{0.325\textwidth}
    \centering
    \includegraphics[width=1.0\linewidth]{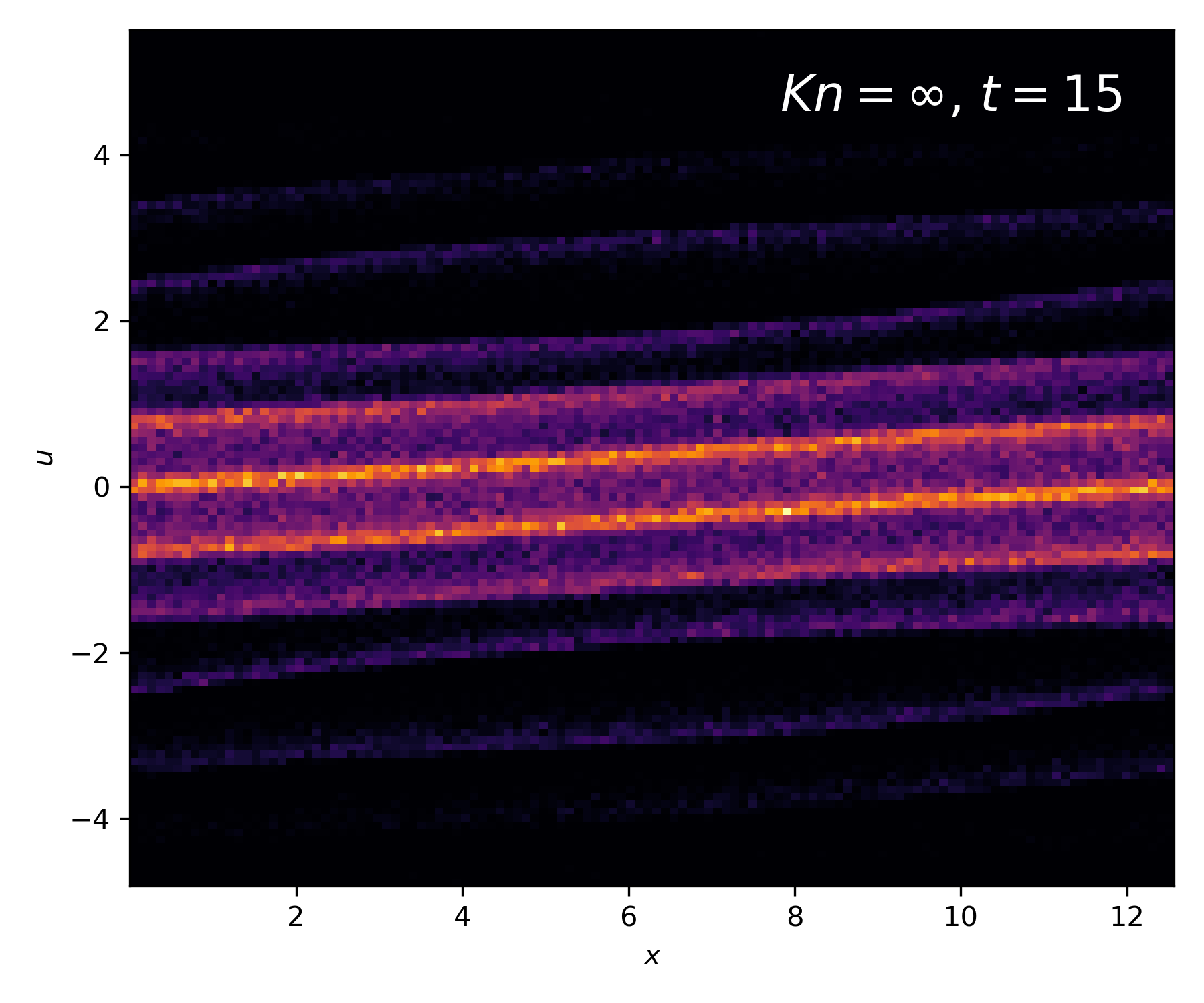}
    \end{subfigure}
    \hfill
    \begin{subfigure}[b]{0.325\textwidth}
    \centering
    \includegraphics[width=1.0\linewidth]{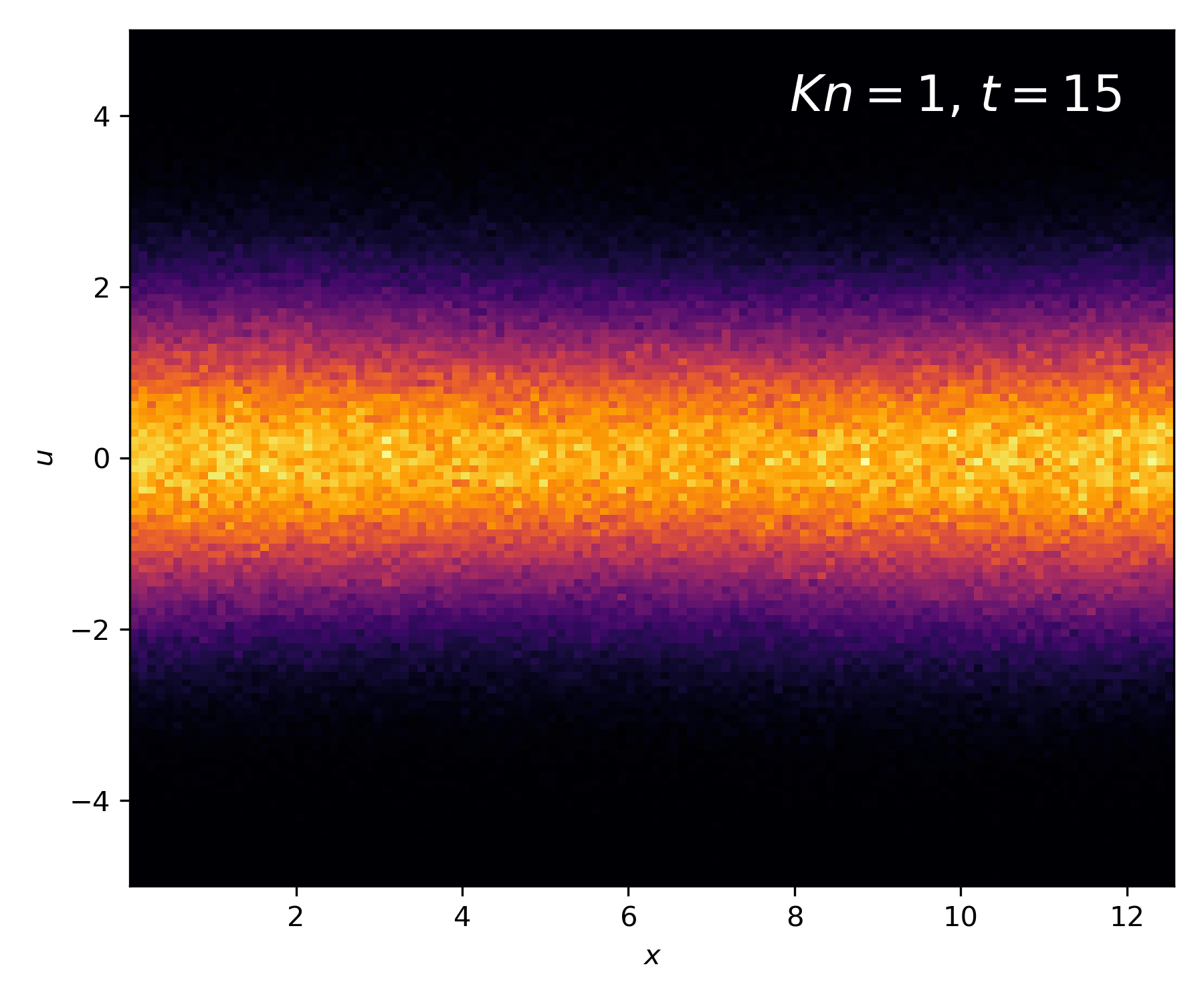}
    \end{subfigure}
    \hfill
    \begin{subfigure}[b]{0.325\textwidth}
    \centering
    \includegraphics[width=1.0\linewidth]{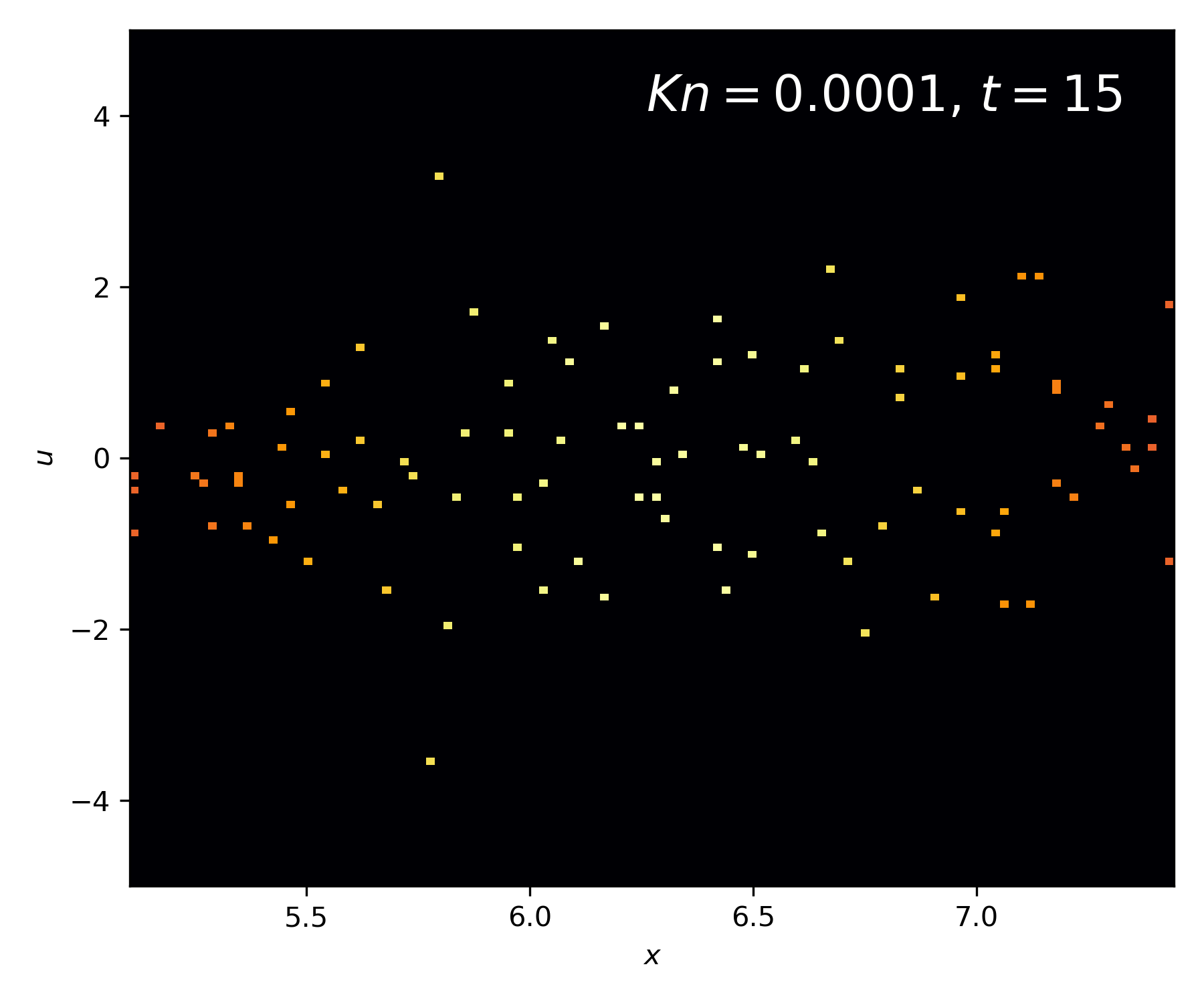}
    \end{subfigure}
    \caption{Nonlinear Landau damping with different Knudsen numbers by UGKWP-RPE, from left to right are $Kn=\infty, 1, 10^{-4}$. $\lambda = 1, \Delta x = 0.1, \text{CFL}=0.9$. The Figure shows the phase diagram at different times.}
    \label{fig:nld_phase_diagram-wp}
\end{figure}

Figure \ref{fig:nld-distribution} displays the velocity distribution at $x=L/2$ at various times as computed by the UGKS. At $t=1$, the distribution functions are approximately Maxwellian. For $Kn=0.0001$, the distribution exhibits the highest peaks, followed by $Kn=1$, and then $Kn=\infty$. At later times, the distribution in the collisionless regime ($Kn=\infty$) exhibits stronger wave-particle interactions, deviating from equilibrium. However, the distributions for $Kn=1$ and $Kn=0.0001$ appear to remain closer to equilibrium states compared to the collisionless case.

\begin{figure}
    \centering
    \begin{subfigure}[b]{0.48\textwidth}
    \centering
    \includegraphics[width=1.0\linewidth]{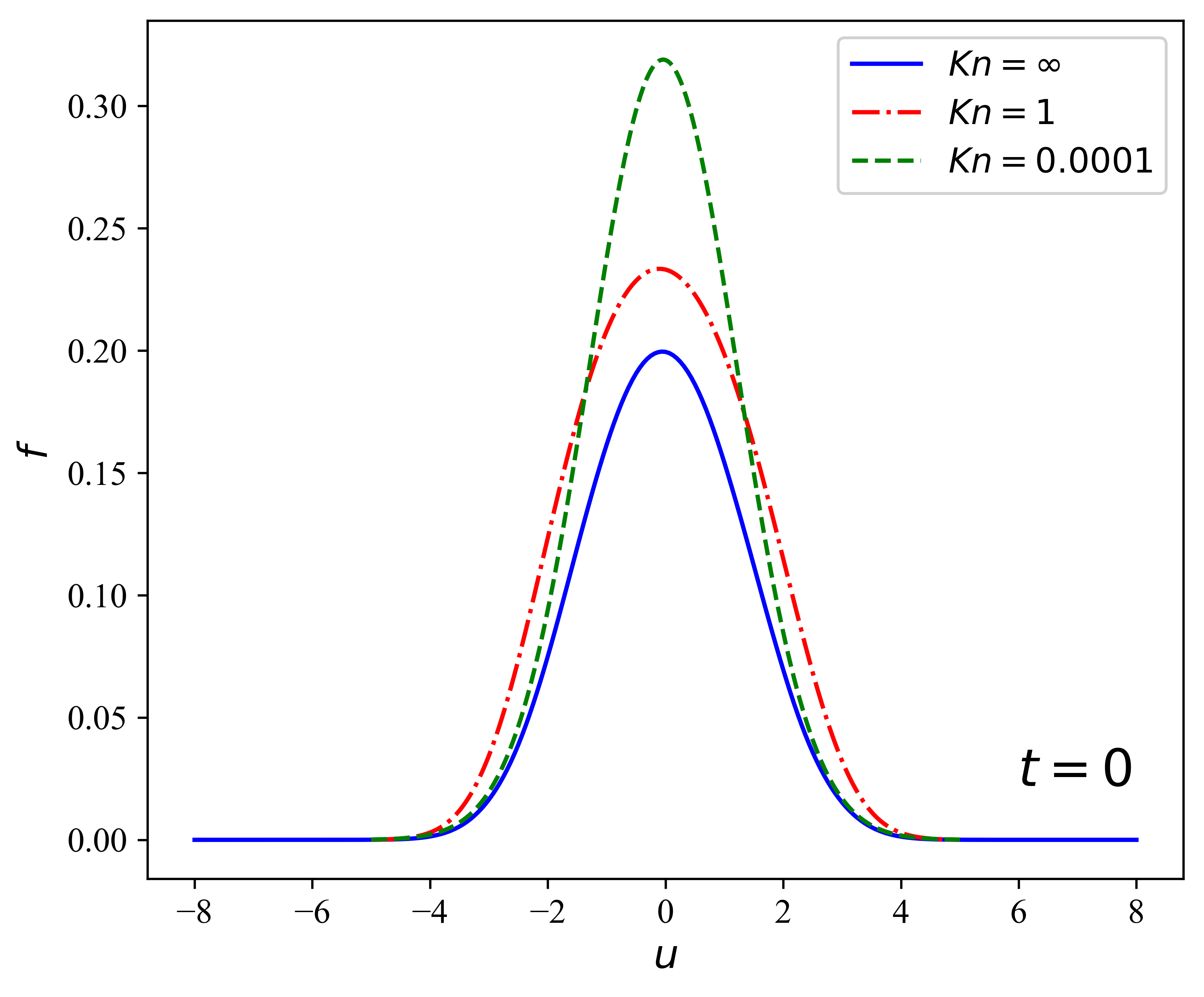}
    \end{subfigure}
    \hfill
    \begin{subfigure}[b]{0.48\textwidth}
    \centering
    \includegraphics[width=1.0\linewidth]{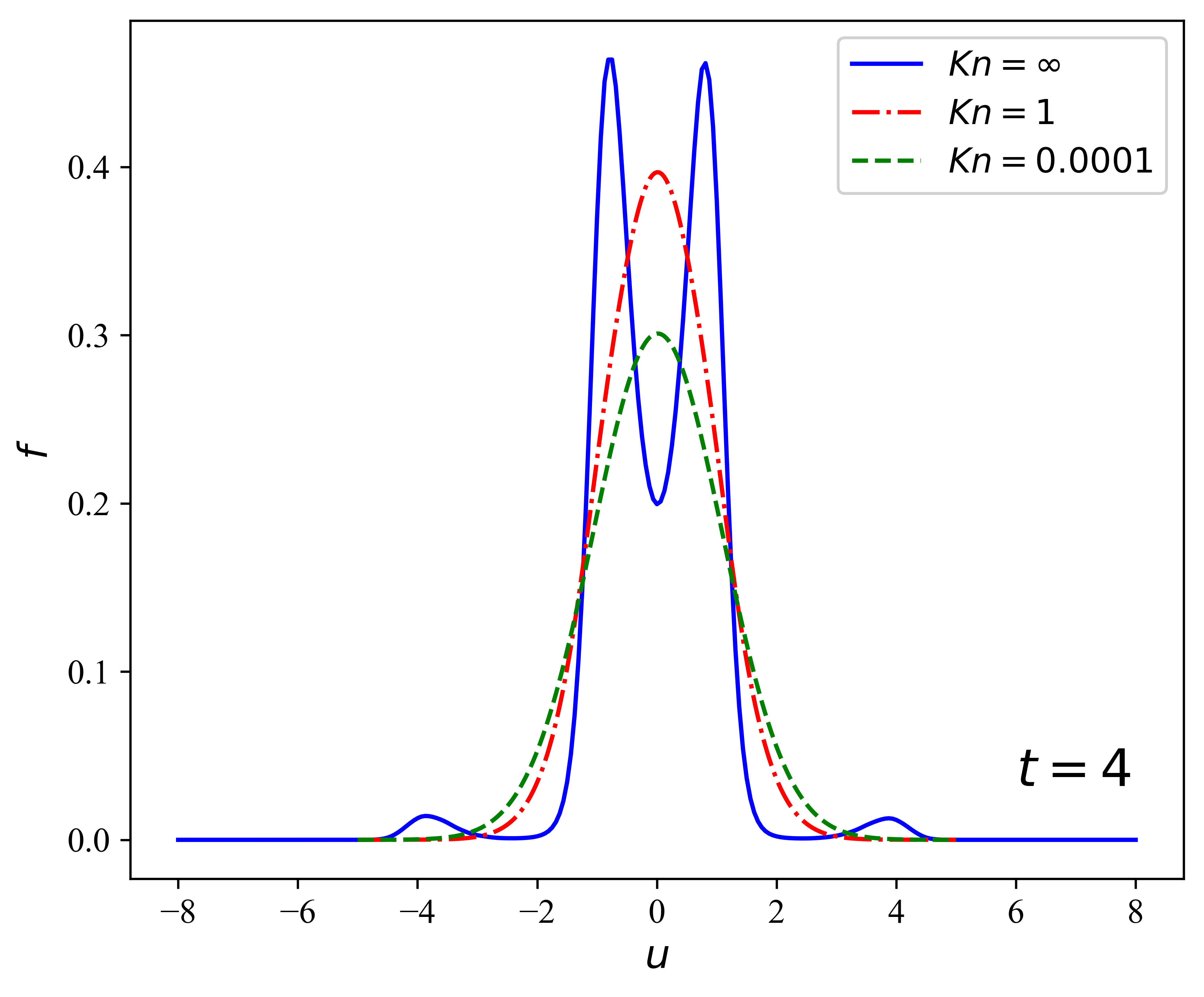}
    \end{subfigure}
    \vfill
    \begin{subfigure}[b]{0.48\textwidth}
    \centering
    \includegraphics[width=1.0\linewidth]{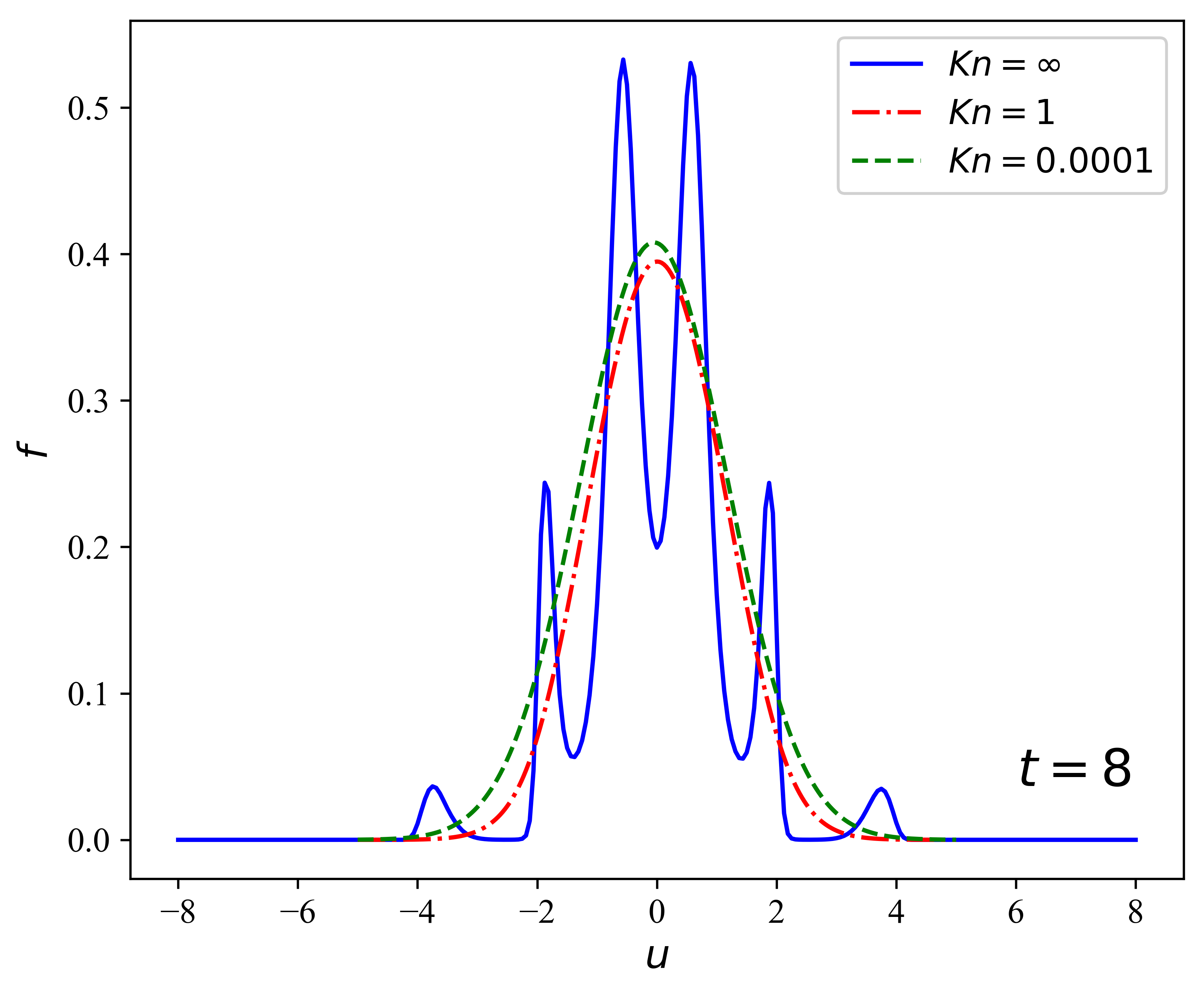}
    \end{subfigure}
    \hfill
    \begin{subfigure}[b]{0.48\textwidth}
    \centering
    \includegraphics[width=1.0\linewidth]{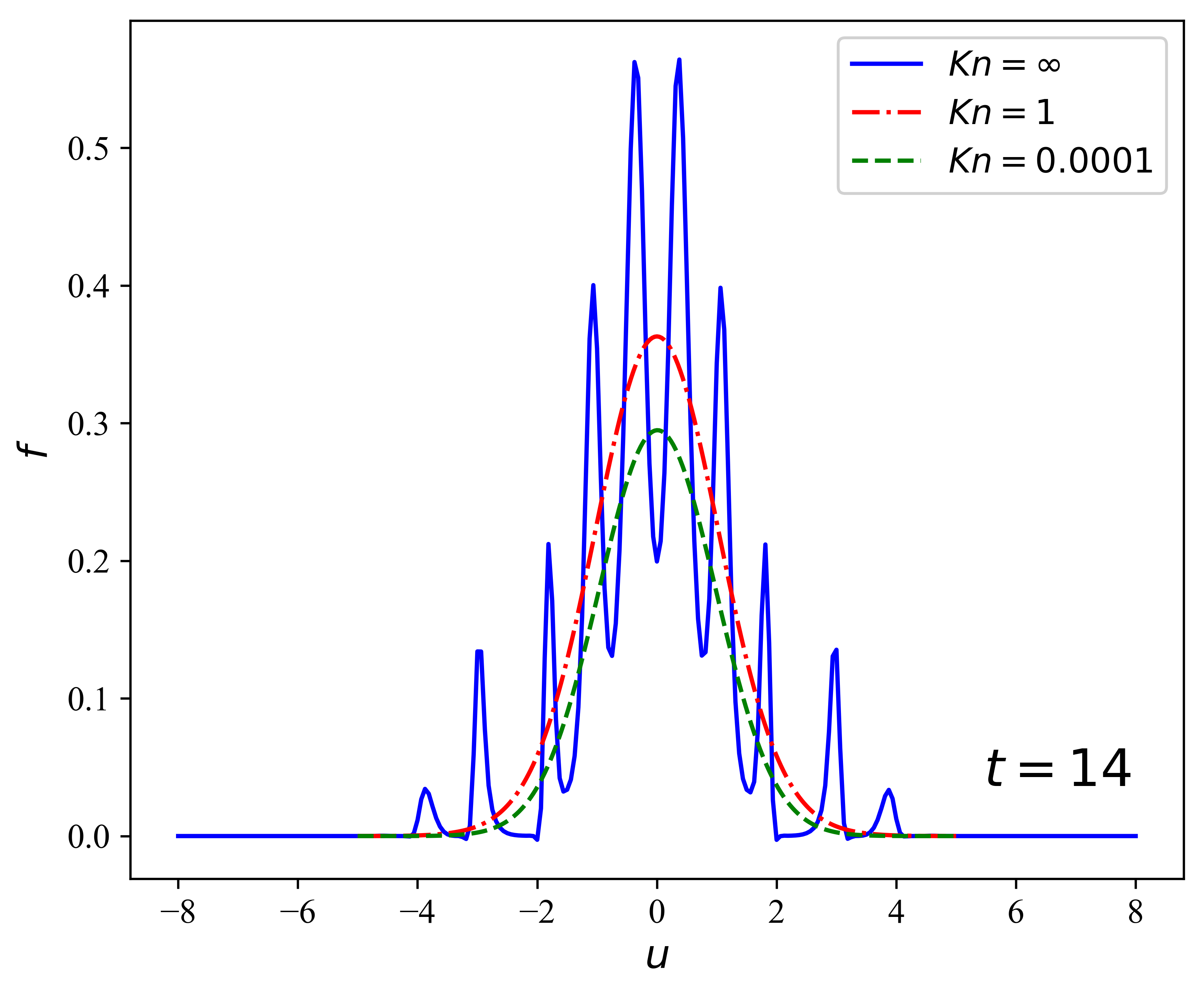}
    \end{subfigure}
    \caption{Nonlinear Landau damping with different Knudsen numbers by UGKS-RPE: $\lambda = 1, \Delta x = 0.1, \text{CFL}=0.9$. The figure shows distribution at velocity space $u$ for different time at $x=L/2$, from top-left to bottom-right are t = 1, 4, 8, 15.}
    \label{fig:nld-distribution}
\end{figure}

\subsection{Bump on tail instability}
\label{sec:bti}

This section investigates the bump-on-tail instability to assess the numerical scheme's performance. We maintain a constant Knudsen number of $Kn = \infty$ and systematically reduce the Debye length, $\lambda$, from $1$ to $10^{-6}$. This variation spans the electrostatic and quasineutral regimes, enabling us to evaluate the scheme's ability to accurately resolve the relevant physics in the quasineutral limit without being constrained by the Debye length.

The initial distribution function for the bump-on-tail instability is defined as:
$$
f_0(x,u) = f_p(u)(1+\alpha \cos(kx)),
$$
where the perturbation parameter is $\alpha = 0.04(0.01+0.99\lambda)$, and the wavenumber is $k=0.3$. The function $f_p(u)$ represents a Maxwellian distribution with a bump on its tail, expressed as:
$$
f_p(u) = n_pe^{-u^2/2} + n_b e^{-(u-U)^2/2RT}
$$
Here, $n_p=\frac{9}{10\sqrt{2\pi}}$ and $n_b=\frac{2}{10\sqrt{2\pi}}$ represent the number densities of the plasma and bump components, respectively. The macroscopic velocity of the bump is $U=4.5$, and its thermal energy is characterized by $RT=0.25$. The spatial domain extends over $x\in[0, L]$, where $L=2\pi/k$ with 256 discretized mesh grids. For UGKS, the velocity space is defined as $u\in[-6, 9]$ for $\lambda \geq 10^{-3}$ and expanded to $u\in[-12,12]$ for $\lambda < 10^{-3}$ with 256 velocity points. For UGKWP, 1000 simulation particles are used in each cell. For $\lambda\leq10^{-3}$, we compute the electrical energy as $E_p = \frac{1}{2} \int\left|\nabla\phi\right|^2 d x$.

Figure \ref{fig:bti-kninf-energy} presents the temporal evolution of the electrical energy and the corresponding timestep for the resolved case with $\lambda=1$, $\Delta x = 0.08$, and $\text{CFL}=0.9$. The left panel shows that the evolution of electrical energy obtained with the UGKS-RPE is in good agreement with that of the UGKWP-RPE method. The electrical energy initially decreases, followed by an increase attributed to the onset of the instability. The right panel depicts the timestep evolution, revealing that UGKS employs a smaller timestep compared to UGKWP-RPE. This difference arises from the larger velocity space used by UGKS-RPE, which, at the same CFL number, yields a smaller time step. Figure \ref{fig:bti-phase} displays the phase-space diagrams at time $t=40$. The phase-space distribution obtained with UGKS-RPE (left panel) is in good agreement with the UGKWP-RPE results (right panel). The instability is clearly shown in the phase space at this time.

\begin{figure}
\begin{subfigure}[b]{0.48\textwidth}
\centering
    \includegraphics[width=1.0\textwidth]{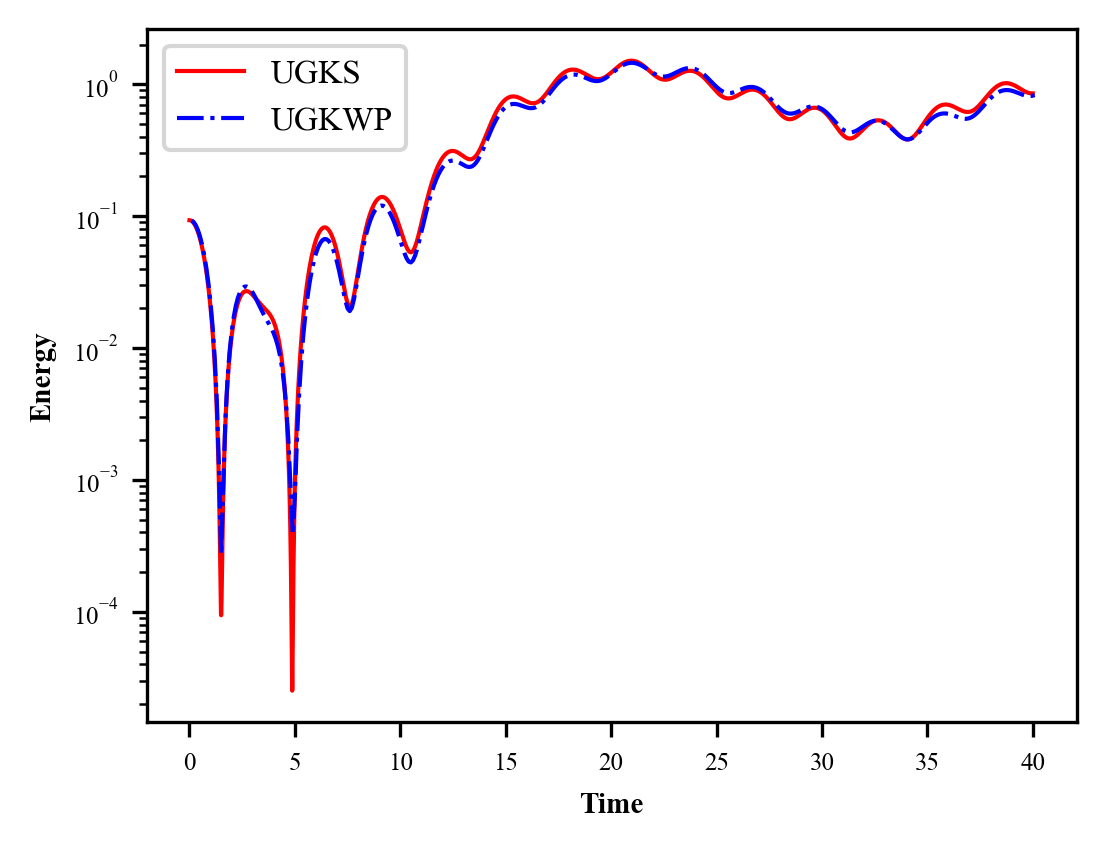}
\caption{}
\label{fig:lldk0.3}
\end{subfigure}
\begin{subfigure}[b]{0.48\textwidth}
\centering
    \includegraphics[width=1.0\textwidth]{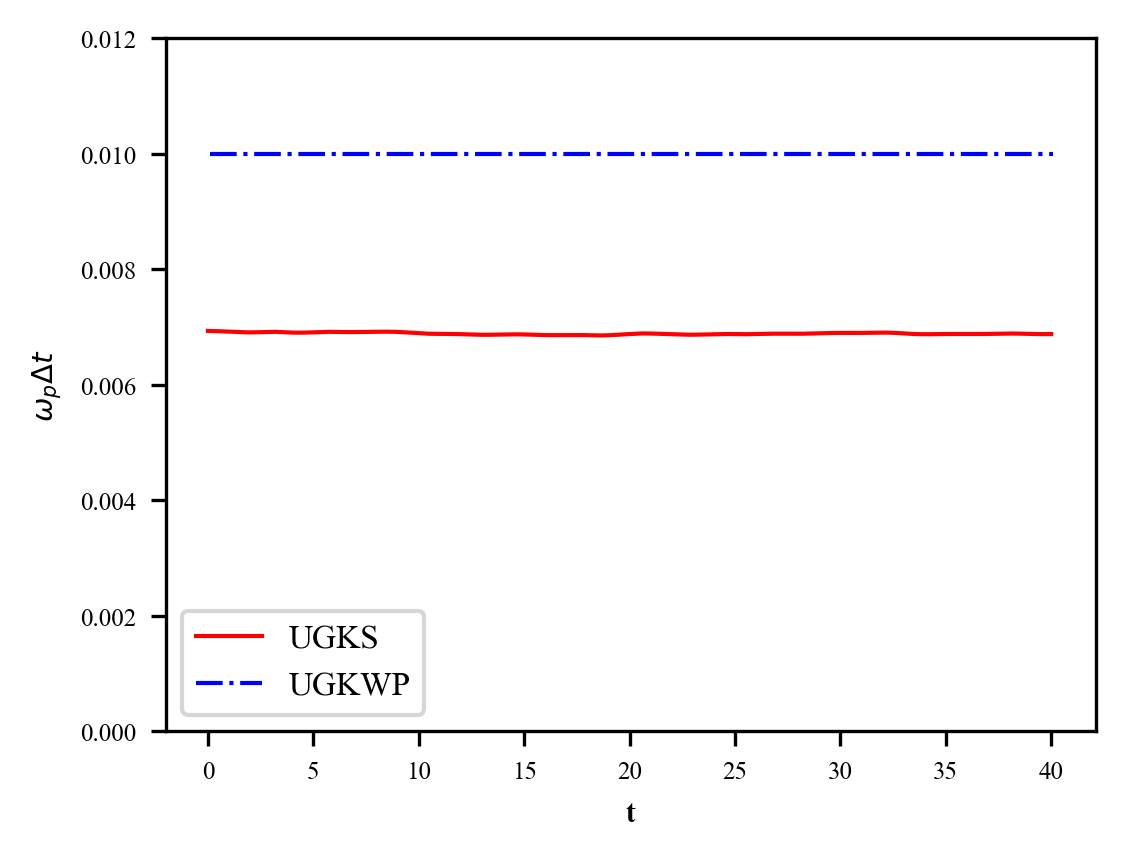}
\caption{}
\label{fig:lldk1}
\end{subfigure}
\caption{Bump on tail instability by UGKS-RPE and UGKWP-PRE, resolved case: $\lambda = 1, \Delta x = 0.08, \text{CFL} = 0.9$ at t=40. Left: time evolution of electrostatic energy. Right: time evolution of timestep. }
\label{fig:bti-kninf-energy}
\end{figure}

\begin{figure}
\begin{subfigure}[b]{0.48\textwidth}
\centering
    \includegraphics[width=1.0\textwidth]{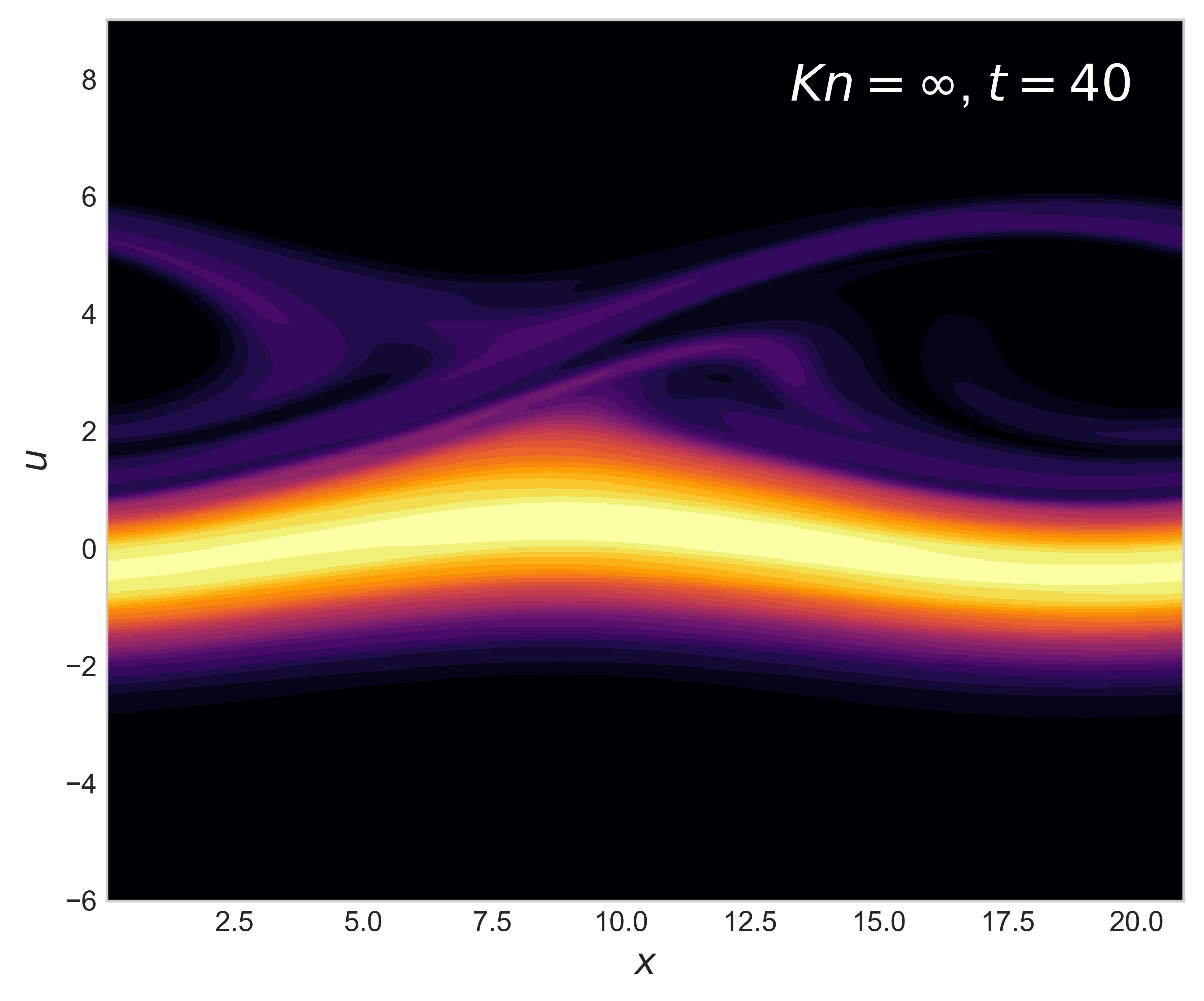}
\caption{}
\label{fig:lldk0.3}
\end{subfigure}
\begin{subfigure}[b]{0.48\textwidth}
\centering
    \includegraphics[width=1.0\textwidth]{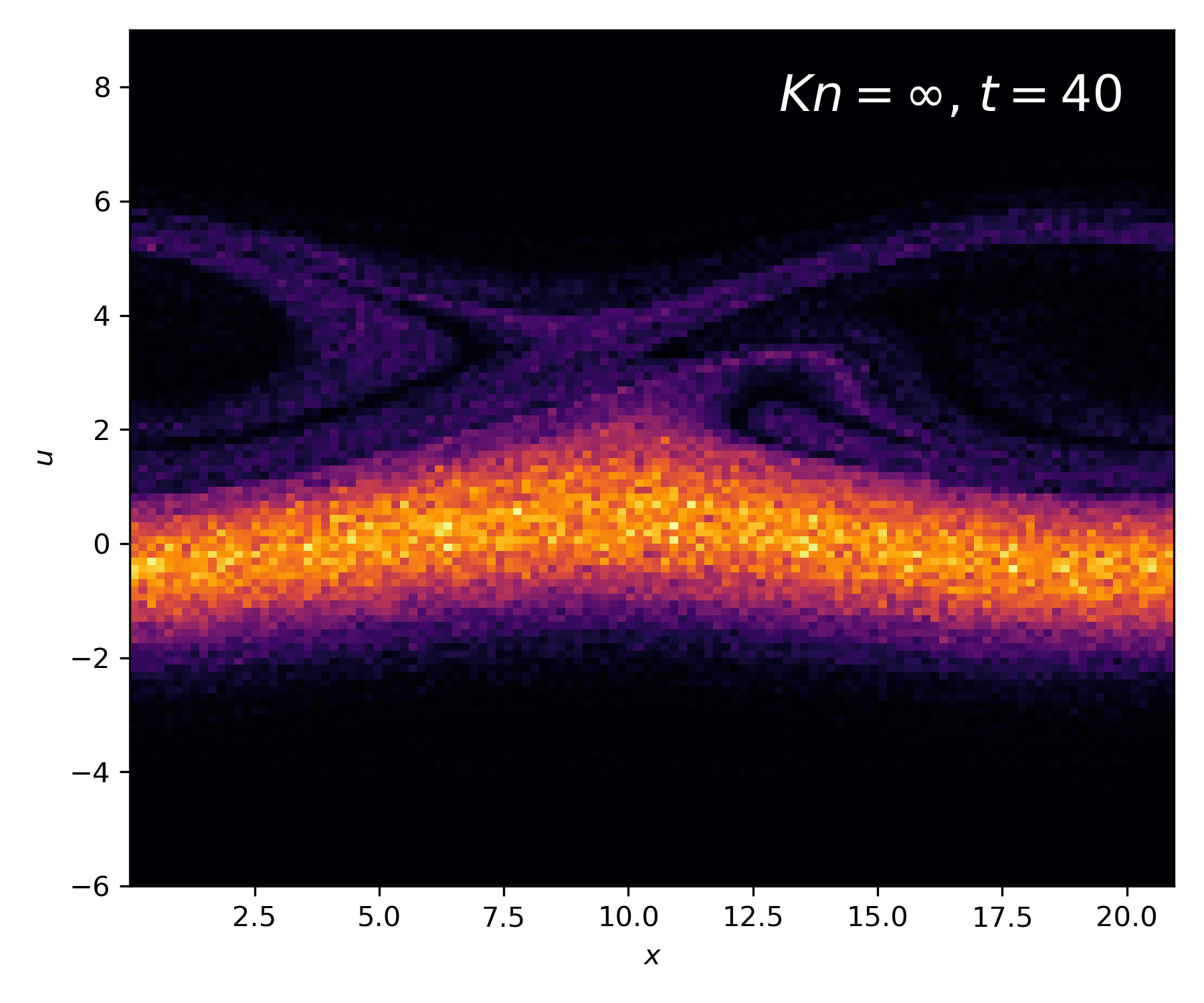}
\caption{}
\label{fig:lldk1}
\end{subfigure}
\caption{Bump on tail instability, resolved case: $\lambda = 1, \Delta x = 0.08, \text{CFL} = 0.9$ at t=40. Left: phase plot of the distribution function predicted by UGKS-RPE. Right: phase space plot of the distribution function predicted by UGKWP-RPE. }
\label{fig:bti-phase}
\end{figure}

Figure \ref{fig:btiqnugks} presents the simulation results for the bump-on-tail instability in the quasineutral regime by UGKS-RPE, with Debye lengths ranging from $\lambda=10^{-3}$ down to $\lambda=0$. The left panel shows the evolution of electrostatic energy. Notably, the electrostatic energy exhibits a consistent trend as the Debye length decreases. It initially decreases, then increases, and finally decays toward a stable state, indicating the disappearance of the instability. The right panel shows the timestep evolution, demonstrating that the timestep is not limited by the Debye length. The smaller timesteps observed for $\lambda=10^{-6}$ and $\lambda=0$ compared to $\lambda=10^{-3}$ arise from the use of a larger velocity space, $[-12,12]$, for the former cases, versus $[-6,9]$ for the latter.
Figure \ref{fig:bti-phaseqn} illustrates the phase-space evolution of UGKS-RPE at $\lambda=10^{-3}$ at times $t=0, 2, 3,$ and $15$. At $t=2$, instability is evident, but it gradually decreases by $t=5$. By $t=15$, the phase diagram closely resembles the initial state at $t=0$, indicating that the instability has been effectively smoothed out.

\begin{figure}
\begin{subfigure}[b]{0.48\textwidth}
\centering
    \includegraphics[width=1.0\textwidth]{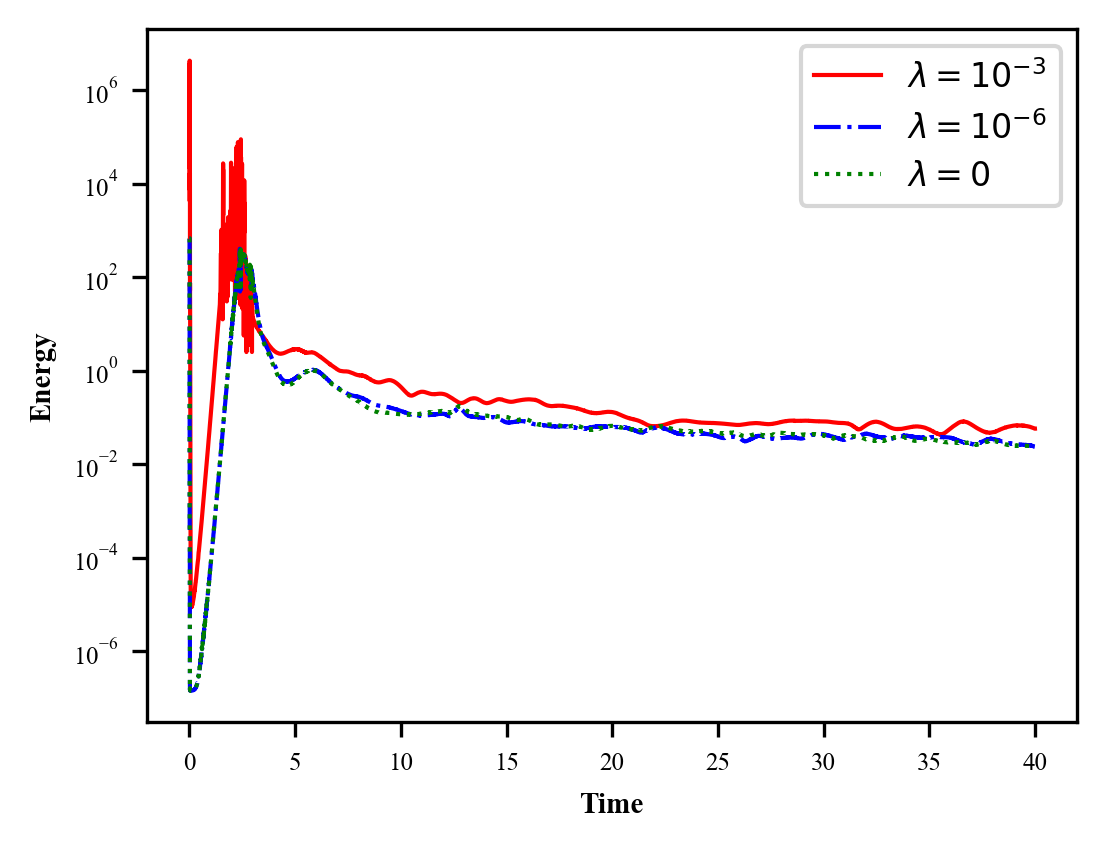}
\caption{}
\label{fig:lldk0.3}
\end{subfigure}
\begin{subfigure}[b]{0.48\textwidth}
\centering
    \includegraphics[width=1.0\textwidth]{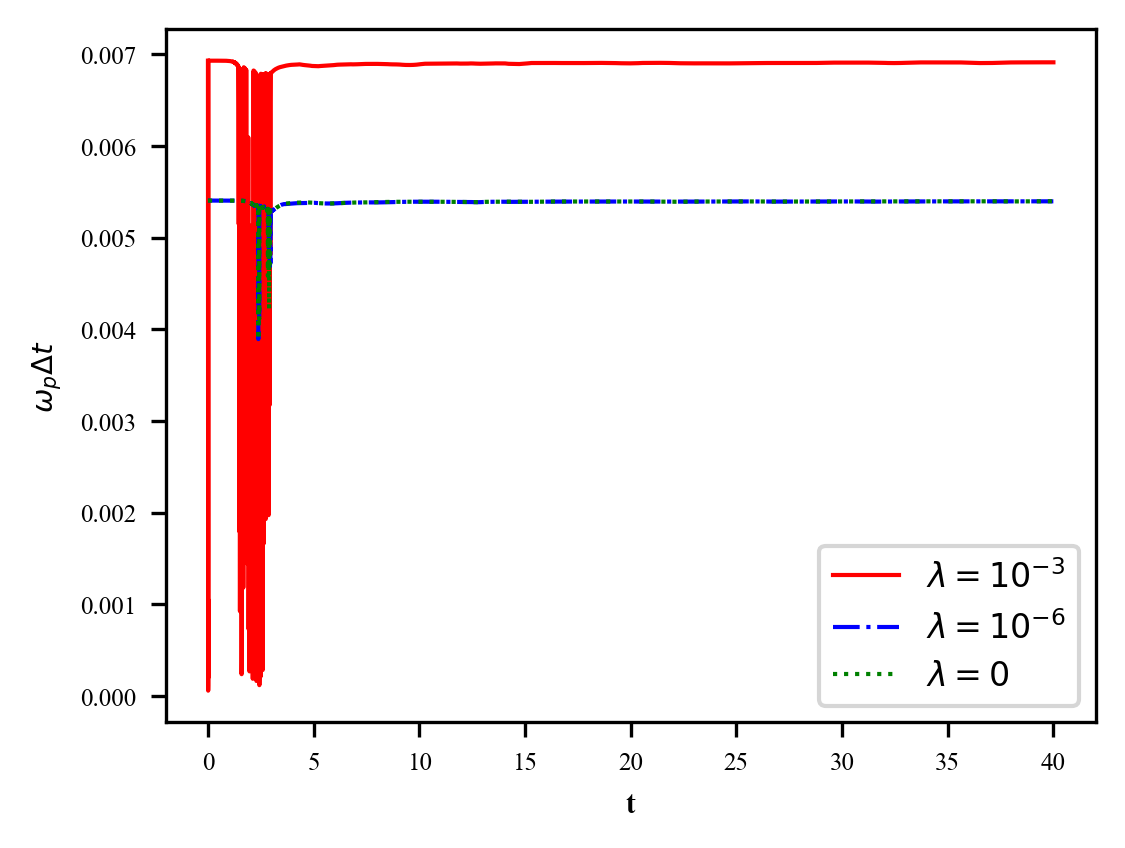}
\caption{}
\label{fig:lldk1}
\end{subfigure}
\caption{Bump on tail instability in the quasineutral limit by UGKS-RPE: $\lambda = 10^{-3}, 10^{-6}, 0, \Delta x = 0.08, \text{CFL} = 0.9$. Left: time evolution of electrostatic energy. Right: time evolution of timestep. The timestep is much larger than the Debye length. When $\lambda=10^{-6}$ and $0$, the timestep is smaller compared the $\lambda=10^{-3}$ because a larger velocity space is used.}
\label{fig:btiqnugks}
\end{figure}

\begin{figure}
\begin{subfigure}[b]{0.48\textwidth}
\centering
    \includegraphics[width=1.0\textwidth]{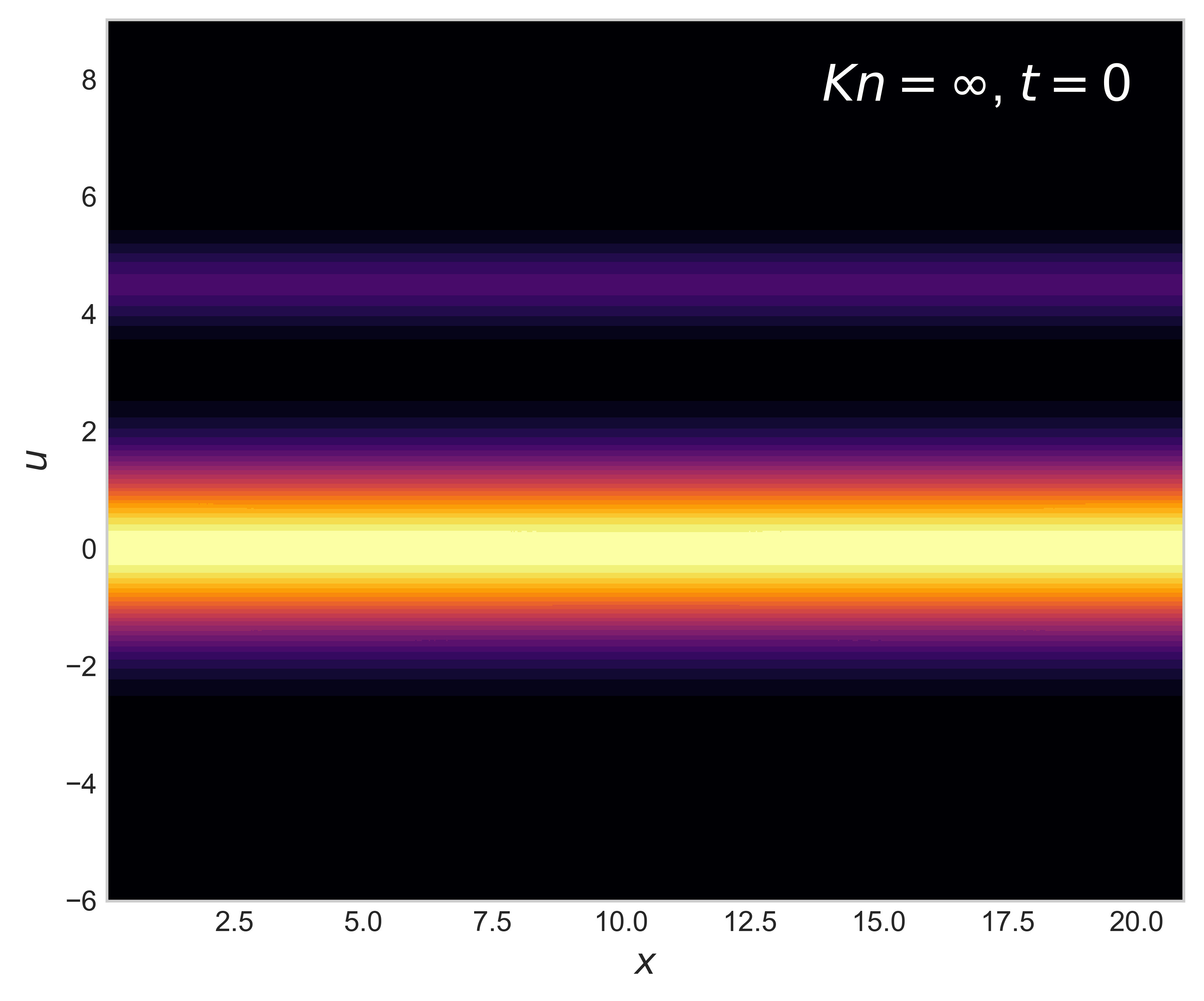}
    \caption{}
\end{subfigure}
\begin{subfigure}[b]{0.48\textwidth}
\centering
    \includegraphics[width=1.0\textwidth]{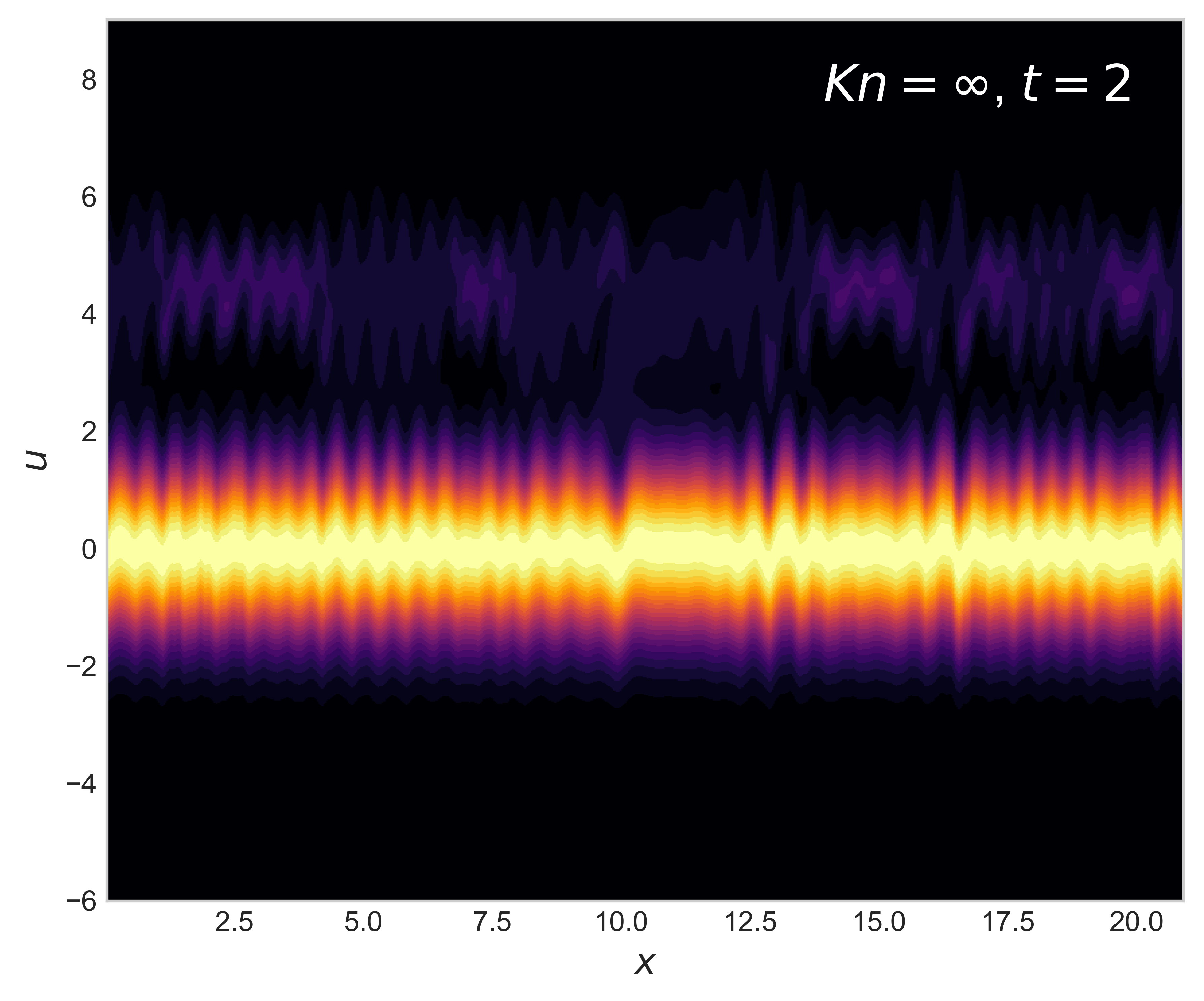}
\caption{}
\label{fig:lldk1}
\end{subfigure}
\vfill
\begin{subfigure}[b]{0.48\textwidth}
\centering
    \includegraphics[width=1.0\textwidth]{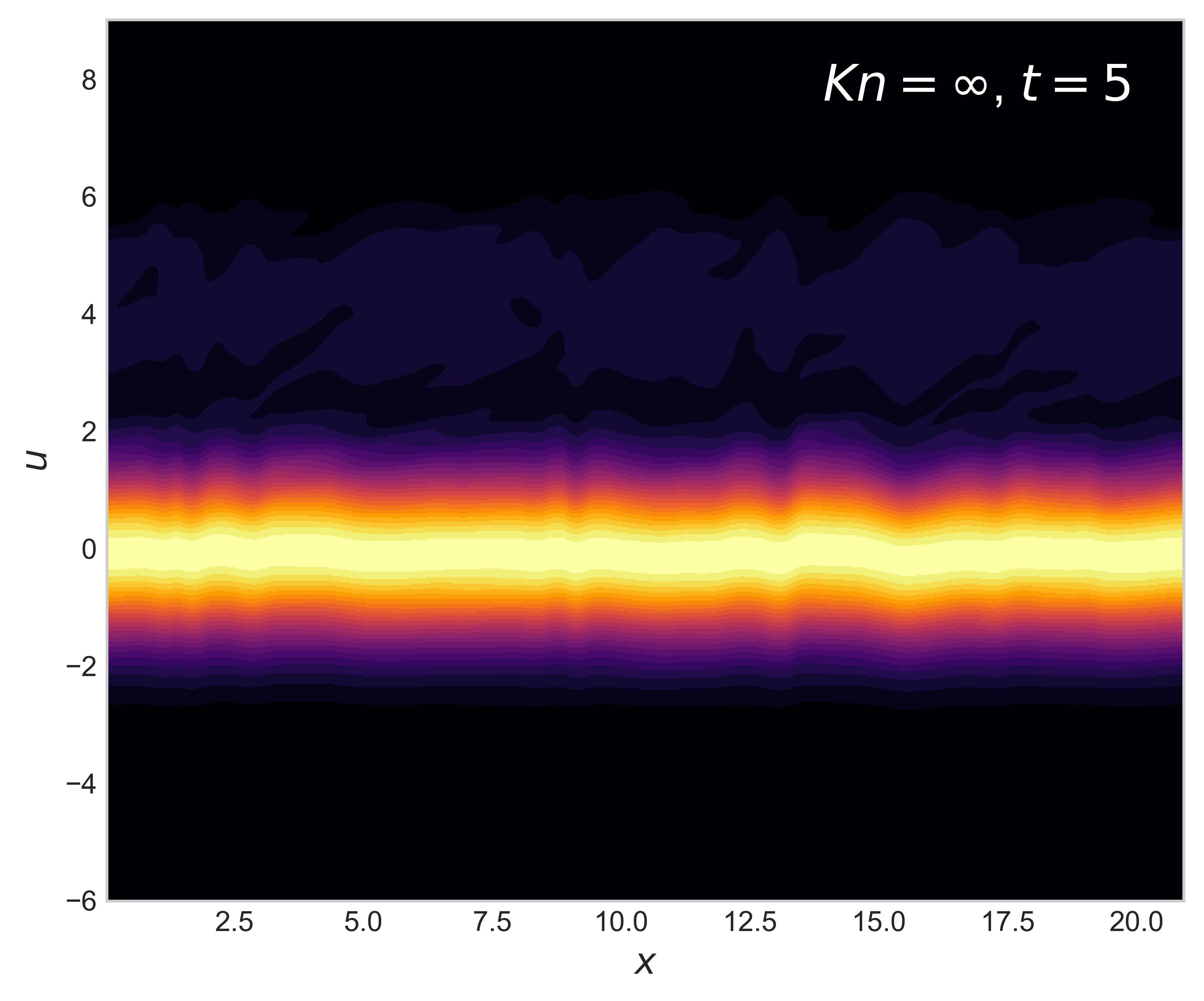}
    \caption{}
\end{subfigure}
\begin{subfigure}[b]{0.48\textwidth}
\centering
    \includegraphics[width=1.0\textwidth]{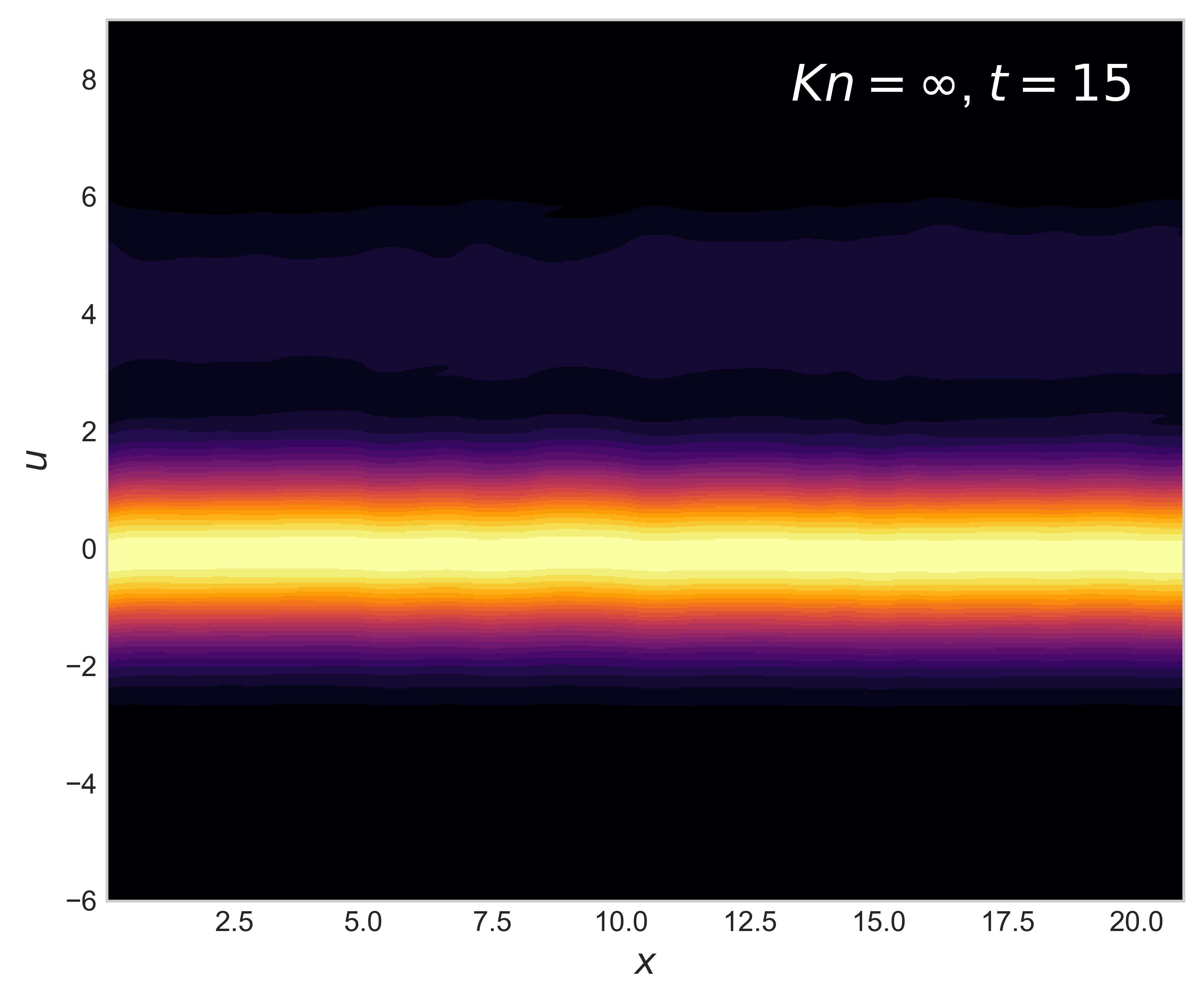}
    \caption{}
\end{subfigure}
\caption{Bump on tail instability, quasi-neutral case by UGKS-RPE: $\lambda = 10^{-3}, \Delta x = 0.08, \text{CFL} = 0.9$ phase diagram at (a) t=0 (b) t=2 (c) t=5 (d) t=15.}
\label{fig:bti-phaseqn}
\end{figure}

Figure \ref{fig:btikninfwp} presents the results of the bump-on-tail instability in the under-resolved regime, obtained using the UGKWP method. The Debye length is set to $\lambda=0.1$, and the mesh size is $\Delta x = 0.16 \geq \lambda$. The left panel illustrates the evolution of electrical energy for both UGKWP-PE and UGKWP-RPE. UGKWP-PE diverges in this under-resolved regime, whereas UGKWP-RPE remains stable. The electrical energy of UGKWP-RPE is on the order of $10^{-1}$, which aligns with the results observed in the UGKS-RPE simulations in Figure \ref{fig:btiqnugks}. In contrast to the deterministic nature of UGKS, the stochasticity of UGKWP introduces additional noise, resulting in a noisier electrical energy signal. The right panel shows that the UGKWP-RPE timestep exceeds $\omega_{pe}^{-1}$, demonstrating its ability to remain stable in the under-resolved regime. Although the timestep of UGKWP-PE is set to $\Delta t=\omega_{pe}^{-1}$, it still fails to maintain stability.

\begin{figure}
\begin{subfigure}[b]{0.48\textwidth}
\centering
    \includegraphics[width=1.0\textwidth]{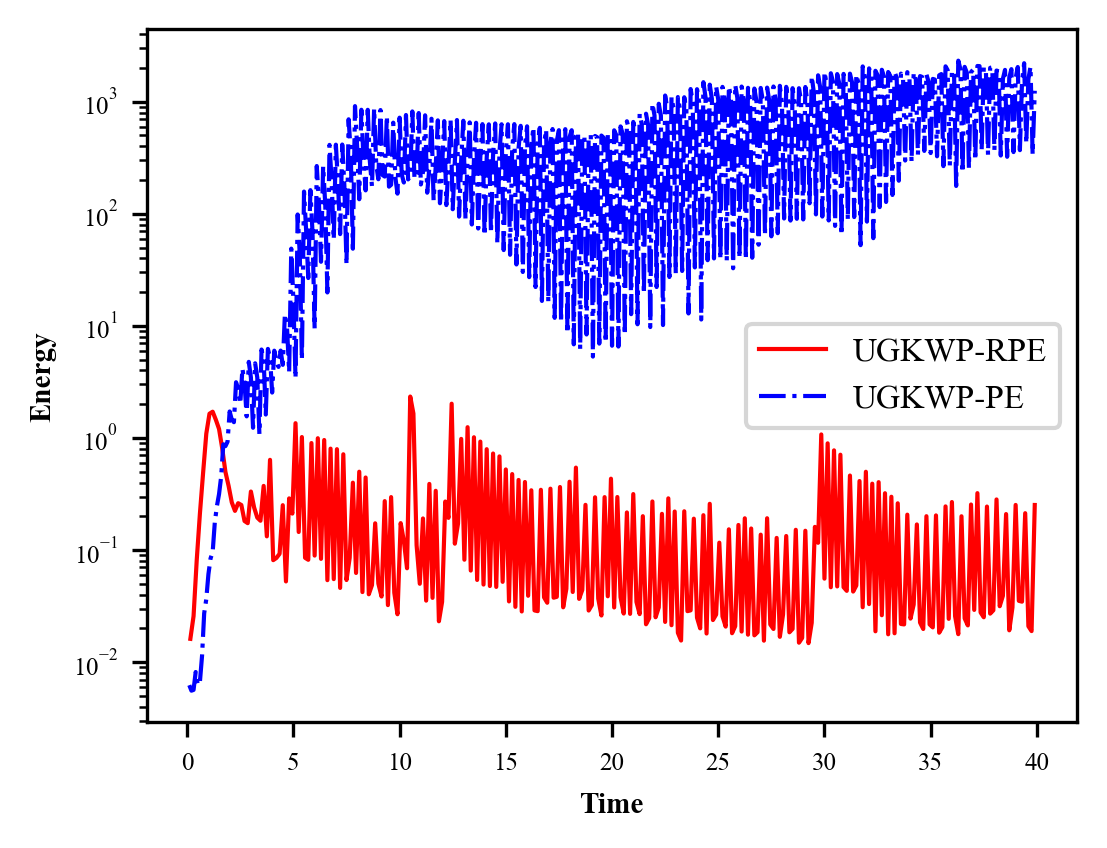}
\caption{}
\end{subfigure}
\begin{subfigure}[b]{0.48\textwidth}
\centering
    \includegraphics[width=1.0\textwidth]{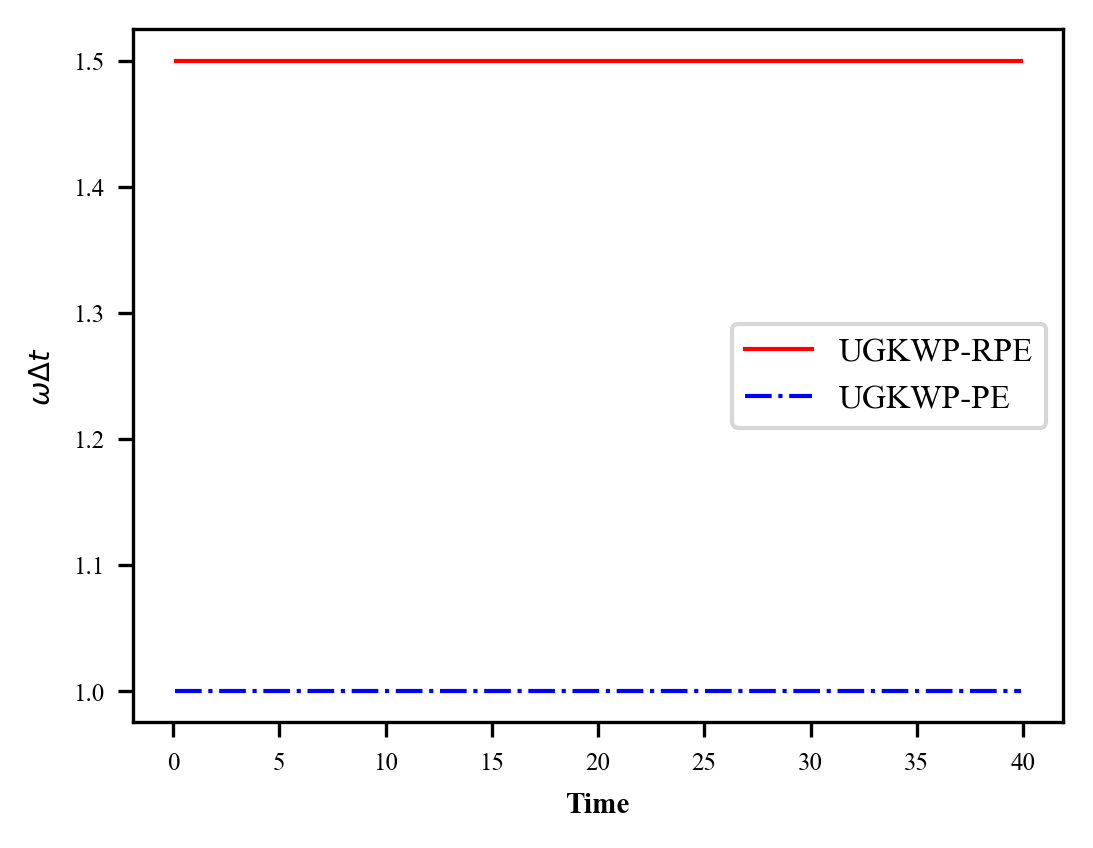}
\caption{}
\end{subfigure}
\caption{Bump on tail instability by UGKWP-RPE, under-resolved case: $\lambda = 0.1, \Delta x = 0.16 > \lambda$. Left: time evolution of electrostatic energy. Right: time evolution of timestep. }
\label{fig:btikninfwp}
\end{figure}

\subsection{Linear Landau damping}
\label{sec:lld}

This section examines the linear Landau damping case. To assess the numerical scheme's ability to resolve physics across the hydrodynamical and quasineutral limits simultaneously, we systematically vary the Knudsen number and the Debye length. The initial electron distribution function is the same as that employed in the nonlinear Landau damping study detailed in section \ref{sec:nld}, with the specific parameters for this investigation being $k = 1.0$ and perturbation amplitude $\alpha=0.0001$. The spatial domain is discretized into 128 cells.

Figure \ref{fig:lldkninfugksqn} presents the results obtained using the UGKS in the quasineutral and collisionless regime. Simulations were conducted with Debye lengths ranging from $\lambda = 10^{-3}$ to $\lambda = 10^{-4}$, and also with $\lambda = 0$. The electrical energy is calculated as $E_p = \frac{1}{2} \int\left|\nabla\phi\right|^2 d x$. The left panel of Figure \ref{fig:lldkninfugksqn} demonstrates that UGKS accurately captures the quasineutral behavior in the quasineutral regime. Specifically, the observed damping rate closely matches the theoretical prediction of -1.73 reported by Crouseilles et al. \cite{crouseilles2016multiscale}. The right panel of Figure \ref{fig:lldkninfugksqn} illustrates the temporal evolution of the timestep, revealing that the timestep remains independent of the Debye length, even at the smallest Debye lengths investigated. Figure \ref{fig:lldkninfugksqn_phase} depicts the phase space diagram at times $t = 1$ and $t = 10$. The phase space distribution remains consistently close to the equilibrium distribution. These observations indicate that, in the quasineutral regime, the plasma exhibits fluid-like behavior.
\begin{figure}
    \centering
    \begin{subfigure}[b]{0.48\textwidth}
    \centering
    \includegraphics[width=0.95\linewidth]{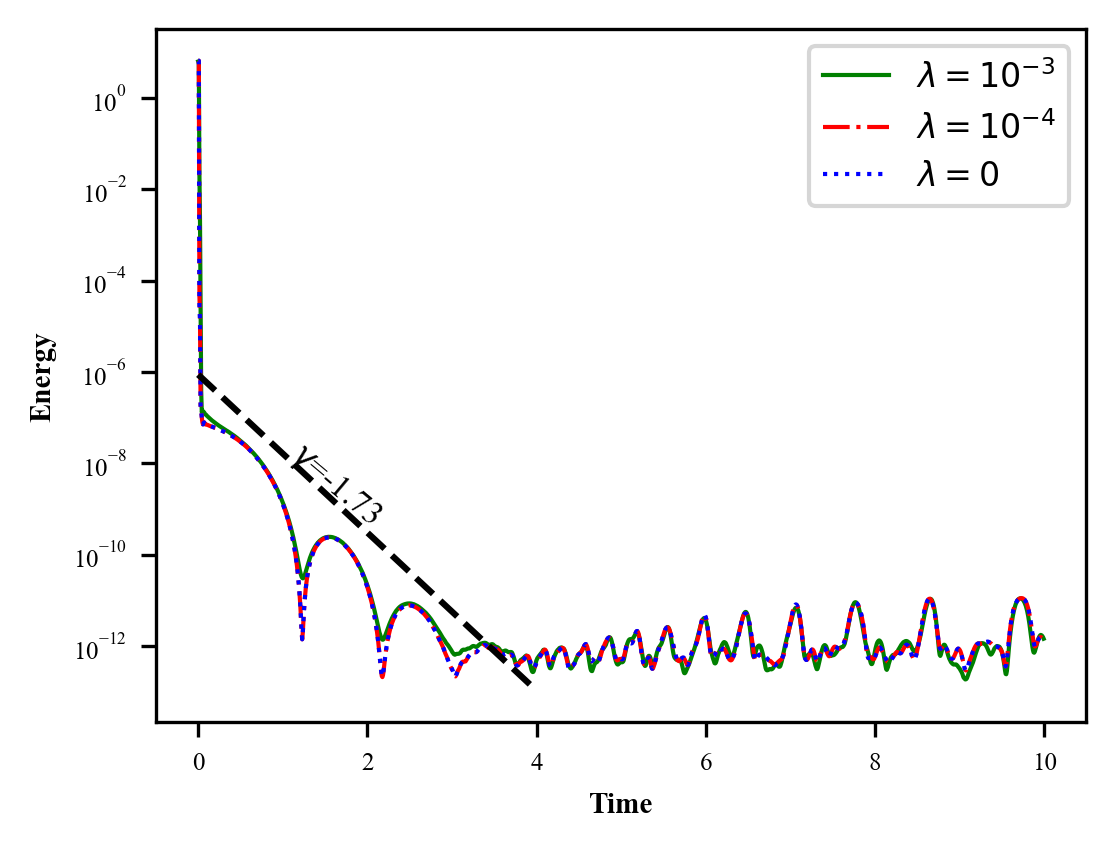}
    \end{subfigure}
    \hfill
    \begin{subfigure}[b]{0.48\textwidth}
    \centering
    \includegraphics[width=0.95\linewidth]{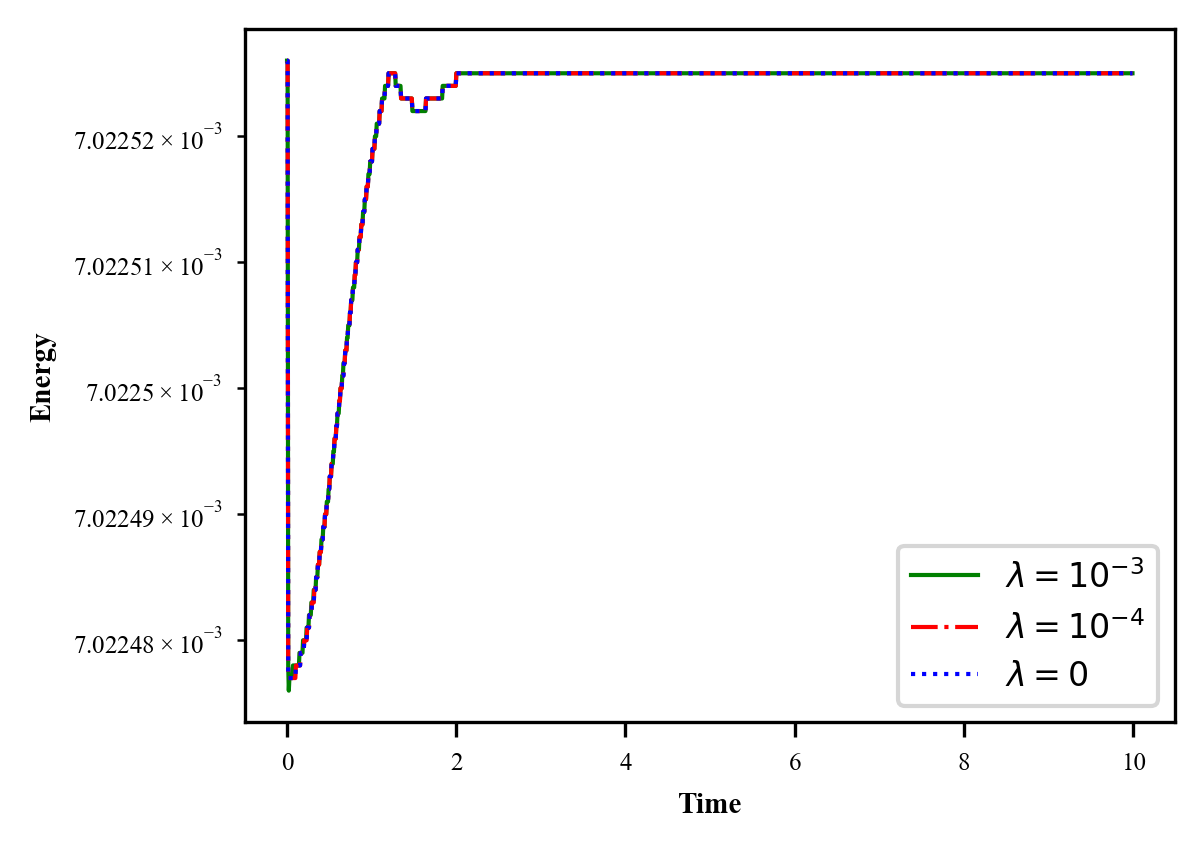}
    \end{subfigure}
    \caption{Linear Landau damping, the quasineutral limit case by UGKS-RPE at $\lambda=10^{-3}, \lambda=10^{-4}, \lambda = 0, Kn=\infty$, $\Delta x=0.05$ and CFL=0.9. Left: time evolution of the electric field. Right: time evolution of timestep.}
    \label{fig:lldkninfugksqn}
\end{figure}

\begin{figure}
    \centering
    \begin{subfigure}[b]{0.48\textwidth}
    \centering
    \includegraphics[width=0.95\linewidth]{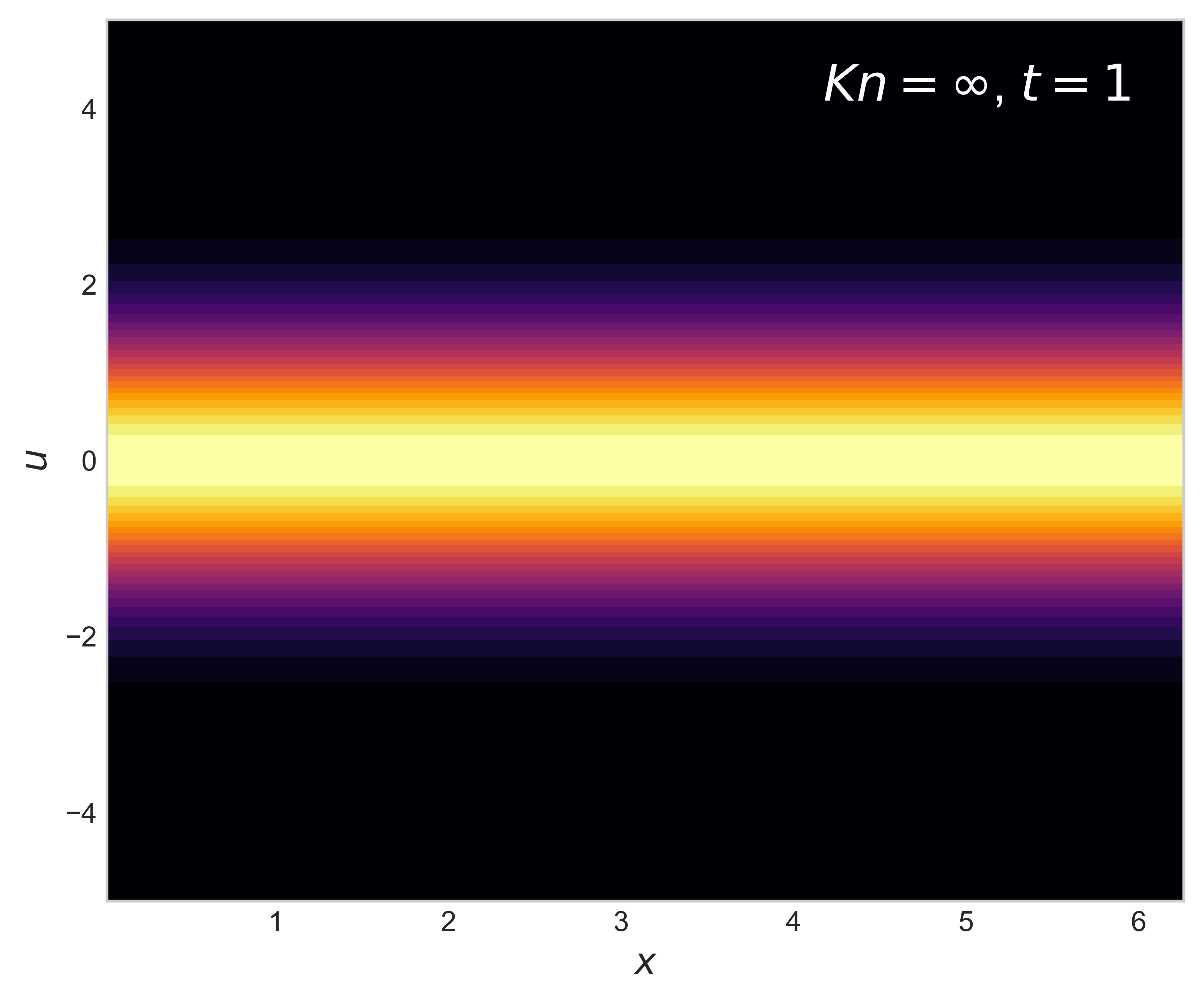}
    \end{subfigure}
    \hfill
    \begin{subfigure}[b]{0.48\textwidth}
    \centering
    \includegraphics[width=0.95\linewidth]{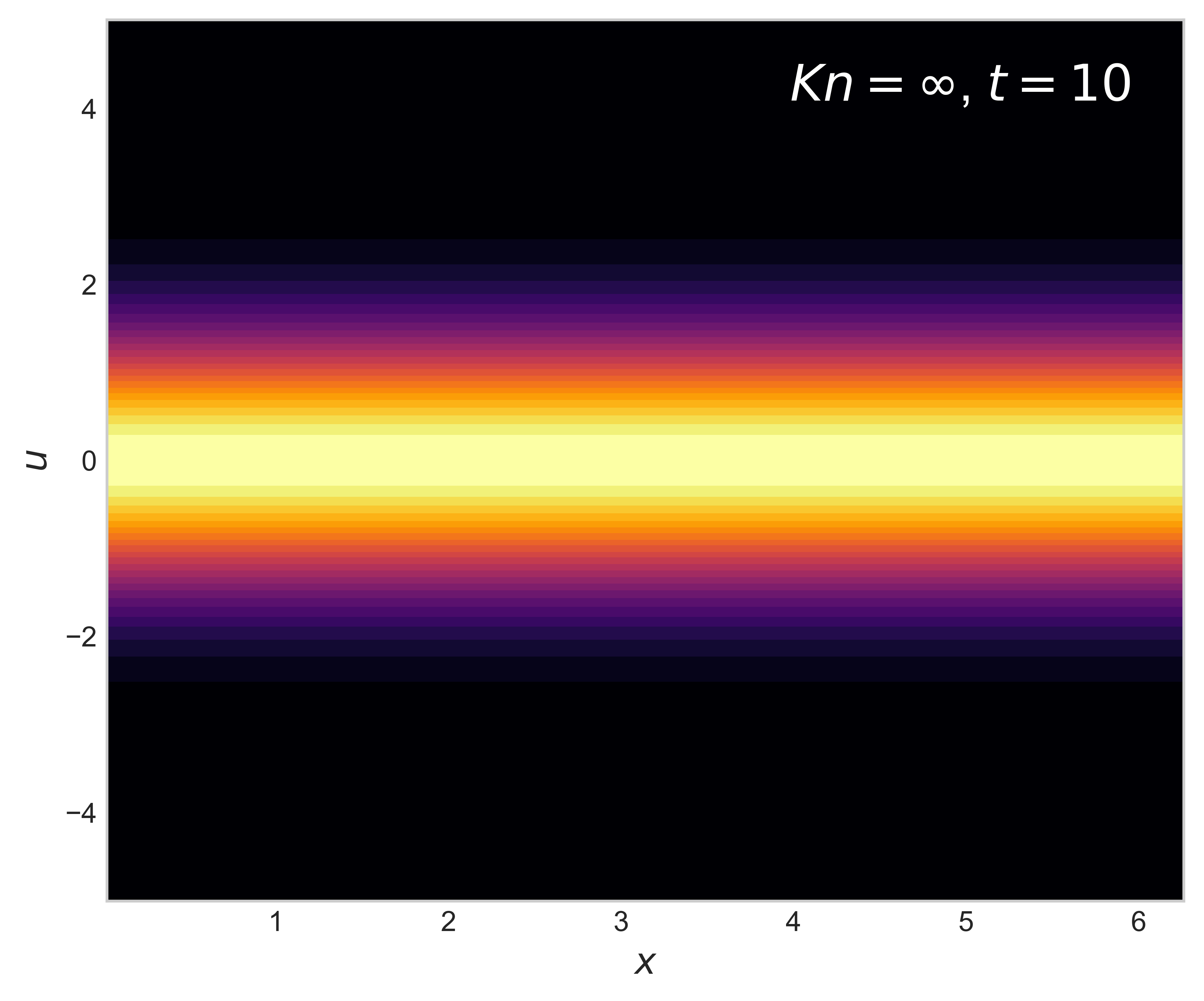}
    \end{subfigure}
    \caption{Linear Landau damping by UGKS-RPE, the quasi-neutral case $\lambda=10^{-4}, Kn=\infty$, $\Delta x=0.05$ and CFL=0.9. The Figure shows the phase diagram at t=1 and t=10.}
    \label{fig:lldkninfugksqn_phase}
\end{figure}

Figure \ref{fig:lld-us} displays the results obtained using the UGKWP-RPE method for parameters $Kn=\infty$, $\lambda = 0.001$, and $\Delta x = 0.05$. These conditions result in grid spacing that is significantly larger than the Debye length. The time step was dictated by a CFL condition of 0.5, yielding $\Delta t \approx 0.025$. This corresponds to approximately 25 times the plasma oscillation period, $\omega_{pe}^{-1} = 0.001$. The left panel illustrates the temporal evolution of the electric field energy. This energy profile exhibits an initial rapid decay, followed by a prolonged period of stability. Such behavior confirms the algorithm's ability to suppress high-frequency oscillations even when the corresponding spatial and temporal scales are not explicitly resolved. This observation is in agreement with prior findings in the literature \cite{degond2010asymptotic}. The right panel presents the electron number density distribution at $T = 40$. Throughout the domain, the electron density remains close to $n = 1$, indicating that the system maintains quasi-neutrality, consistent with the ion number density also being unity.

\begin{figure}
\begin{subfigure}[b]{0.48\textwidth}
\centering
    \includegraphics[width=1.0\textwidth]{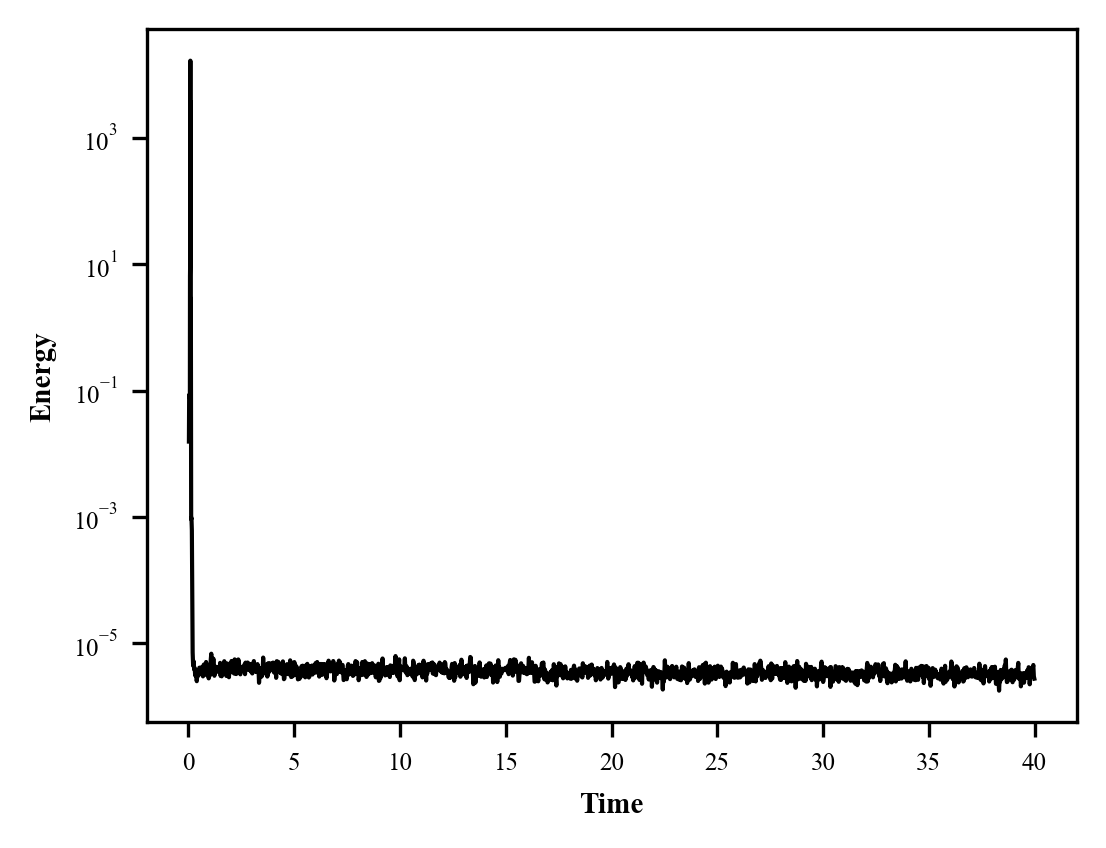}
\label{fig:lld-us-e}
\end{subfigure}
\begin{subfigure}[b]{0.48\textwidth}
\centering
    \includegraphics[width=1.0\textwidth]{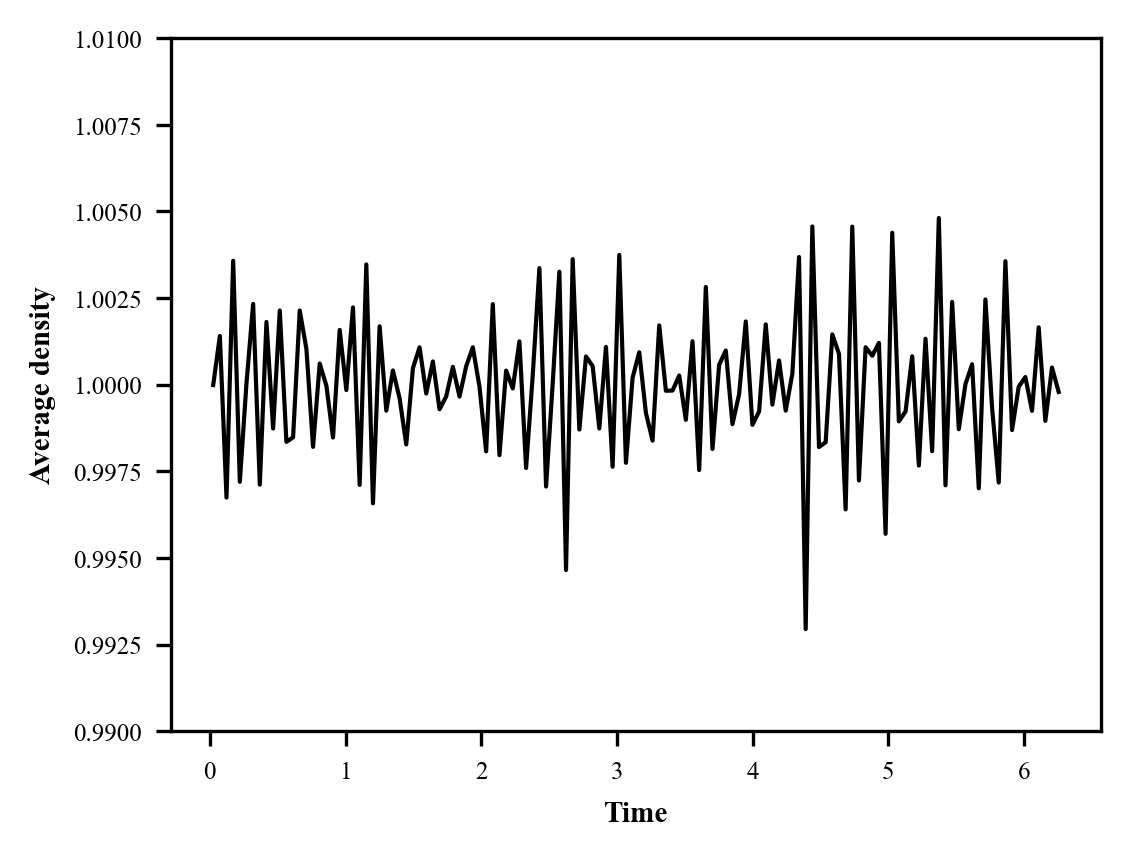}
\label{fig:lld-us-n}
\end{subfigure}
\caption{Linear Landau damping case by UGKWP-RPE of under-resolved cases $k=1.0, \lambda=0.001, N=128, \Delta x=0.05\gg\lambda, \omega_{pe}\Delta t \approx 25 > 1$. Left: Evolution of electrical energy, right: Averaged number density.}
\label{fig:lld-us}
\end{figure}

Figure \ref{fig:ugksqnfluid} presents the results obtained using the UGKS-RPE in the quasineutral regime ($\lambda = 10^{-6}$) across a range of rarefaction regimes, characterized by Knudsen numbers of $Kn = \infty, 10, 1, 10^{-4}$. These simulations employed a spatial discretization of $\Delta x = 0.05$ and a CFL number of 0.9. The left panel of Figure \ref{fig:ugksqnfluid} illustrates the temporal evolution of the electrical energy. The asymptotic solution of the quasineutral regime is clearly resolved. As the Knudsen number decreases, indicating a transition to a fluid-like state, the damping rate decreases and ultimately vanishes in the small-Knudsen-number limit ($Kn = 10^{-4}$). This behavior is consistent with the expected physical picture: the kinetic wave-particle interactions weaken and eventually disappear as the plasma transitions from a particle-dominated to a fluid-dominated state. The right panel of Figure \ref{fig:ugksqnfluid} depicts the temporal evolution of the timestep. The timestep exhibits minor oscillations around a value of $4 \times 10^{-3}$, but critically, it remains independent of both the Debye length and the Knudsen number.

\begin{figure}
    \centering
    \begin{subfigure}[b]{0.48\textwidth}
    \centering
    \includegraphics[width=0.95\linewidth]{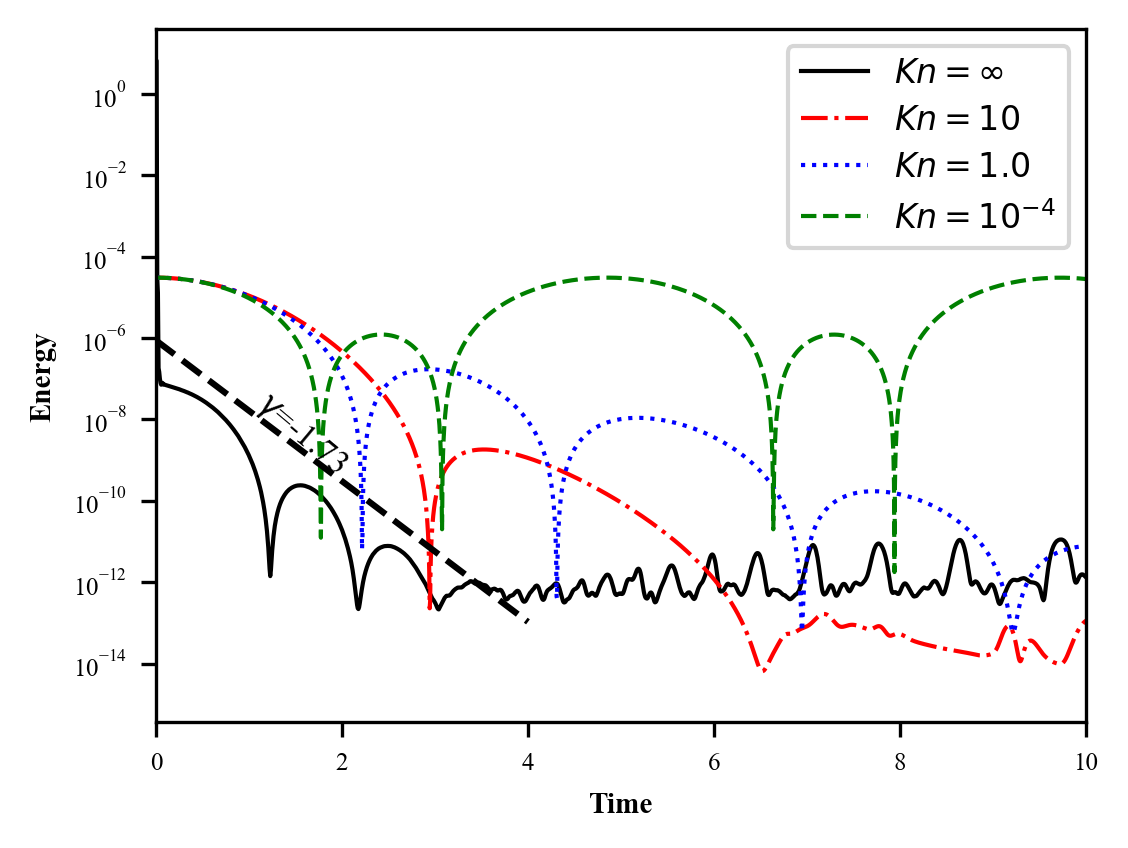}
    \end{subfigure}
    \hfill
    \begin{subfigure}[b]{0.48\textwidth}
    \centering
    \includegraphics[width=0.95\linewidth]{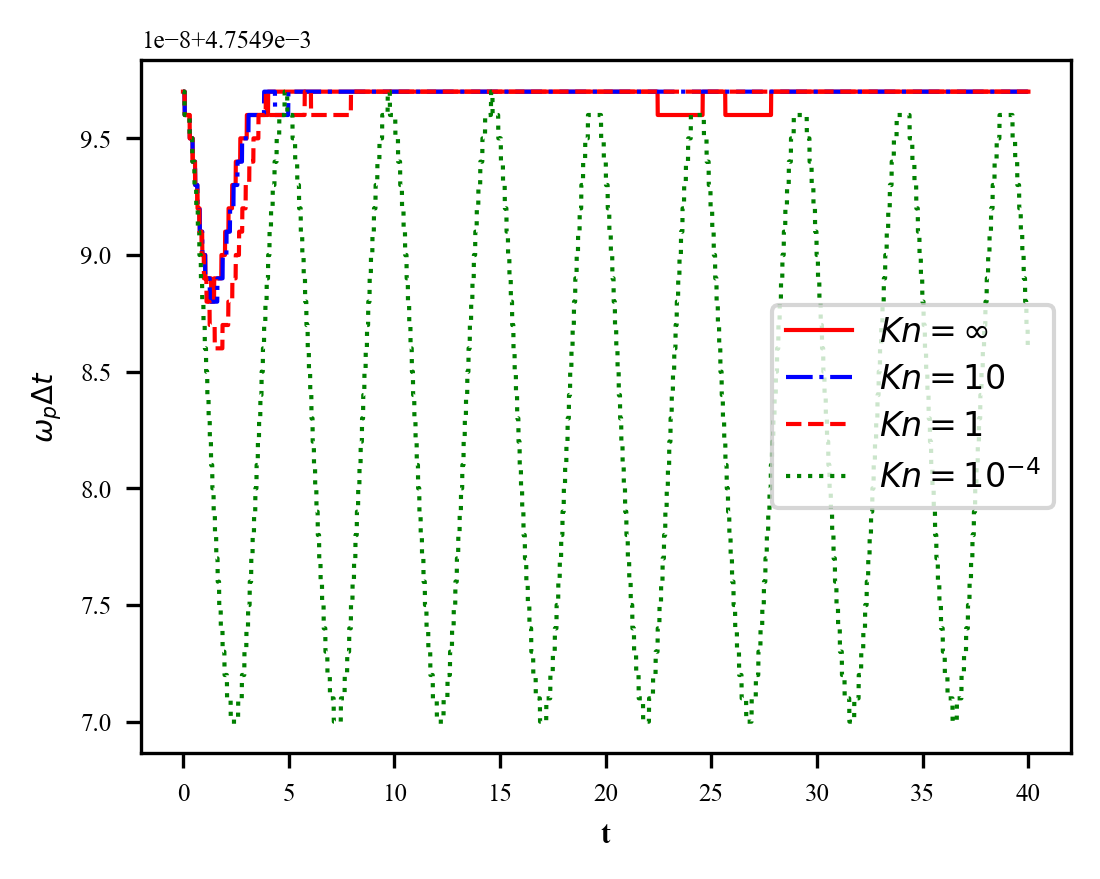}
    \end{subfigure}
    \caption{Linear Landau damping by UGKS-RPE, the quasineutral and hydrodynamic limit case $\lambda=10^{-6}, Kn=\infty, 10, 1$  and $10^{-4}$, $\Delta x=0.05$ and CFL=0.9. Left: time evolution of the electric field. Right: time evolution of timestep.}
    \label{fig:ugksqnfluid}
\end{figure}

Figure \ref{fig:lldkncompare} presents the results obtained using the UGKWP-RPE method in the quasineutral regime ($\lambda = 0.0001$) for varying Knudsen numbers $Kn=0.1, 0.01, 0.001, 0.0001$. The left panel illustrates the temporal evolution of the electrical energy. The simulations remain stable across all Knudsen numbers investigated. The electrical energy increases with increasing Knudsen number. This trend is attributed to the enhanced sampling noise arising from the increased number of particles at higher Knudsen numbers, which consequently elevates the overall electrical energy. The right panel of Figure \ref{fig:lldkncompare} depicts the temporal evolution of the timestep. The timestep is demonstrably independent of both the Debye length and the Knudsen number.
Furthermore, the timestep increases as the Knudsen number decreases. This behavior is attributed to the increasing proportion of the system behaving as a fluid. As the fluid component dominates, the influence of fast electrons on the timestep diminishes, allowing for a larger overall timestep at a fixed CFL number.

\begin{figure}
    \centering
    \begin{subfigure}[b]{0.48\textwidth}
        \includegraphics[width=1.0\linewidth]{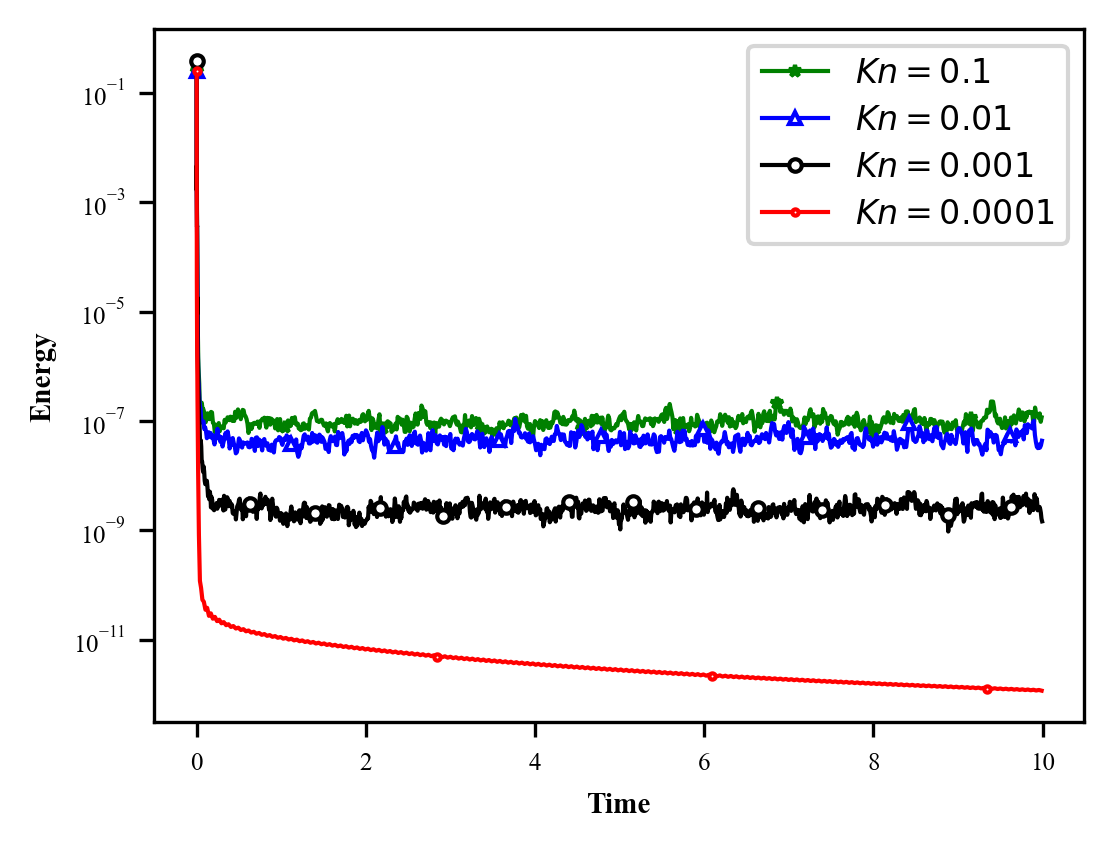}
    \end{subfigure}
    \begin{subfigure}[b]{0.48\textwidth}
        \includegraphics[width=1.0\textwidth]{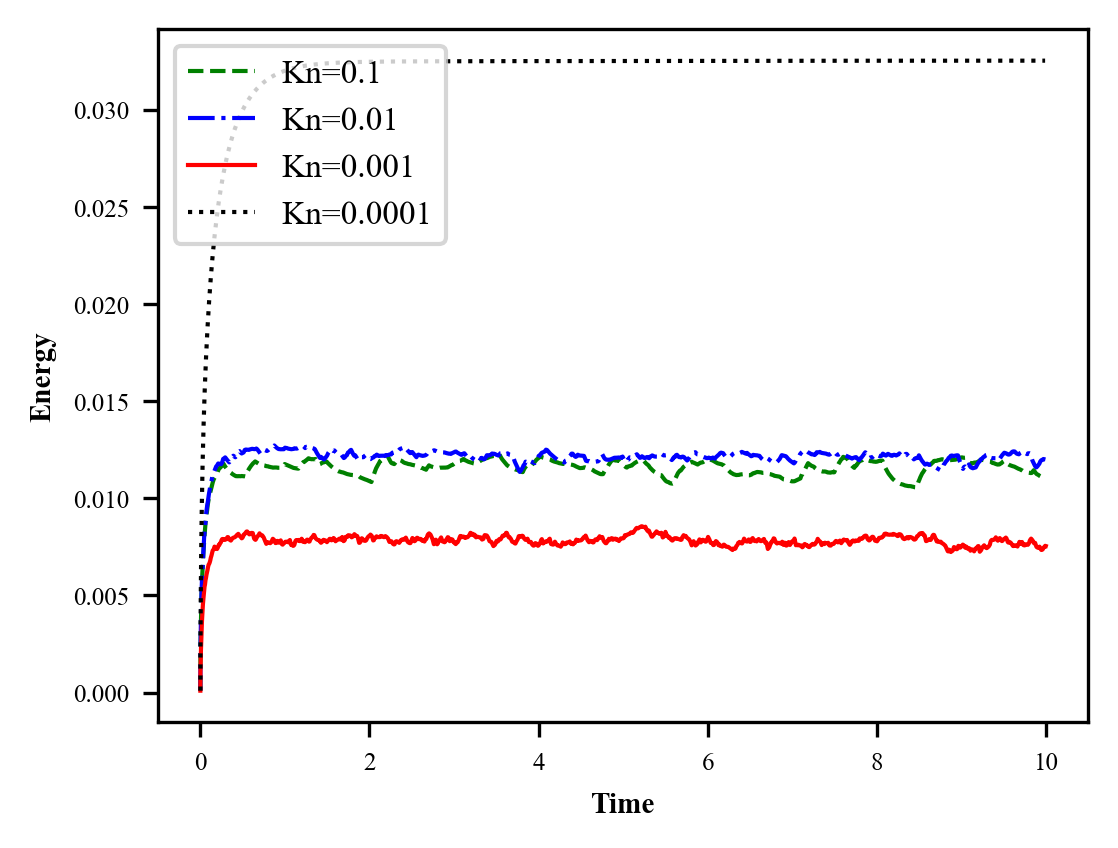}
    \end{subfigure}
    \caption{Linear Landau damping by UGKWP-RPE of under-resolved cases of varying Knudsen number $Kn = 0.1, 0.01, 0.001, 0.0001$ and $\lambda = 0.0001, k=1.0, N=128, \Delta x=0.05\gg\lambda, \omega_{pe}\Delta t > 1, \text{CFL}=0.5$. Stability is maintained in all cases. Left: evolution of electrical energy, right: evolution of timestep.}
    \label{fig:lldkncompare}
\end{figure}

\section{Conclusion}
\label{sec:conclusion}

In this study, we extend the UGKS and the UGKWP methods to model multiscale electrostatic plasmas by integrating them with the Reformulated Poisson Equation. The coupled collision–transport numerical flux inherent in these schemes removes the limitations associated with the mean free path and mean collision time. Meanwhile, the RPE formulation, through its implicit coupling with the macroscopic moment equations, permits the use of larger time steps even in regimes with vanishingly small Debye lengths.
Benchmark tests demonstrate that the proposed UGKS–RPE and UGKWP–RPE methods accurately simulate electrostatic plasma dynamics without constraints imposed by either the mean free path or the Debye length. This work establishes a solid foundation for the future application of the UGKWP–RPE framework to the efficient modeling of complex hypersonic plasma flows.

\section*{Acknowledgements}
The current research is supported by National Key R\&D Program of China (Grant Nos. 2022YFA1004500), National Science Foundation of China (12172316, 92371107), and Hong Kong research grant council (16301222, 16208324).

\bibliographystyle{elsarticle-num}
\bibliography{ref}

\end{document}